\documentclass[ba,preprint]{imsart}%
\firstpage{1}
\lastpage{1}

\pdfoutput=1

\usepackage{natbib}
\usepackage[colorlinks,citecolor=blue,urlcolor=blue,filecolor=blue,backref=page]{hyperref}
\usepackage{url}            %
\usepackage{booktabs}       %
\usepackage{amsfonts}       %
\usepackage{amsthm}
\usepackage{amssymb}        %
\usepackage{amsmath}
\usepackage{bm}
\usepackage{mathtools}
\usepackage{nicefrac}       %
\usepackage{microtype}      %
\usepackage{tikz}           %
\usepackage[T1]{fontenc}

\usepackage[normalem]{ulem}
\usepackage[ruled,vlined,linesnumbered]{algorithm2e}
\usepackage{graphicx}
\usepackage{subcaption}

\theoremstyle{definition}
\newtheorem{definition}{Definition}
\newtheorem{example}{Example}

\theoremstyle{remark}
\newtheorem{remark}{Remark}

\theoremstyle{plain}
\newtheorem{theorem}{Theorem}
\newtheorem{proposition}{Proposition}

\newcommand\ci{\perp\!\!\!\perp}

\DeclareMathOperator*{\argmax}{argmax}

\newcommand{\X}{X}
\newcommand{\Y}{Y}
\newcommand{\Z}{Z}
\newcommand{\x}{x}
\newcommand{\y}{y}
\newcommand{\z}{z}
\newcommand{\xs}{\text{X}}
\newcommand{\ys}{\text{Y}}

\renewcommand{\d}{\mathrm{d}}

\newcommand{\Prior}{\pi_{\X}}
\newcommand{\Like}{\pi_{\Y|\X}}

\newcommand{\Post}{\pi_{\X|\Y}}
\newcommand{\Joint}{\pi_{\X,\Y}}
\newcommand{\ApproxPost}{\widetilde{\pi}_{\X|\Y}} 
\newcommand{\OptPost}{\pi^*_{\X|\Y}}
\newcommand{\DataMarg}{\pi_{\Y}}
\newcommand{\ApproxDataMarg}{\pi^*_{\Y}}

\newcommand{\Cobs}{\Gamma_{\mathrm{obs}}}
\newcommand{\Cpr}{\Gamma_{\mathrm{pr}}}
\newcommand{\BasisX}{U}
\newcommand{\BasisY}{V}

\newcommand{\KLDiv}{D_{\textrm{KL}}}
\newcommand{\R}{\mathbb{R}}
\newcommand{\E}{\mathbb{E}}
\DeclareMathOperator{\Tr}{Trace}
\newcommand{\Id}{\textrm{I}}
\newcommand{\CCA}{\textrm{CCA}}
\newcommand{\PCA}{\textrm{PCA}}
\newcommand{\Cov}{\mathbb{C}\mathrm{ov}}

\begin{document}

\begin{frontmatter}
\title{Gradient-based data and parameter dimension reduction for Bayesian models: an information theoretic perspective}
\runtitle{Information theoretic data and parameter dimension reduction}

\begin{aug}
\author{\fnms{Ricardo} \snm{Baptista}\thanksref{addr1}\ead[label=e1]{rsb@mit.edu}},
\author{\fnms{Youssef} \snm{Marzouk}\thanksref{addr1}\ead[label=e2]{ymarz@mit.edu}}
\and
\author{\fnms{Olivier} \snm{Zahm}\thanksref{addr2}\ead[label=e3]{olivier.zahm@inria.fr}}

\runauthor{R. Baptista et al.}

\address[addr1]{Massachusetts Institute of Technology\\
       Cambridge, MA 02139-4301, USA
    \printead{e1},
    \printead*{e2}
}

\address[addr2]{Universit\'e Grenoble Alpes, Inria, CNRS, Grenoble INP, LJK \\
    38000 Grenoble, France
    \printead{e3}
}

\end{aug}

\begin{abstract}
We consider the problem of reducing the dimensions of parameters and data in non-Gaussian Bayesian inference problems. Our goal is to identify an ``informed'' subspace of the parameters and an ``informative'' subspace of the data so that a high-dimensional inference problem can be approximately reformulated in low-to-moderate dimensions, thereby 
improving the computational efficiency of many inference techniques.
To do so, we exploit gradient evaluations of the log-likelihood function.
Furthermore, we use an information-theoretic analysis to derive a bound on the posterior error due to parameter and data dimension reduction. This bound relies on logarithmic Sobolev inequalities, and it reveals the appropriate dimensions of the reduced variables. 
We compare our method with classical dimension reduction techniques, such as principal component analysis and canonical correlation analysis, on applications ranging from mechanics to image processing. 
\end{abstract}



\end{frontmatter}

\section{Introduction} \label{sec:introduction}
Many statistical problems throughout science and engineering involve inferring parameters $\X$ from observations $\Y$, where both $\X$ and $\Y$ are high-dimensional vectors. These vectors often arise from discretizations of infinite-dimensional signals, as for example in full waveform inversion or medical imaging. In the Bayesian setting, $\X$ and $\Y$ are modeled as random variables, and the goal of computation in Bayesian inference is generally to characterize the posterior distribution, whose density\footnote{Throughout this paper, we consider probability distributions that are absolutely continuous, and thus have densities, with respect to the Lebesgue measure. For simplicity, we will thus use similar notation for distributions and densities unless otherwise indicated.}
is given by Bayes rule as
\begin{equation*}
\Post(\x|\y) \propto \Like(\y|\x) \Prior(\x).
\end{equation*}
Here, $\Prior$ is the prior density of $\X$, $\Like$ is the conditional density of the data $\Y$, and $\x \mapsto \Like(\y \vert \x)$ is the likelihood function for any realized value of the data $\y$. While many sampling-based algorithms have been developed for Bayesian inference, their computational costs typically scale poorly with increasing dimensions of $\X$ and $\Y$~\citep{roberts2001optimal,agapiou2017importance, mangoubi2019nonconvex, chen2020fast}, especially for distributions that do not satisfy strong log-concavity assumptions. Similarly,  %
the costs of variational Bayesian methods~\citep{rezende2015variational, blei2017variational, detommaso2018stein} can scale poorly with dimension, particularly when accurate characterizations of posterior structure are desired.

Dimension reduction has received increasing attention as a way of reducing the computational cost of inference procedures. 
On the one hand, many recently proposed methods seek to reduce the dimension of the parameter $\X$. For instance, in~\citet{cui2014likelihood, zahm2018certified,  constantine2016accelerating, chen2020projected}, the gradient of the log-likelihood function is used---through a variety of different constructions---to identify a so-called \emph{likelihood informed subspace}; this subspace captures parameter directions where the data are most informative relative to the prior.
Projecting $\X$ onto this lower-dimensional subspace can realize
immense computational savings when applying MCMC to complex 
high-dimensional posterior distributions~\citep{cui2016dimension,izmailov2020subspace}. Similar projections have proven useful in variational inference; for instance, \citet{brennan2020greedy} uses these projections to focus the expressiveness of transport maps or normalizing flows on the informed subspace, yielding better posterior approximations.

On the other hand, reducing the dimension of $\Y$ is critical to performing 
inference in settings with high-dimensional data, such as spectra and time-series signals~\citep{ma2013magnetic}. For instance, approximate Bayesian computation (ABC) generates posterior samples by comparing simulated and observed data~\citep{sisson2018handbook}, which becomes increasingly difficult in high dimensions. Recent methods based on conditional density estimation~\citep{bishop1994mixture, papamakarios2016fast} and transportation of measure~\citep{radev2020bayesflow, spantini2019coupling, cui2021conditional} 
simulate from conditional densities by learning functions of both the parameters and the data. Reducing the data dimension in this setting 
can yield substantial computational savings.

Previous work in ABC reduces the data dimension by seeking low-dimensional summary statistics designed to retain information about the parameters~\citep{fearnhead2012constructing}. On the other hand, for conjugate linear--Gaussian models, \citet{giraldi2018optimal} find maximally informative 
\emph{subspaces} of the data, of any given dimension, by solving an eigenvalue problem depending on the likelihood and on the prior covariance.
\citet{trippe2019lr} seek low-dimensional projections of the data for generalized linear models, and these projections are endowed with error guarantees under certain conditions (e.g., strongly log-concave posteriors).
Optimal experimental design can also be seen as a way of reducing the data dimension, by sub-selecting the most important \emph{components} of the random vector $\Y$ \citep{krause2008near,wu2020fast,jagalur2021batch}. It is important to note that all of these dimension reduction methods are applied before the data are realized, and hence do not depend on the observed value of $\Y$. These data summaries or subspaces can thus be re-used for multiple instances of $\Y$. Such approaches differ fundamentally from, e.g., Bayesian coresets~\citep{campbell2019automated}, which summarize a given \emph{realization} of $\Y$ via a smaller weighted subset of the data (assuming, moreover, that elements of $\Y$ are conditionally independent given $X$).

In this work, our goal is to concurrently reduce the dimensions of the parameter and of the data using gradients of the log-likelihood $\nabla_\X \nabla_\Y \log\Like(y|x)\in\R^{m\times d}$, where $d$ and $m$ are the dimensions of $\X$ and $\Y$, respectively. To do so, we compute the eigenvectors of the diagnostic matrices
\begin{align}
    H_\X &= \int \big(\nabla_\X \nabla_\Y \log\Like\big)^T\big(\nabla_\X \nabla_\Y \log\Like\big) \d\Joint  \label{eq:HXfirst} \\
    H_\Y &= \int \big(\nabla_\X \nabla_\Y \log\Like\big) \big(\nabla_\X \nabla_\Y \log\Like\big)^T \d\Joint, \label{eq:HYfirst}
\end{align}
and define the informed parameters $\X_r$ and the informative data $\Y_s$ as projections of $\X$ and $\Y$ onto the first $r\ll d$ and $s\ll m$ components of their respective eigenbases. 
Our definitions of $H_\X$ and $H_\Y$ follow from an information theoretic analysis of the approximation error due to simultaneous dimension reduction of $\X$ and $\Y$.  
Specifically, we bound the expected Kullback--Leibler (KL) divergence from the approximate posterior to the exact posterior as follows:
\begin{equation*}
    \mathbb{E}_\Y \left[\KLDiv(\Post||\OptPost) \right] \leq \overline{C}(\Joint)^2 \left(\sum_{i > r} \lambda_i(H_\X) + \sum_{j > s} \lambda_j(H_\Y) \right),
\end{equation*}
where $\OptPost(x|y) \propto \pi_{\Y_s|\X_r}(\y_s|\x_r)\Prior(\x)$ is the posterior approximation defined using the marginal conditional distribution $\pi_{\Y_s|\X_r}$ of the reduced parameter and reduced data.
In the expression above, $\lambda_i(\cdot)$ denotes the $i$th largest eigenvalue of a matrix and $\overline{C}(\cdot)$ 
is the subspace logarithmic Sobolev constant of a probability density. This constant will be defined and discussed extensively later in the paper.
The derivation of this error bound relies on 
a result in~\citet{baptista2021learning} establishing gradient-based bounds on conditional mutual information.
Fast decay of the eigenvalues $\lambda_i(H_\X)$ and $\lambda_j(H_\Y)$ above ensures small error for low dimensions $r,s$.

For Gaussian error models of the form $\Like(y|x)\propto\exp(-\frac12\|G(x)-y\|^2_2)$, the diagnostic matrices simplify to
$$
 H_\X =\int \nabla G^T\nabla G \d\pi_{\X} 
 \qquad\text{and}\qquad
 H_\Y =\int \nabla G \nabla G^T \d\pi_{\X} ,
$$
where $\nabla G(x)\in\R^{m\times d}$ is the Jacobian of the (nonlinear) forward model $G\colon\R^d\rightarrow\R^m$. 
In this particular case, $H_\X$ is the same diagnostic matrix as one introduced in earlier work \citep{cui2020data,cui2021conditional}. The general form of $H_X$ proposed in \eqref{eq:HXfirst}, however, is much more broadly applicable.
Furthermore, a gradient-based method for data reduction---particularly one offering error guarantees in the general setting tackled here---has not, to the best of our knowledge, been previously proposed. 
The information theoretic error analysis of simultaneous data and parameter reduction that we develop in this paper is new as well.
We also show how this approach generalizes classic dimension reduction techniques, such as canonical correlation analysis.

The remainder of this paper is organized as follows. In Section~\ref{sec:DimensionReduction}, we describe the relation between posterior approximation error and gradients of the log-likelihood function, and we propose several methods to identify optimal variable projections. In
Section~\ref{sec:InformationTheory}, we interpret the approximation error using conditional mutual information. In Section~\ref{sec:GaussianLik}, we specialize our results to Gaussian error models and
discuss connections with related work on linear forward models. %
In Section~\ref{sec:OtherDimRedMethods} we compare our approach to other classic dimension reduction strategies. Section~\ref{sec:Algorithms} describes several inference algorithms that can exploit this joint dimension reduction of parameters and data. Lastly, Section~\ref{sec:NumericalExp} presents numerical experiments for inference problems involving partial differential equations, high-dimensional imaging, and stochastic differential equations.

\section{Reducing parameter and data dimensions} \label{sec:DimensionReduction}
Our joint parameter--data dimension reduction method relies on the detection of conditional independence between blocks of variables. Given two unitary matrices $\BasisX \in \R^{d \times d}$ and $\BasisY \in \R^{m \times m}$, partitioned as column blocks $\BasisX = [\BasisX_r, \BasisX_\perp]$ and $\BasisY = [\BasisY_s, \BasisY_\perp]$ with $\BasisX_r\in\R^{d\times r}$ and $\BasisY_s\in\R^{m\times s}$, we decompose $\X$ as
\begin{equation}\label{eq:Xdecomp}
 \X = \BasisX_r \X_r + \BasisX_\perp \X_{\perp} \qquad\text{where } \left\{\begin{array}{l} \X_r = \BasisX_r^{T}\X \\ \X_\perp = \BasisX_\perp^{T}\X \end{array} \right. ,
\end{equation}
and $\Y$ as
\begin{equation}\label{eq:Ydecomp}
 \Y = \BasisY_s \Y_s ~ + \BasisY_\perp \Y_{\perp} ~\qquad\text{where } \left\{\begin{array}{l} \Y_s = \BasisY_s^{T}\Y \\ \Y_\perp = \BasisY_\perp^{T}\Y \end{array} \right..
\end{equation}
In this decomposition, if $\X_\perp$ is independent of the data $\Y$ after conditioning on $\X_r$---that is, if $\X_\perp \ci \Y \vert \X_r$---then $\X_\perp$ is interpreted as the un-informed component of the parameter.
In the same way, if $\Y_\perp$ is independent of $\X$ after conditioning on $\Y_s$---i.e., $\Y_\perp \ci \X | \Y_s$---then $\Y_\perp$ is interpreted as the non-informative component of the data. 
Under these two conditional independence properties\footnote{In fact, $\X_\perp \ci \Y \vert \X_r$ and $\Y_\perp \ci \X_r | \Y_s$ are sufficient to write the factorization in \eqref{eq:OptPost}.}, the joint probability density function of $\X$ and $\Y$ factorizes as $\Joint(\x,\y) = \pi_{\X_\perp|\X_r}(\x_\perp|\x_r)\pi_{\X_r,\Y_s}(\x_r,\y_s) \pi_{\Y_\perp|\Y_s}(\y_\perp|\y_s)$, so that the posterior satisfies $\Post=\OptPost$ with
\begin{align}
 \OptPost(x|y) 
 &:= \pi_{\X_r|\Y_s}(\x_r|\y_s)  \pi_{\X_\perp|\X_r}(\x_\perp|\x_r). \label{eq:OptPost}  
\end{align}
In other words, the inference problem of characterizing $\X|\Y$ can be transformed into a lower-dimensional inference problem that involves characterizing $\X_r|\Y_s$.
In practice, however, the conditional independence criteria $\X_\perp \ci \Y \vert \X_r$ and $\Y_\perp \ci \X | \Y_s$ might not be exactly satisfied, and so $\Post\neq\OptPost$ in general.
In this case, our goal is to identify the unitary matrices $\BasisX$ and $\BasisY$ and to select the smallest possible (in a sense to be clarified later) effective dimensions $r\leq d$ and $s\leq m$ so that the KL divergence from $\OptPost$ to $\Post$ is controlled in expectation over the data; that is
\begin{equation} \label{eq:PostErrorConstraint}
    \mathbb{E}_{\Y}\left[\KLDiv(\Post(\cdot|\Y)||\OptPost(\cdot|\Y))\right] \leq \epsilon  ,
\end{equation}
for some prescribed tolerance $\epsilon > 0$.
The following proposition shows that, given the parameter and data decompositions \eqref{eq:Xdecomp} and \eqref{eq:Ydecomp}, the posterior approximation $\OptPost$ in \eqref{eq:OptPost} is optimal for the expected KL loss.
\begin{proposition} \label{prop:optimalPosteriorApprox}
 Let $(\X,\Y)\sim\Joint$ be decomposed as in \eqref{eq:Xdecomp} and \eqref{eq:Ydecomp}. Then the posterior approximation $\OptPost$ defined in \eqref{eq:OptPost} satisfies
 \begin{equation} \label{eq:KLDiv_OptPost}
  \mathbb{E}\left[\KLDiv(\Post(\cdot|\Y)||\OptPost(\cdot|\Y))\right] \leq \mathbb{E}\left[\KLDiv(\Post(\cdot|\Y)||\ApproxPost(\cdot|\Y))\right] ,
 \end{equation}
 for any posterior approximation of the form $\ApproxPost(\x|\y)=f_1(\x_r,\y_s)f_2(\x_\perp,\x_r)$ with non-negative functions $f_1,f_2$.
\end{proposition}
\begin{proof}
    See Appendix~\ref{app:proofs}.
\end{proof}
Thus, once the matrices $U,V$ are identified and the effective dimensions $r,s$ are determined, the optimal posterior approximation \eqref{eq:OptPost} is accessible via 
$$
 \OptPost(x|y) \propto \pi_{\Y_s|\X_r}(\y_s|\x_r)  \pi_{\X}(\x),
$$
where the reduced likelihood $\pi_{\Y_s|\X_r}(\y_s|\x_r)$ is accessible by marginalizing the likelihood function $\Like(\y|\x)$ over $\y_\perp$ and $\x_\perp$ using the prior weight, i.e.,
\begin{align}
    \pi_{\Y_s|\X_r}(\y_s|\x_r) 
    &= \int_{\R^{m-s}} \int_{\R^{d-r}} \Like(\BasisY_s\y_s+\BasisY_\perp\y_\perp|\BasisX_r\x_r+\BasisX_\perp\x_\perp) \pi_{\X_\perp|\X_r}(\x_\perp|\x_r) \d \x_\perp \d \y_\perp \label{eq:ApproxLikelihood}\\
    &= \frac{1}{\pi_{\X_r}(\x_r)}\int_{\R^{m-s}} \int_{\R^{d-r}} \Joint(\BasisX_r\x_r+\BasisX_\perp\X_\perp,\BasisY_s\y_s+\BasisY_\perp\y_\perp) \d \x_\perp \d \y_\perp \nonumber .
\end{align}
We note that with $s=m$ (i.e., no data reduction), the reduced likelihood $\pi_{\Y_s|\X_r}=\pi_{\Y|\X_r}$ coincides with that used in~\citet{zahm2018certified, zahm2020gradient} when reducing the dimension of the parameter.

Next, in Section~\ref{subsec:UpperBound}, we provide a tractable upper bound for the posterior approximation error that depends explicitly on the decompositions of $\X$ and $\Y$. Section~\ref{subsec:ConstructBasis} then provides two methods for identifying low-dimensional subspaces that minimize this upper bound, and Section~\ref{sec:SelectingTheDimensions} presents a procedure for selecting reduced dimensions $r,s$ that satisfy the constraint in~\eqref{eq:PostErrorConstraint}.

\subsection{Gradient-based bound on expected posterior approximation error} \label{subsec:UpperBound}

In this section we present our main result, which consists in a gradient-based bound on the expected KL divergence~\eqref{eq:PostErrorConstraint}. This bound will guide the construction of the matrices $U$ and $V$. In the following, $\|\cdot\|$ denotes the canonical norm of the Euclidean space.

\begin{definition}[Logarithmic Sobolev inequality]\label{def:LSI} 
 A random variable $\Z$ with density $\pi_\Z$ on $\R^p$ satisfies the \emph{logarithmic Sobolev inequality} if there exists a constant $C<\infty$ such that
 \begin{equation}\label{eq:LSI}
  \int h(\z)\log\left(\frac{h(\z)}{\int h \d\pi_\Z}\right) \pi_\Z(\z)\d\z
  \leq \frac{C}{2} \int \|\nabla h(\z)\|^2 h(\z) \,\pi_\Z(\z)\d\z ,
 \end{equation}
 holds for any smooth function $h\colon\R^p\rightarrow\R_{\geq0}$. The smallest constant $C=C(\pi_Z)$ such that \eqref{eq:LSI} holds is called the \emph{logarithmic Sobolev constant} of $\Z$.
\end{definition}

\begin{definition}[Subspace logarithmic Sobolev inequality]\label{def:condLSI} 
 A random variable $\Z$ with density $\pi_\Z$ on $\R^p$ satisfies the \emph{subspace logarithmic Sobolev inequality} if there exists a constant $\overline{C}<\infty$ such that for any unitary matrix $W\in\R^{p\times p}$ and for any block decomposition $W=[W_t,W_\perp]$ with $W_t\in\R^{p\times t}, t\leq p$, and for any $\z_\perp\in\R^{p-t}$, the conditional random vector $\Z_t|\Z_\perp=\z_\perp$ with $\Z_t=W_t^T\Z$ and $\Z_\perp=W_\perp^T\Z$
 satisfies the logarithmic Sobolev inequality with
 \begin{equation}\label{eq:condLSI}
  C(\pi_{\Z_t|\Z_\perp=\z_\perp})\leq \overline{C}.
 \end{equation}
 The smallest constant $\overline{C}=\overline{C}(\pi_\Z)$ such that \eqref{eq:condLSI} holds is called the \emph{subspace logarithmic Sobolev constant} of $\Z$.
\end{definition}

\begin{theorem} \label{thm:ExpKLBound}
 Let $(\X,\Y)$ be a random vector in $\R^{d}\times\R^m$ which satisfies the subspace logarithmic Sobolev inequality with constant $\overline{C}(\Joint) < \infty$. Then for any unitary matrices $\BasisX = [\BasisX_r, \BasisX_\perp]\in\R^{d\times d}$ and $\BasisY = [\BasisY_s, \BasisY_\perp]\in\R^{m\times m}$ we have
 \begin{equation} \label{eq:KL_UpperBound}
 \mathbb{E}\left[\KLDiv(\Post(\cdot,Y)||\OptPost(\cdot,Y))\right] \leq \overline{C}(\Joint)^2 \left(\Tr(\BasisX_\perp^{T} H_\X \BasisX_\perp) + \Tr(\BasisY_\perp^{T} H_\Y \BasisY_\perp) \right) .
\end{equation}
 Here, $\OptPost$ is as in \eqref{eq:OptPost} and the matrices $H_\X \in \R^{d \times d}$ and $H_\Y \in \R^{m \times m}$ are given by
\begin{align}
H_\X &= \int \big(\nabla_\X \nabla_\Y \log\Like(y|x)\big)^T\big(\nabla_\X \nabla_\Y \log\Like(\y|\x)\big) \Joint(\x,\y)\d\x\d\y \label{eq:DiagnosticMatrixX} \\
H_\Y &= \int \big(\nabla_\X \nabla_\Y \log\Like(y|x)\big)\big(\nabla_\X \nabla_\Y \log\Like(\y|\x)\big)^{T} \Joint(\x,\y)\d\x\d\y, \label{eq:DiagnosticMatrixY}
\end{align}
where the matrix $\nabla_\X \nabla_\Y \log\Like(y|x)\in\R^{m\times d}$ is defined by
$$
 \Big(\nabla_\X \nabla_\Y \log\Like(y|x)\Big)_{i,j}
 = \partial_{x_j}\partial_{y_i} \log\Like(y|x).
$$

\end{theorem}
\begin{proof}
 See Section~\ref{sec:InformationTheory}.
\end{proof}

Throughout this paper, we will use the bound~\eqref{eq:KL_UpperBound} by finding structured unitary matrices $\BasisX,\BasisY$ that minimize the right-hand side of \eqref{eq:KL_UpperBound}. Due to their central role, the matrices $H_\X$ and $H_\Y $ are called the \emph{diagnostic matrices}.

Before going further, let us comment on the assumption $\overline{C}(\Joint)<\infty$. As shown in \citet{zahm2018certified}, a sufficient condition for a distribution $\pi_\Z$ to satisfy the subspace log-Sobolev inequality is that it has convex support and that its density be of the form $\pi_\Z(\z)=\exp(-V(\z)-\Psi(\z))$, where $V$ is a smooth convex function such that $\nabla^2 V(\z) \succeq \rho I$ for some $\rho>0$, and where $\Psi$ is a function with bounded oscillation such that $\kappa=\sup\Psi-\inf\Psi <\infty$. Then, from the Bakry--\'{E}mery theorem \citep{bakry1985diffusions} and the Holley--Stroock perturbation lemma \citep{holley1986logarithmic}, we obtain $\overline{C}(\pi_\Z) \leq \exp(\kappa)/\rho$. 
As shown in the following example, this condition is (trivially) satisfied when the joint distribution $\Joint$ is Gaussian. 
We refer the reader to~\citet{zahm2018certified} for additional examples of distributions that satisfy the subspace log-Sobolev inequality. In these general cases, the constant $\overline{C}(\Joint)$ will not be available or computable in practice. Yet we can still exploit the bound~\eqref{eq:KL_UpperBound} without having access to $\overline{C}(\Joint)$. %

\begin{example}[Gaussian joint density] \label{ex:gaussianJoint}
 Let $\Joint$ be the joint density
 \begin{equation}\label{eq:gaussianJoint}
  \Joint(\x,\y) \propto \exp\left( -\frac12\|\y - G\x\|^2 -\frac12\| x \|^2 \right),
 \end{equation}
 where $G \in\R^{m\times d}$.
 This corresponds to a Bayesian inverse problem with standard normal prior, linear forward model $\x\mapsto G\x$, and standard normal observational error. Given that $\Joint(\x,\y)\propto\exp(-V(\x,\y))$ with the quadratic potential
 $V(\x,\y)=\frac12(\begin{smallmatrix}\x\\\y\end{smallmatrix})^T\Sigma^{-1}
 (\begin{smallmatrix}\x\\\y\end{smallmatrix})$ where
 $$
  \Sigma \coloneqq \Cov\begin{pmatrix}\X\\\Y\end{pmatrix} 
  = \begin{bmatrix} \Id_d +G^TG & -G^T \\ -G &  \Id_m \end{bmatrix}^{-1}
  = \begin{bmatrix} \Id_d & G^T \\ G & GG^T + \Id_m \end{bmatrix},
 $$
 we deduce (see \citet{zahm2018certified}) that $\overline{C}(\Joint)$ is bounded by $\lambda_{\max}(\Sigma)$, i.e., the largest eigenvalue of the joint covariance matrix $\Sigma$. 
 As shown in Appendix~\ref{app:calculations}, $\lambda_{\max}(\Sigma)$ can be computed explicitly, so that we obtain
 \begin{equation} \label{eq:MaxEig_Gaussian}
  \overline{C}(\Joint) \leq 
  \frac{1}{2} \left(2 + \sigma_{\max}(G)^2 + \sigma_{\max}(G)\sqrt{\sigma_{\max}(G)^2 + 4} \right),
 \end{equation}
 where $\sigma_{\max}(G)$ is the maximum singular value of $G$.
 Furthermore, the diagnostic matrices $H_\X$ and $H_\Y $ in \eqref{eq:DiagnosticMatrixX} and \eqref{eq:DiagnosticMatrixY} are given by $H_\X = G^T G$ and $H_\Y = G G^T$.
 
\end{example}

Next we show that the joint distribution $\Joint$ arising in certain \emph{nonlinear} inverse problems can also satisfy the subspace logarithmic Sobolev inequality.

\begin{example}[Nonlinear inverse problem with Gaussian noise and Gaussian prior]
Let the joint density of $\X$ and $\Y$ be
 $$
 \Joint(\x,\y) \propto \exp\left(-\frac{1}{2}\|\y - G(\x) \|^2 -\frac12\|x\|^2\right),
 $$
 where $G \colon \mathbb{R}^{d} \rightarrow \mathbb{R}^{m}$ is a nonlinear forward model.
 Denoting the Jacobian of the forward model by $\nabla G(\x)\in\R^{m\times d}$, we can write 
 \begin{align*}
  -\nabla^2 \log\Joint(\x,\y) 
  &=
  \begin{bmatrix}
    -\nabla_{X}\nabla_X \log\Joint(\x,\y) & -\nabla_Y \nabla_X \log\Joint(\x,\y) \\
    -\nabla_X \nabla_Y \log\Joint(\x,\y) & -\nabla_{Y}\nabla_Y \log\Joint(\x,\y) 
  \end{bmatrix} \\
  &= 
  \begin{bmatrix}
   \Id_d + \nabla G(\x)^T  \nabla G(\x) & -\nabla G(\x)^T\\
   - \nabla G(\x)  & \Id_m 
  \end{bmatrix} 
  -
  \begin{bmatrix}
   A(\x,\y)&  0\\ 0 & 0
  \end{bmatrix},
 \end{align*}
 where the matrix $A(\x,\y)\in\R^{d\times d}$ is given by $(A(\x,\y))_{i,j} = \sum_{k=1}^n \partial_{i,j}^2 G_k(x) (\y-G(\x))_k.$
 As in the previous example, we have
 \begin{align*}
  \lambda_{\min}\left(\begin{bmatrix}
   \Id_d + \nabla G(\x)^T  \nabla G(\x) & -\nabla G(\x)^T\\
   - \nabla G(\x)  & \Id_m 
  \end{bmatrix}\right)
  = \lambda(\x)^{-1} ,
 \end{align*}
 where $\lambda(\x)=\frac{1}{2} (2 + \sigma_{\max}(\nabla G(\x))^2 + \sigma_{\max}(\nabla G(\x))\sqrt{\sigma_{\max}(\nabla G(\x))^2 + 4} )$ so that
 $$
 -\nabla^2 \log\Joint(\x,\y)  \succeq (\lambda(\x)^{-1} - \lambda_{\max}(A(\x,\y)) ) 
 \begin{bmatrix}
   \Id_d&  0\\ 0 & \Id_m
  \end{bmatrix}.
 $$
 Therefore, if there exists a constant $C<\infty$ such that $\lambda(\x)^{-1} - \lambda_{\max}(A(\x,\y)) \geq 1/C$ uniformly over $\x,\y$, then $\Joint$ satisfies the subspace log-Sobolev inequality with $\overline C(\Joint)\leq C$.
 Furthermore, since $\nabla_\X \nabla_\Y \log\Like(\y|\x) = \nabla G(x) \in \R^{m \times d}$, the diagnostic matrices are given by
 \begin{align*}
    H_{\X} &=  \int  \nabla G(x)^{T} \nabla G(x) \Prior(x)\d\x  \\
    H_{\Y} &=  \int  \nabla G(x)  \nabla G(x)^T  \Prior(x)\d\x .
 \end{align*}
\end{example}

\subsection{Constructing $\BasisX,\BasisY$ by minimizing the upper bound} \label{subsec:ConstructBasis}

In this section, we assume the reduced dimensions $r$ and $s$ are prescribed. (A discussion of how to select $r,s$ is postponed to Section \ref{sec:SelectingTheDimensions}.)
We propose two different approaches to build the unitary matrices $\BasisX=[\BasisX_r,\BasisX_\perp]$ and $\BasisY=[\BasisY_s,\BasisY_\perp]$.
The first, referred to as \emph{optimal rotation}, consists in minimizing the upper bound~\eqref{eq:KL_UpperBound} by solving
\begin{equation} \label{eq:BasisOpt}
\min_{\BasisX_\perp, \BasisY_\perp } \Tr(\BasisX_\perp^T H_\X \BasisX_\perp) + \Tr(\BasisY_\perp^T H_\Y \BasisY_\perp),
\end{equation}
subject to $U_\perp^T U_\perp = \Id_{d-r}$ and $V_\perp^T \BasisY_\perp = \Id_{m-s}$.
The second approach, referred to as \emph{optimal permutation}, consists in solving \eqref{eq:BasisOpt} with the additional constraint that $\BasisX,\BasisY$ are permutation matrices. That way, $\X_r=\BasisX_r^T \X$ and $\Y_s=\BasisY_s^T\Y$ contain a subset of coordinates of $\X$ and $\Y$ and hence the dimension reduction corresponds to a coordinate selection. %

In both approaches, the optimal solutions $\BasisX,\BasisY$ are independent of the reduced dimensions $r$ and $s$. More specifically, there exist matrices $\BasisX,\BasisY$ (independent of $r,s$) such that the solution to~\eqref{eq:BasisOpt} can be extracted from the last columns of $\BasisX,\BasisY$ for any $r,s$.

\subsubsection{Optimal rotation} \label{subsec:OptRotation} We recall Corollary 4.3.39 in~\citet{horn2012matrix} for the variational characterization of eigenvalues of Hermitian matrices:

\begin{proposition} \label{eq:traceMin} Let $H \in \R^{p \times p}$ be a symmetric positive definite matrix with eigenpairs $(\lambda_i,w_i) \in \R_{>0} \times \R^p$, meaning $Hw_i = \lambda_i w_i$, where $\lambda_i \geq \lambda_{i+1}$ and $\|w_i\|_2 = 1$ for all $i$. Then, for any $t < p$ we have
\begin{equation*}
    \min_{\substack{W_\perp \in \R^{p \times (p-t)}\\W_\perp^T W_\perp = \Id_{p-t}}} \Tr(W_\perp^T H W_\perp) = \sum_{i=t+1}^p \lambda_i,
\end{equation*}
where the solution is given by $W_\perp = [w_{t+1},\dots,w_{p}]$. 
\end{proposition}

Let $(\lambda_i(H_\X),u_i)$ and $(\lambda_i(H_\Y),v_i)$ denote the $i$-th largest eigenpairs of $H_\X$ and $H_\Y$, respectively. 
Then, Proposition~\ref{eq:traceMin} ensures that for any $r,s$,
\begin{align*}
    \BasisX_\perp &= [u_{r+1},\hdots,u_d] \\
    \BasisY_\perp &= [v_{s+1},\hdots,v_m] ,
\end{align*}
is the optimal solution to \eqref{eq:BasisOpt}. This choice yields the optimal bound
\begin{equation} \label{eq:PostErr_Subspaces}
    \mathbb{E} \left[\KLDiv(\Post(\cdot|\Y)||\OptPost(\cdot|\Y)) \right] \leq \overline{C}(\Joint)^2 \left(\sum_{i=r+1}^d \lambda_i(H_\X) + \sum_{i=s+1}^m \lambda_i(H_\Y) \right).
\end{equation}
The eigenvectors $u_i,v_i$ can be precomputed without knowing $r$ and $s$. The advantage of this construction is that to increase $r$ or $s$, one only needs to compute the additional eigenvectors.

\begin{remark}
    In practice, it is sufficient to compute the matrices $\BasisX_r=[u_1,\hdots,u_r]$ and $\BasisY_s=[v_1,\hdots,v_s]$ to reduce the parameter and data dimensions; see Section~\ref{sec:Algorithms}. The (possibly much larger) matrices $\BasisX_\perp$ and $\BasisY_\perp$ are never assembled in practice.
\end{remark}

\subsubsection{Optimal permutation} \label{subsec:OptPermutation} 
We now constrain $U$ and $V$ to be permutation matrices so that $\BasisX\X = (\X_{\sigma_\X(1)},\hdots,\X_{\sigma_\X(d)})$ and $\BasisY\Y = (\Y_{\sigma_\Y(1)},\hdots,\Y_{\sigma_\Y(m)})$ where $\sigma_\X$ and $\sigma_\Y$ are permutations of $\{1,\hdots,d\}$ and $\{1,\hdots,m\}$, respectively. Then \eqref{eq:BasisOpt} becomes
\begin{equation} \label{eq:VariableSelection}
\min_{\sigma_\X,\sigma_\Y} \left(\sum_{i=r+1}^{d} (H_\X)_{\sigma_\X(i),\sigma_{\X}(i)} + \sum_{i=s+1}^{m} (H_\Y)_{\sigma_\Y(i),\sigma_{\Y}(i)}\right).
\end{equation}
The optimal permutations $\sigma_\X$ and $\sigma_\Y$ are those which sort the \emph{diagonal} terms of $H_\X$ and $H_\Y$ in decreasing order; that is
\begin{align*}
 (H_\X)_{\sigma_\X(i),\sigma_{\X}(i)} &\geq (H_\X)_{\sigma_\X(i+1),\sigma_{\X}(i+1)} \\
 (H_\Y)_{\sigma_\Y(i),\sigma_{\Y}(i)} &\geq (H_\Y)_{\sigma_\Y(i+1),\sigma_{\Y}(i+1)}.
\end{align*}
This choice yields the upper bound
\begin{equation}  \label{eq:PostErr_Coordinate}
    \mathbb{E} \left[\KLDiv(\Post(\cdot|\Y)||\OptPost(\cdot|\Y)) \right] \leq \overline{C}(\Joint)^2 \left(\sum_{i=r+1}^{d} (H_\X)_{\sigma_\X(i),\sigma_{\X}(i)} + \sum_{i=s+1}^{m} (H_\Y)_{\sigma_\Y(i),\sigma_{\Y}(i)} \right). %
\end{equation}
Let us note that because permutation matrices are unitary matrices, the bound in~\eqref{eq:PostErr_Coordinate} is larger than or equal to the optimal bound in~\eqref{eq:PostErr_Subspaces}. Thus the optimal permutation approach might be less efficient than the optimal rotation approach, but it preserves the interpretability of the reduced components.

\subsection{Selecting the reduced dimensions}\label{sec:SelectingTheDimensions}

We now discuss the problem of selecting the reduced dimensions. We propose to select $r$ and $s$ by minimizing the computational cost of exploring the reduced posterior $\OptPost$ under the constraint that $\OptPost$ is sufficiently accurate.

Let $c(r,s)\geq0$ be a function that reflects the computational cost and dimension dependence/scaling of solving the reduced Bayesian inference problem with posterior $\OptPost$ as in~\eqref{eq:OptPost}. 
The choice of $c(r,s)$ strongly depends on the inference method (e.g., different MCMC algorithms, variational inference, conditional transport maps~\citep{chewi2021optimal, papamakarios2016fast,cui2021conditional}). %
For instance, we may have $c(r,s) = \alpha_\X r + \alpha_\Y s$ or $c(r,s) = \alpha_\X r^2 + \alpha_\Y s^2$ for some weights $\alpha_\X,\alpha_\Y \geq 0$ which prescribe the relative cost of truncating the parameter or data dimensions.
Given a prescribed tolerance $\epsilon$, the ideal way to select $r,s$ is to solve
\begin{align}\label{eq:FindingReducedDim_ideal}
    \min_{r,s} &\; c(r,s)
    \qquad\text{s.t.} \qquad \mathbb{E}\left[\KLDiv(\Post(\cdot|\Y)||\OptPost(\cdot|\Y))\right] \leq \epsilon. 
\end{align}
Because the expected KL divergence is not accessible in practice, we rather select $r,s$ by solving
\begin{align} \label{eq:FindingReducedDim}
    \min_{r,s} &\; c(r,s)
    \qquad\text{s.t.} \qquad  B(r,s) \leq \epsilon' ,
\end{align}
where $B(r,s)$ is defined by either
\begin{align*}
 B(r,s)&=\sum_{i=r+1}^d \lambda_i(H_\X) + \sum_{i=s+1}^m \lambda_i(H_\Y), \\
 \text{or}\qquad B(r,s)&=\sum_{i=r+1}^{d} (H_\X)_{\sigma_\X(i),\sigma_{\X}(i)} + \sum_{i=s+1}^{m} (H_\Y)_{\sigma_\Y(i),\sigma_{\Y}(i)} ,
\end{align*}
depending on whether one uses the optimal rotation approach (Section~\ref{subsec:OptRotation}) or the optimal permutation approach (Section~\ref{subsec:OptPermutation}) to build $\BasisX,\BasisY$.
Given that 
$$
 \mathbb{E}\left[\KLDiv(\Post(\cdot|\Y)||\OptPost(\cdot|\Y))\right]   \leq \overline{C}(\Joint) B(r,s),
$$
the solution to~\eqref{eq:FindingReducedDim} with $\epsilon'=\epsilon/\overline{C}(\Joint)$ provides a feasible approximate solution to \eqref{eq:FindingReducedDim_ideal}. 
In the case where the log-Sobolev constant is not known, we propose to select $r,s$ by solving~\eqref{eq:FindingReducedDim} with $\epsilon'=\epsilon$. While the resulting solution may not satisfy $\mathbb{E}[\KLDiv(\Post(\cdot|\Y)||\OptPost(\cdot|\Y))] \leq \epsilon$, it still provides a good heuristic for selecting the reduced dimensions, as illustrated in Section~\ref{sec:NumericalExp}. Note that solving~\eqref{eq:FindingReducedDim} for different tolerances yields a mapping $\epsilon' \mapsto (r,s)$, which yields the same locus of optimal reduced dimensions $r,s$ as the solution to~\eqref{eq:FindingReducedDim_ideal}.

\begin{remark} When $c(r,s) = \alpha_\X r + \alpha_\Y s$, the optimization problem in~\eqref{eq:FindingReducedDim} can be formulated as a $0-1$ knapsack problem, which is known to be NP-complete~\citep{kellerer2004introduction}. Given that $(r,s)$ is only two-dimensional, we can often enumerate all combinations of reduced dimensions to find the optimal solutions. An alternative approximate solution, which does not require enumeration, is to split the constraint and to select $(r,s)$ individually based on a weighted error tolerance for the parameters and data. For example, %
we can identify $(r,s)$ by finding the smallest integers that meet the constraints
$$ \sum_{i=1}^{r} \lambda_i(H_\X) \leq \frac{\alpha_\X}{\alpha_\X + \alpha_\Y}\epsilon , \quad\quad \sum_{i=1}^{s} \lambda_i(H_\Y) \leq \frac{\alpha_\Y}{\alpha_\X + \alpha_\Y}\epsilon.$$
The setting $\alpha_\X = \alpha_\Y$ corresponds to choosing the reduced dimensions so that the errors from both reductions are balanced.
\end{remark}

\section{Information theory and proof of Theorem~\ref{thm:ExpKLBound}} \label{sec:InformationTheory}

In this section we relate the posterior approximation error to information-theoretic quantities that measure conditional independence. We then show how to bound these quantities to derive the upper bound in Theorem~\ref{thm:ExpKLBound}.

We begin by defining mutual information and conditional mutual information, which are two well known measures of the strength of dependence between random variables.

\begin{definition} Let $\X$ and $\Y$ be two random variables with joint density $\pi_{\X,\Y}$. The mutual information between $\X$ and $\Y$ is given by
$$I(\X;\Y) \coloneqq \int \pi_{\X,\Y}(\x,\y) \log \left(\frac{\pi_{\X,\Y}(\x,\y)}{\pi_{\X}(\x)\pi_{\Y}(\y)} \right) \d\x\d\y,$$
where $\pi_{\Y}(\y) = \int \pi_{\X,\Y}(\x,\y)\d\x$ and $\pi_{\X}(\x) = \int \pi_{\X,\Y}(\x,\y)\d\y$. %
\end{definition}
The mutual information is equivalently expressed as the KL divergence from the product of the marginal densities to the joint probability density function, i.e., $I(\X;\Y) = \KLDiv(\pi_{\X,\Y}||\pi_{\X}\pi_{\Y})$. The mutual information measures the dependence of $\X$ and $\Y$. In particular, $I(\X;\Y) = 0$ if and only if $\pi_{\X,\Y} = \pi_{\X}\pi_{\Y}$, meaning that $\X$ and $\Y$ are independent. 
 
\begin{definition} The conditional mutual information between random variables $\X$ and $\Y$ given a third random variable $\Z$ with joint density $\pi_{\X,\Y,\Z}$ is given by
$$I(\X;\Y|\Z) \coloneqq \int \pi_{\X,\Y,\Z}(\x,\y,\z) \log \left(\frac{\pi_{\X,\Y|\Z}(\x,\y|\z)}{\pi_{\X|\Z}(\x|\z)\pi_{\Y|\Z}(\y|\z)} \right) \d\x\d\y\d\z,$$ %
where $\pi_{\X|\Z}(\x|\z) = \int \pi_{\X,\Y|\Z}(\x,\y|\z)\d\y$ and $\pi_{\Y|\Z}(\y|\z) = \int \pi_{\X,\Y|\Z}(\x,\y|\z)\d\x$. %
\end{definition}

Analogously, the conditional mutual information $I(\X;\Y|\Z)$ %
is defined as the KL divergence from the product $\pi_{\Y|\Z}\pi_{\X|\Z}$ to the conditional density $\pi_{\X,\Y|\Z}$ in expectation over $\Z$, meaning $I(\X;\Y|\Z)= \mathbb{E}[\KLDiv(\pi_{\X,\Y|\Z}(\cdot,\cdot|\Z)||\pi_{\X|\Z}(\cdot|\Z)\pi_{\Y|\Z}(\cdot|\Z))]$.  %
The conditional mutual information serves as a measure of conditional independence between random variables, i.e., $I(\X;\Y|\Z) = 0$ if and only if $\X \ci \Y | \Z$.

The following proposition shows that the expected KL divergence from the optimal posterior approximation (given some decomposition \eqref{eq:Xdecomp}--\eqref{eq:Ydecomp}) to the true posterior is related to a difference between (conditional) mutual informations.
\begin{proposition} \label{prop:PosteriorErr_MI}
Let $\Post$ be the distribution of $\X|\Y$ and $\OptPost$ be the optimal posterior approximation in~\eqref{eq:OptPost} with $r$-dimensional informed parameters and $s$-dimensional informative data. %
Then we have
\begin{align} 
    \mathbb{E} \left[\KLDiv(\Post(\cdot|\Y)||\OptPost(\cdot|\Y)) \right] 
    &= I(\X;\Y) - I(\X_r;\Y_s)\label{eq:MIdiff_JointDimReduction}  \\
    &= I(\X_\perp;\Y|\X_r) + I(\X;\Y_\perp|\Y_s) - I(\X_\perp;\Y_\perp|\X_r,\Y_s) .\label{eq:ExpKL_MIconnection}
\end{align}
\end{proposition}
\begin{proof}
    See Appendix \ref{app:proofs}.
\end{proof}
\begin{remark}
The right-hand sides of~\eqref{eq:MIdiff_JointDimReduction} and of \eqref{eq:ExpKL_MIconnection} simplify if we only consider reducing either the dimension of the parameter or the data solely. For instance, if the data are not reduced, i.e., $\Y_s = \Y$, we have 
$\mathbb{E} [\KLDiv(\Post(\cdot|\Y)||\OptPost(\cdot|\Y)) ] = I(\X_\perp;\Y|\X_r) .$
Analogously, if the parameter is not reduced, i.e., $\X_r = \X$, we have 
$\mathbb{E} [\KLDiv(\Post(\cdot|\Y)||\OptPost(\cdot|\Y)) ] = I(\X;Y_\perp|\Y_s) .$
\end{remark}

\begin{remark}
An important property of mutual information is that it is invariant to invertible marginal transformations of the variables. 
For instance, by applying the linear transformations $\overline{X} = AX$ and $\overline{Y} = BX$ for some invertible matrices $A \in \R^{d \times d}$ and $B \in \R^{m \times m}$, we have $I(\X;\Y) = I(\overline{\X};\overline{\Y})$ but also $I(\X_\perp;\Y|\X_r) = I(\X_\perp;\overline{\Y}|\X_r)$ and $I(\X;\Y_\perp|\Y_s) = I(\overline{\X};\Y_\perp|\Y_s)$.
\end{remark}

While the (conditional) mutual information is tractable to compute for Gaussians and certain classes of  parametric distributions, it does not admit a closed-form expression for arbitrary non-Gaussian distributions. For a density that satisfies the subspace log-Sobolev inequality in~\eqref{eq:condLSI}, the following proposition provides an upper bound for the conditional mutual information based on the integrated mixed partial derivatives of the log-density. %

\begin{proposition} \label{prop:CMI_bound}
Let $\pi_{\X,\Y,\Z}$ be the joint density of random variables $(\X,\Y,\Z)$, satisfying the subspace logarithmic Sobolev inequality with constant $\overline{C}(\pi_{\X,\Y,\Z})$. %
Then, the conditional mutual information is upper bounded by
\begin{equation} \label{eq:CMI_UpperBound}
I(\X;\Y|\Z) \leq \overline{C}(\pi_{\X,\Y,\Z})^2 \mathbb{E} \| \nabla_\X \nabla_\Y \log \pi_{\X,\Y,\Z}(\X,\Y,\Z) \|_{F}^2,
\end{equation}
where $\|\cdot\|_{F}$ denotes the Frobenius norm.
\end{proposition}
\begin{proof}
 The proof follows closely from the proof of Theorem 2 in~\citet{baptista2021learning}. It is given in Appendix~\ref{app:proofs}.
\end{proof}

Collecting the results in Propositions~\ref{prop:PosteriorErr_MI} and~\ref{prop:CMI_bound}, we now give the proof of Theorem~\ref{thm:ExpKLBound}. \\

\begin{proof}[\textbf{Proof of Theorem~\ref{thm:ExpKLBound}}]\label{proof:ExpKLBound}
    Because the conditional mutual information is positive, Proposition~\ref{prop:PosteriorErr_MI} lets us write
    \begin{equation*} \label{eq:CMI_joint_upperbound}
     \mathbb{E} \left[\KLDiv(\Post(\cdot|\Y)||\OptPost(\cdot|\Y)) \right] \leq %
     I(\X_\perp;\Y|\X_r) + I(\X;\Y_\perp|\Y_s).
    \end{equation*}
    By Proposition~\ref{prop:CMI_bound}, each term on the right-hand side above is upper bound by the expectation of mixed partial derivatives of $\Joint$ as 
    \begin{align*}
     \mathbb{E} & \left[\KLDiv(\Post(\cdot|\Y)||\OptPost(\cdot|\Y)) \right] \\
     &\leq 
     \overline{C}(\Joint)^2 \left(\mathbb{E}\|\nabla_{\X_\perp}\nabla_{\Y} \log \Joint(\X,\Y)\|_F^2 + \mathbb{E}\|\nabla_{\X}\nabla_{\Y_\perp} \log \Joint(\X,\Y)\|_F^2 \right)\\
     &= \overline{C}(\Joint)^2 \left(\mathbb{E}\|\nabla_\X \nabla_\Y \log \Joint(\X,\Y) \BasisX_\perp \|_F^2 + \mathbb{E}\|\BasisY_\perp^T \nabla_\X \nabla_\Y \log \Joint(\X,\Y) \|_F^2 \right) \\
     &= \overline{C}(\Joint)^2 \left(\mathbb{E}\|\nabla_\X \nabla_\Y \log \Like(\Y|\X) \BasisX_\perp \|_F^2 + \mathbb{E}\|\BasisY_\perp^T \nabla_\X \nabla_\Y \log \Like(\Y|\X) \|_F^2 \right).
    \end{align*}
    Expanding the Frobenius norm using the trace, we arrive at equation~\eqref{eq:KL_UpperBound}.
\end{proof}

\section{Gaussian error models} \label{sec:GaussianLik}
In this section, we consider the data-generating process $ \Y = G(\X) + \varepsilon$, where $G \colon \mathbb{R}^{d} \rightarrow \mathbb{R}^{m}$ is a (nonlinear) forward model and $\varepsilon \sim \mathcal{N}(0,\Cobs)$ is a Gaussian observational error which is independent of $X$.
This situation corresponds to a  likelihood function $\Like(\y|\x) \propto \exp(-\frac{1}{2}\|\y - G(\x) \|_{\Cobs^{-1}}^2 )$ and a joint density of the form of
\begin{equation}\label{eq:gaussianLikelihood}
 \Joint(\x,\y) \propto \exp\left(-\frac{1}{2}\|\y - G(\x) \|_{\Cobs^{-1}}^2 \right) \Prior(\x) ,
\end{equation}
where $\Prior$ is any prior density. 
Without further assumptions, the subspace log-Sobolev constant $\overline{C}(\Joint)$ remains unknown.

\subsection{Whitening} \label{sec:WhiteningGaussianLik}

Next, we propose a change of variables for $\X$ and $\Y$ which can be interpreted as a preconditioning of the dimension reduction procedure. Notice that with a change of variables $\overline{X} = A\X$ and $\overline{Y} = B\Y$, the left-hand side of \eqref{eq:KL_UpperBound} remains unchanged (see Section~\ref{sec:InformationTheory}) while the right-hand side is modified in several ways through the subspace log-Sobolev constant and the diagnostic matrices. There is freedom in this choice.
Finding a change of variables which minimizes the right-hand side of \eqref{eq:KL_UpperBound}---i.e., which yields the \emph{tightest} upper bound on the posterior approximation error---is a difficult task, mostly because the subspace log-Sobolev constant $\overline{C}(\pi_{\X_A,\Y_B})$ is not readily available. 
Instead, we propose a heuristic which consists of whitening the parameter and the data as follows:
\begin{equation}\label{eq:whitening}
 \overline\X =\Cpr^{-1/2}\X
 \quad\text{and}\quad
 \overline\Y = \Cobs^{-1/2}\Y,
\end{equation}
where $\Cpr \coloneqq \Cov(X)$ %
is the prior covariance, assuming it exists.
Then, we reduce the dimensions of $\overline\X$ and $\overline\Y$ using the corresponding diagnostic matrices $H_{\overline\X}$ and $H_{\overline\Y}$ which, using \eqref{eq:gaussianLikelihood} and \eqref{eq:whitening}, are given by
\begin{align}
    H_{\overline\X} &= \Cpr^{1/2} \left( \int  \nabla G(x)^{T} \Cobs^{-1} \nabla G(x) \Prior(x)\d\x \right)\Cpr^{1/2} \label{eq:Gaussian_Hx} \\
    H_{\overline\Y} &= \Cobs^{-1/2} \left( \int  \nabla G(x) \Cpr \nabla G(x)^T  \Prior(x)\d\x \right)\Cobs^{-1/2}. \label{eq:Gaussian_Hy} 
\end{align}
Denoting by $\overline\BasisX_r=[\overline{u}_1,\hdots,\overline{u}_r]$ and $\overline\BasisY_s=[\overline{v}_1,\hdots,\overline{v}_s]$ the matrices containing the first eigenvectors of $H_{\overline{X}}$ and $H_{\overline{Y}}$, respectively, the reduced parameter and the reduced data are $\X_r=\overline\BasisX_r^T \overline\X$ and $\Y_s=\overline\BasisY_s^T \overline\Y$ which, using \eqref{eq:whitening}, are given by
\begin{align}
 \X_r = \BasisX_r^T \X &\qquad\text{where}\qquad \BasisX_r = \Cpr^{-1/2}\overline\BasisX_r  \label{eq:Xdecomp_Precond}\\
 \Y_s = \BasisY_s^T \Y &\qquad\text{where}\qquad \BasisY_s = \Cobs^{-1/2}\overline\BasisY_s .\label{eq:Ydecomp_Precond}
\end{align}
With the above definition, the matrices $\BasisX_r$ and $\BasisY_s$ have no longer orthogonal columns in the Euclidean sense, but they satisfy $\BasisX_r^T\Cpr\BasisX_r = \Id_r$ and $\BasisY_s^T \Cobs \BasisY_s = \Id_s$.
The error bound \eqref{eq:KL_UpperBound} thus becomes
\begin{equation} \label{eq:KL_UpperBound_whiten}
 \mathbb{E}\left[\KLDiv(\Post(\cdot,Y)||\OptPost(\cdot,Y))\right] \leq \overline{C}(\pi_{\overline{X},\overline{Y}})^2 \left(\Tr(\overline{\BasisX}_\perp^{T} H_{\overline{X}} \overline{\BasisX}_\perp) + \Tr(\overline{\BasisY}_\perp^{T} H_{\overline{\Y}} \overline{\BasisY}_\perp) \right) .
\end{equation}

\begin{remark}[Generalized eigenvalue problems]\label{rmk:GeneralizedEigendecomp}
 Let 
\begin{align*}
 \mathcal{H}_{\overline{X}} = \int  \nabla G(x)^{T} \Cobs^{-1} \nabla G(x) \Prior(x)\d\x 
 \quad\text{and}\quad
 \mathcal{H}_{\overline{Y}} = \int  \nabla G(x) \Cpr \nabla G(x)^T  \Prior(x)\d\x ,
\end{align*}
so that $H_{\overline\X} = \Cpr^{1/2} \mathcal{H}_{\overline{X}} \Cpr^{1/2} $ and $H_{\overline\Y} =  \Cobs^{-1/2}\mathcal{H}_{\overline{Y}} \Cobs^{-1/2}$. 
With the change of variables $u_i = \Cpr^{-1/2}\overline{u}_i$ and $v_i = \Cobs^{-1/2}\overline{v}_i$, 
the simple eigenvalue problems $H_{\overline\X}\overline{u}_i = \lambda_i(H_{\overline\X})\overline{u}_i$ and $H_{\overline\Y}\overline{v}_i = \lambda_i(H_{\overline\Y})\overline{v}_i$ are equivalent to the \emph{generalized} eigenvalue problems
\begin{align}
 \mathcal{H}_{\overline{X}} w_i &= \lambda_i(\mathcal{H}_{\overline{X}},\Cpr^{-1}) \Cpr^{-1} w_i, \qquad u_i =\Cpr^{-1}w_i , \label{eq:Gaussian_eig_param} \\
 \mathcal{H}_{\overline{Y}} v_i &= \lambda_i(\mathcal{H}_{\overline{Y}},\Cobs) \Cobs v_i , \label{eq:Gaussian_eig_data}
\end{align}
where $\lambda_i(\mathcal{H}_{\overline{X}},\Cpr^{-1}) = \lambda_i(H_{\overline\X})$ and $\lambda_i(\mathcal{H}_{\overline{Y}},\Cobs) = \lambda_i(H_{\overline\Y})$\footnote{We use the notation $\lambda_i(A,B)$ to denote the generalized eigenvalues of the matrix pencil $(A,B)$.}.  
 We note that $\mathcal{H}_{\overline{X}}$ is the same diagnostic matrix introduced in \citet[Section 4]{cui2020data} in the setting of Gaussian error models. $\mathcal{H}_{\overline{X}}$ is also similar to the diagnostic proposed in~\citet{cui2014likelihood} for finding the likelihood-informed subspace, with the key difference being that $\mathcal{H}_{\overline{X}}$ integrates over the prior distribution instead of the posterior. 
\end{remark}

\begin{remark}
Instead of the decompositions \eqref{eq:Xdecomp} and \eqref{eq:Ydecomp}, the proposed change of variables \eqref{eq:whitening} yields a decomposition of $\X$ and $\Y$ of the form
\begin{align*}
 \X &= \Cpr^{1/2}\left(\overline\BasisX_r \X_r  + \overline\BasisX_\perp X_{\perp}\right) &\text{where \ \ } \left\{\begin{array}{l} \X_r = \overline\BasisX_r^{T}\Cpr^{-1/2}\X \\ \X_\perp = \overline\BasisX_\perp^{T}\Cpr^{-1/2}\X \end{array} \right. , \\
 \Y &= \Cobs^{1/2}\left(\overline\BasisY_s \Y_s  + \overline\BasisY_\perp \Y_{\perp}\right) &\text{where \ \ } \left\{\begin{array}{l} \Y_s = \overline\BasisY_s^{T}\Cobs^{-1/2}\Y \\ \Y_\perp = \overline\BasisY_\perp^{T}\Cobs^{-1/2}\Y \end{array} \right. .
\end{align*}

\end{remark}

\subsection{Linear--Gaussian setting} \label{subsec:LinearGaussian}
Now we consider the case where the forward model is linear, i.e., $\x\mapsto G\x$ where $G \in \R^{m \times d}$ is a matrix.
In this case, the diagnostic matrices ${H}_{\overline{X}}$ and ${H}_{\overline{Y}}$ are written as
\begin{align*}
  H_{\overline{X}} &= (\Cpr^{1/2}G^T\Cobs^{-1/2})(\Cobs^{-1/2}G\Cpr^{1/2}) \\
  H_{\overline{Y}} &= (\Cobs^{-1/2}G\Cpr^{1/2})(\Cpr^{1/2}G^T\Cobs^{-1/2}) .
\end{align*}
The eigendecompositions of $H_{\overline{X}}$ and $H_{\overline{Y}}$ can be obtained by computing the singular value decomposition (SVD) of the so-called ``whitened forward model'': 
\begin{equation} \label{eq:WhitenedForwardModel}
 \Cpr^{1/2}G^T\Cobs^{-1/2} = \sum_{i=1}^{\min\{d,m\}} \sigma_i \overline u_i \overline v_i^T.
\end{equation}
In particular, the non-zero eigenvalues of the parameter-space and data-space diagnostic matrices are now the same, i.e.,  $\lambda_i(H_{\overline{X}}) = \lambda_i(H_{\overline{Y}}) = \sigma_i^2$ for all $i\leq\min\{d,m\}$; any remaining eigenvalues, i.e., for $\min\{d,m\} < i \leq \max\{d,m\}$ are zero. %

The eigendecompositions of $H_{\overline{X}}$ and $H_{\overline{Y}}$ (or, equivalently, the generalized eigendecompositions \eqref{eq:Gaussian_eig_param} and \eqref{eq:Gaussian_eig_data} of $\mathcal{H}_{\overline{X}}$ and $\mathcal{H}_{\overline{Y}}$, as in Remark \ref{rmk:GeneralizedEigendecomp}) have been used to reduce the parameter and data dimensions in linear--Gaussian inverse problems. \citet{spantini2015optimal} solve~\eqref{eq:Gaussian_eig_param} to approximate the posterior covariance as a low-rank update of the prior covariance. These eigenvectors also yield a projector for the parameter which matches the one derived above. Furthermore, Algorithm 1 in~\citet{spantini2015optimal} solves both eigenvalue problems, \eqref{eq:Gaussian_eig_param} and \eqref{eq:Gaussian_eig_data}, to derive an approximation to the posterior mean (which minimizes a Bayes risk with weighted squared error loss) as a linear projection of the data $\Y$. 
\citet{giraldi2018optimal} show the equivalence between the solution to~\eqref{eq:Gaussian_eig_data} and finding the vectors that solve $\max_{\BasisY_s} I(\BasisY_s^T\Y,\X)$, which Proposition~\ref{prop:PosteriorErr_MI} then shows is equivalent to minimizing the expected KL divergence from the true posterior, with reduced data. To minimize this expected KL divergence for linear inverse problems over the column vectors $\BasisY_s$, \citet{giraldi2018optimal} use Riemannian optimization algorithms on a Grassmannian manifold; they extend their approach to nonlinear forward models simply by using a Laplace approximation of the posterior.
Lastly, \citet{jagalur2021batch} derive mutual information bounds for coordinate selection of data in linear--Gaussian problems. These bounds are used to develop various greedy algorithms, with guarantees for cardinality-constrained optimization.

\subsection{Gap in the linear--Gaussian setting}

We analyze now the gap 
in \eqref{eq:KL_UpperBound_whiten} for a linear--Gaussian likelihood model with a Gaussian prior. 
We denote by $\sigma_i$ the $i$-th largest singular value of the whitened forward model $\Cpr^{1/2} G^T \Cobs^{-1/2}$ in~\eqref{eq:WhitenedForwardModel}. 
Using the closed-form expression for the mutual information of Gaussian variables (see Appendix~\ref{app:calculations}), %
we have
\begin{equation} \label{eq:MIGaussian}
\mathbb{E}\left[\KLDiv(\Post(\cdot|\Y)||\OptPost(\cdot|\Y))\right]
\overset{\eqref{eq:MIdiff_JointDimReduction}}{=}I(\X,\Y) - I(\X_r,\Y_s) = \frac{1}{2} \sum_{i > \min\{r,s\}}^{\min\{d,m\}} \log(1 + \sigma_i^2).
\end{equation}
In comparison, the upper bound in~\eqref{eq:KL_UpperBound_whiten} evaluated at the optimal rotation $U_\perp$ and $V_\perp$ is given by 
\begin{equation} \label{eq:UpperBoundGaussian}
    \overline{C}(\pi_{\overline{X},\overline{Y}})^2 \left(\sum_{i\,>\,r}^{\min\{d,m\}} \sigma_i^2 + \sum_{i\,>\,s}^{\min\{d,m\}} \sigma_i^2 \right),
\end{equation}
where the subspace log-Sobolev constant $\overline{C}(\pi_{\overline{X},\overline{Y}})$ can be bounded in terms of $\sigma_1$, as shown in Example~\ref{ex:gaussianJoint}.
Using a first-order Taylor expansion of $\log(1 + \sigma_i^2)$ as $\sigma_i\rightarrow 0$, the ratio between~\eqref{eq:MIGaussian} and~\eqref{eq:UpperBoundGaussian} satisfies
\begin{align}
    \frac{\mathbb{E}\left[\KLDiv(\Post(\cdot|\Y)||\OptPost(\cdot|\Y))\right]}{\overline{C}(\pi_{\overline{X},\overline{Y}})^2 \left(\sum_{i\,>\,r}^{\min\{d,m\}} \sigma_i^2 + \sum_{i\,>\,s}^{\min\{d,m\}} \sigma_i^2 \right)} 
    &\,= \frac{1}{2 \overline{C}(\pi_{\overline{X},\overline{Y}})^2}  \frac{\sum_{i > \min\{r,s\}}^{\min\{d,m\}}  \sigma_i^2 + \mathcal{O}(\sigma_i^4) }{  \left(\sum_{i\,>\,r}^{\min\{d,m\}} \sigma_i^2 + \sum_{i\,>\,s}^{\min\{d,m\}} \sigma_i^2 \right) } \label{eq:GapUpperBound} \\
    &\overset{r=s}{=} \frac{1}{4 \overline{C}(\pi_{\overline{X},\overline{Y}})^2} ( 1 + \mathcal{O}(\sigma_{r}^2) ) \label{eq:GapUpperBound_EqualReducedDim}.
\end{align}
In the limit of $\sigma_{r}\rightarrow0$, the above ratio converges to the constant $1/(4 \overline{C}(\pi_{\overline{X},\overline{Y}})^2)$. Thus, the expected KL divergence and its bound go to zero at the same rate. 
Let us remark that if either $r = d$ or $s = m$, i.e., when only the parameter or the data are reduced but not both, the ratio in~\eqref{eq:GapUpperBound} goes to $1/(2\overline{C}(\pi_{\overline{X},\overline{Y}})^2)$.

As a numerical illustration, we consider the linear inverse problem introduced in~\citet[Example 1]{spantini2015optimal} with identity forward model $G = I_d$ and $m = d = 50$. The prior covariance is constructed as $\Cpr = W D W^T$, where $W$ is one realization of a random unitary matrix drawn uniformly from the unitary group\footnote{This is typically done by computing the QR factorization of a random matrix with standard Gaussian entries.} 
and $D$ is a diagonal matrix with $D_{ii} = \lambda_0/i^\vartheta + \tau$, $\lambda_0 = 1$, $\vartheta = 2$ and $\tau = 10^{-6}$. We follow the same procedure to realize the observation noise covariance $\Cobs$ with $\lambda_0 = 500$ and $\vartheta = 1$.

Figure~\ref{fig:linear_joint_CMI} plots the expected KL divergence (of the approximate posterior from the exact posterior) for the optimal parameter and data projectors at any given pair of reduced dimensions $(r,s)$. This quantity is computed using  analytical expressions for the mutual information of Gaussian random vectors, as above.
Figure~\ref{fig:linear_joint_CMI_upperbound} plots the corresponding value of the upper bound in \eqref{eq:KL_UpperBound_whiten},
up to the unknown log-Sobolev constant, evaluated at the optimal projectors.
Figure~\ref{fig:linear_joint_gap} then plots the ratio between this upper bound and the posterior approximation error, confirming the analytical results derived above in that the ratio approaches $1/4$ for $r=s$ and a maximum value of $1/2$ for $r=d$ or $s=m$. 
For any tolerance level $\epsilon$, we can (in Figures~\ref{fig:linear_joint_CMI} and \ref{fig:linear_joint_CMI_upperbound}) observe the Pareto front of reduced dimensions that yield the same approximation error. The dashed lines in Figure~\ref{fig:linear_joint_CMI_upperbound} highlight reduced dimensions that solve~\eqref{eq:FindingReducedDim} for a linear cost function $c(r,s) = \alpha_\X r + \alpha_\Y s$ with different weights $\alpha_X \in \{0.2,0.5,0.8\}$ and $\alpha_\Y = 1 - \alpha_\X$, for five different values of the tolerance. We see that these choices for the weight trade off the cost of keeping the parameters versus the data. %

\begin{figure}[!ht]
     \centering
     \begin{subfigure}[t]{0.32\textwidth}
         \centering
         \includegraphics[width=\textwidth]{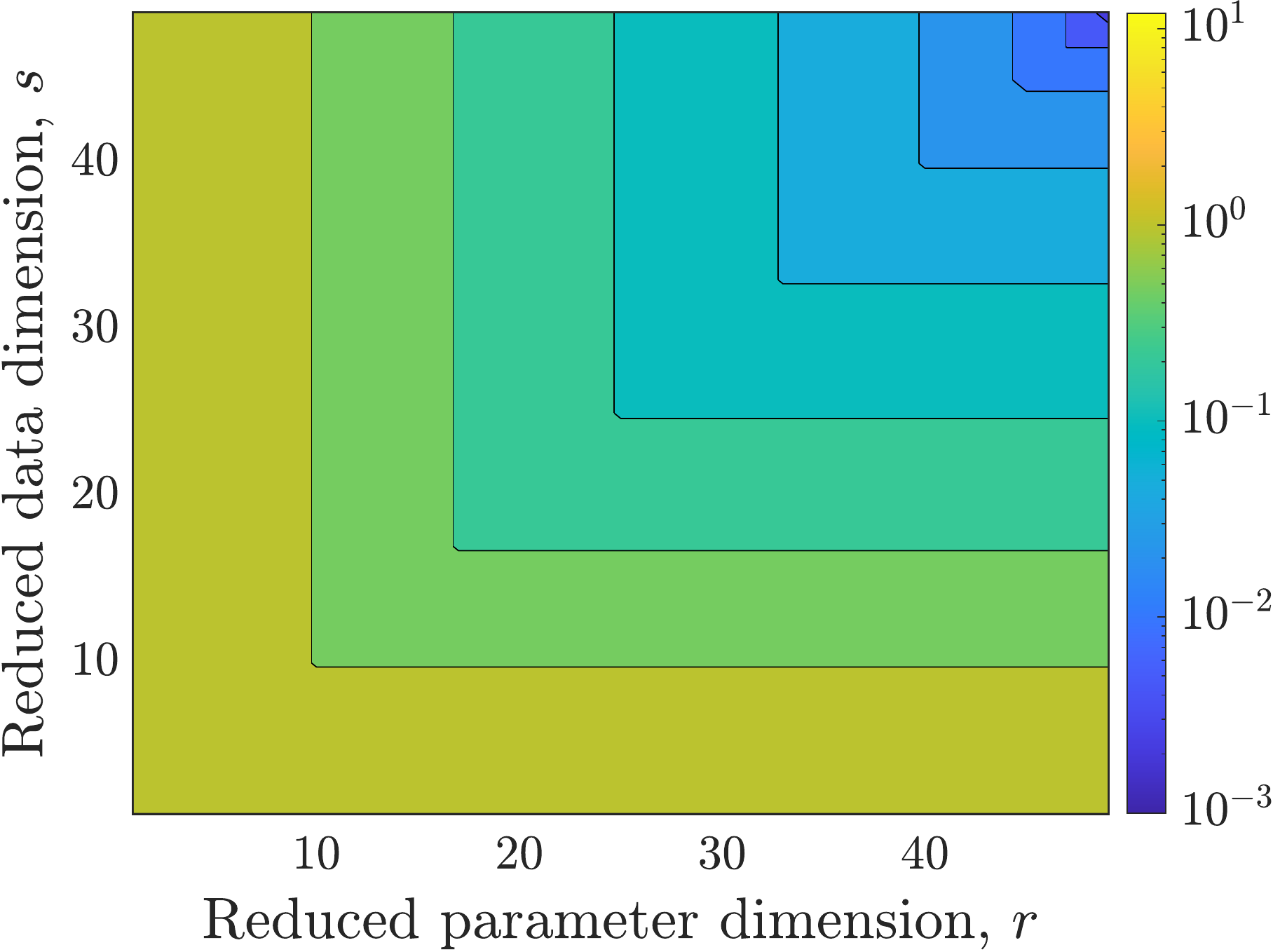}
         \caption{}
         \label{fig:linear_joint_CMI}
     \end{subfigure}
     \hfill
     \begin{subfigure}[t]{0.32\textwidth}
         \centering
         \includegraphics[width=\textwidth]{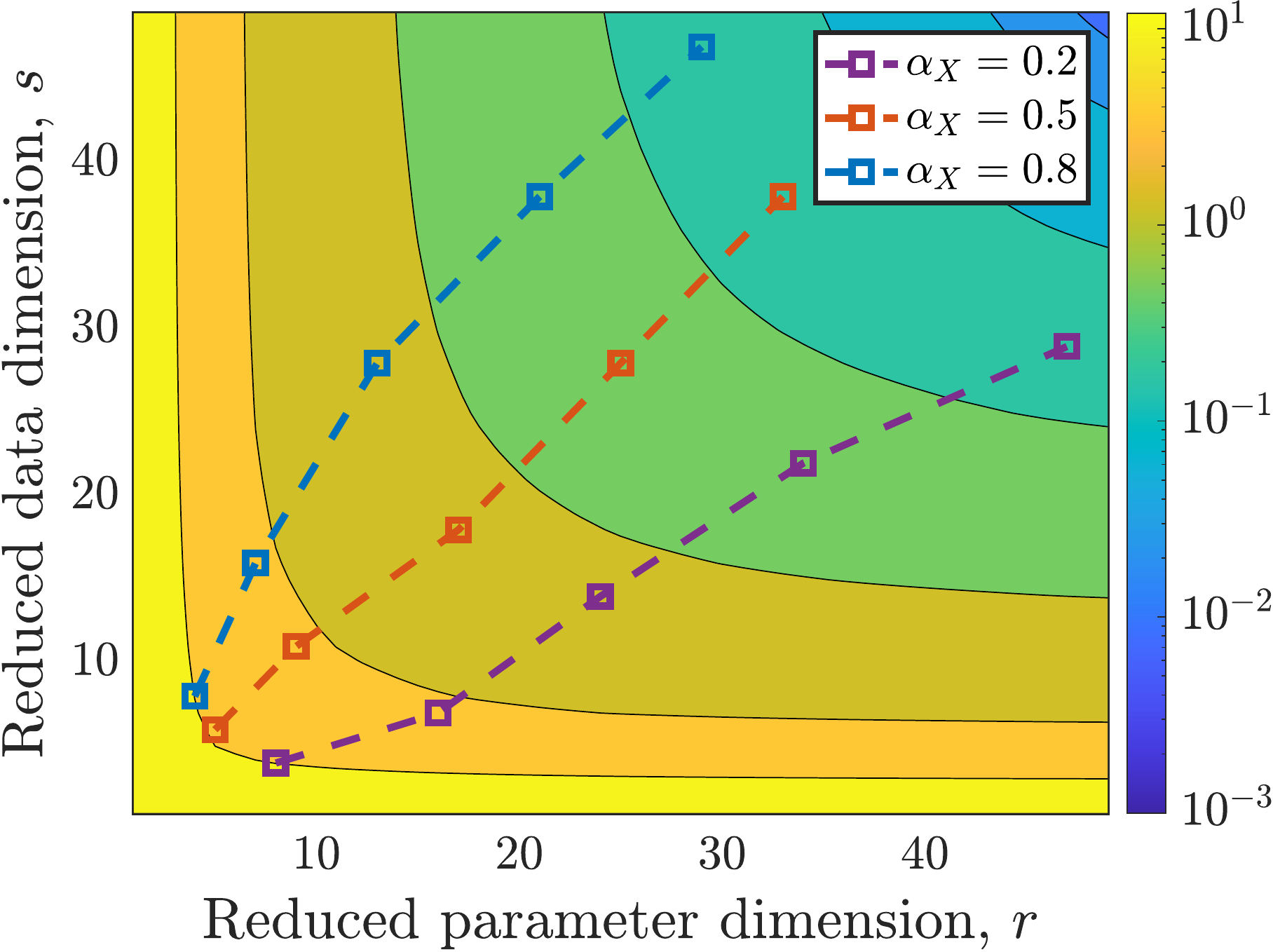}
         \caption{}
         \label{fig:linear_joint_CMI_upperbound}
     \end{subfigure}
     \hfill
     \begin{subfigure}[t]{0.32\textwidth}
         \centering
         \includegraphics[width=\textwidth]{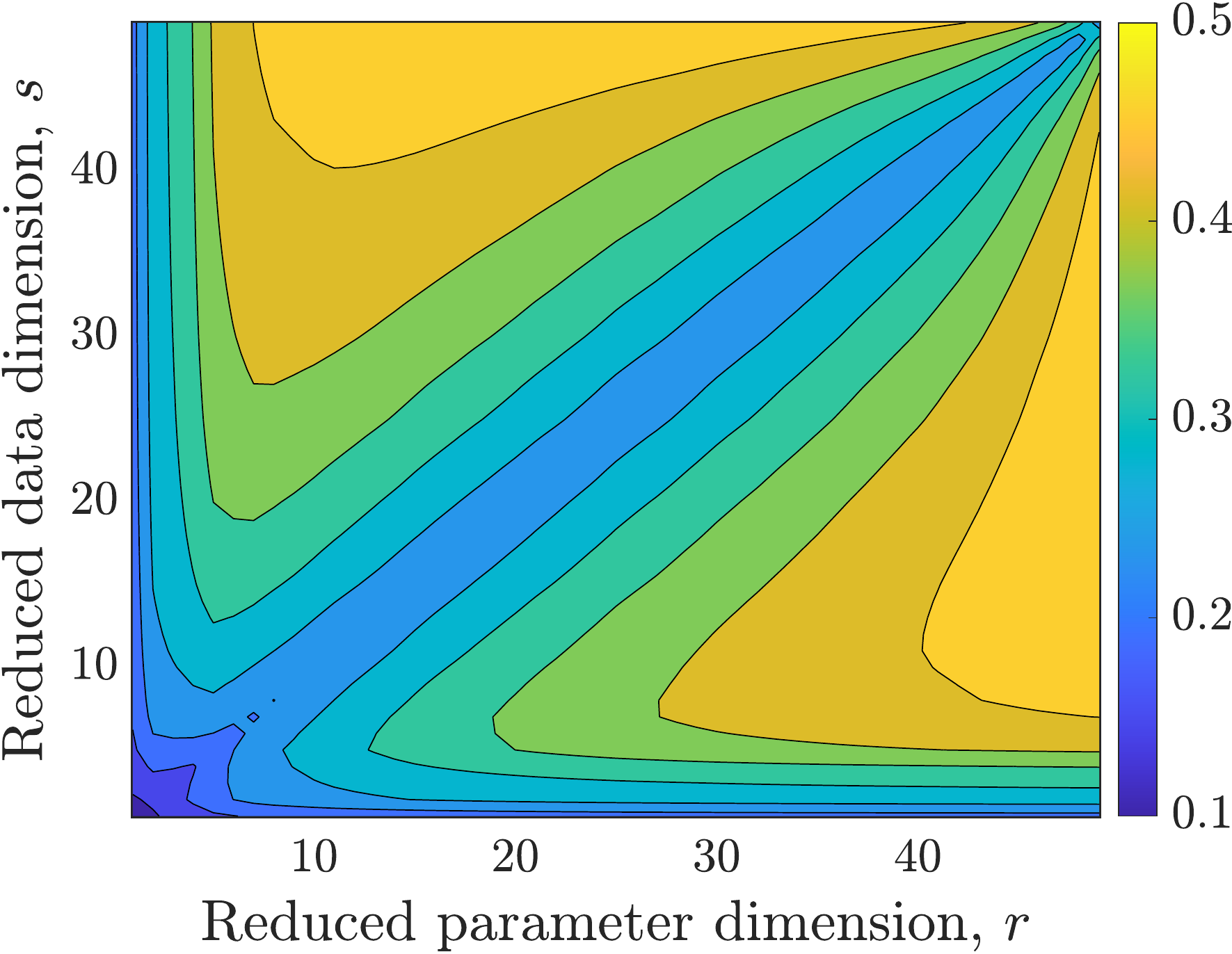}
         \caption{}
         \label{fig:linear_joint_gap}
     \end{subfigure}
     \caption{(a) Posterior approximation error $\mathbb{E}_{\Y}[\KLDiv(\Post(\cdot|\Y)||\OptPost(\cdot|\Y))]$ as a function of reduced dimensions $(r,s)$; (b) Tractable part of the upper bound for the posterior approximation error; (c) Gap in the upper bound, up to the constant $\overline{C}(\pi_{\overline{X},\overline{Y}})^2$, approaches $1/4$ for $r = s$ and $1/2$ for $r = d$ or $s = m$.} %
\end{figure}

\section{Comparisons to PCA and CCA} \label{sec:OtherDimRedMethods}

Two popular methods for linear dimension reduction are principal component analysis (PCA) and canonical correlation analysis (CCA).

PCA consists in reducing the dimension of a mean-zero random vector $\X$ by minimizing the $L^2$ error $\E[ \| \X - \BasisX_r \BasisX_r^T \X \|^2 ]$ over matrices $\BasisX_r\in\R^{d\times r}$ with orthogonal columns~\citep{hotelling1933analysis, jolliffe2002principal}. The solution is $\BasisX_r = [u_1^\PCA, \dots, u_r^\PCA]$ where $u_i^\PCA$ are the leading eigenvectors of the covariance matrix $\Cov(\X)$. That is,
\begin{equation}
 \Cov(\X) u_i^\PCA = \lambda_i (\Cov(\X)) u_i^\PCA.   
\end{equation}
The same procedure can be applied to reduce the dimension of $\Y$, which yields $\BasisY_s = [v_1,\hdots,v_s]$ where $v_i^\PCA$ are the leading eigenvectors of the covariance matrix $\Cov(\Y)$. That is,
\begin{equation}
\Cov(\Y) v_i^\PCA = \lambda_i(\Cov(\Y)) v_i^\PCA.    
\end{equation}
There are two main drawbacks of using this dimension reduction method for Bayesian inference problems. 
The first is that PCA is an \emph{unsupervised} dimension reduction method. That is, the directions identified by PCA are meant to reconstruct $\X$ and $\Y$ marginally, but it does not account for the dependence between $\X$ and $\Y$.
Second, an accurate low-dimensional PCA approximation depends on fast decay of the eigenvalues of the covariances $\Cov(\X)$ and $\Cov(\Y)$. In many inference problems, however, we can have low-dimensional structure without having sharp decay in the spectra of these covariances; cf.\ Example~\ref{ex:gaussianJoint} where $\Cov(\X)=\Id_d$ and $\Cov(\Y) = \Id_m + GG^T$. 

Alternatively, CCA seeks linear combinations of $\X$ and $\Y$ that are maximally correlated~\citep{hotelling1992relations, hardoon2004canonical}. That is, CCA solves
\begin{equation} \label{eq:CCA_maxcorr}
    (\BasisX_r^{\CCA},\BasisY_r^{\CCA}) = 
    \argmax_{\substack{\BasisX_r^T \Cov(\X) \BasisX_r = \Id_r \\ \BasisY_r^T \Cov(\Y) \BasisY_r = \Id_r}} \Tr (U_r^T \Cov(\X,\Y) V_r),
\end{equation}
where $\Cov(\X,\Y)$ is the cross-covariance of $\X$ and $\Y$, and $r \leq \min\{d,m\}$.
The vectors $(\BasisX_r^{\CCA})^T\X$ and $(\BasisY_r^{\CCA})^T\Y$ are called the \emph{pairs of canonical variables}. %
It can be shown that $\BasisX_r,\BasisY_r$ can be found by solving the generalized eigenvalue problems
\begin{align}
    \Cov(\X,\Y)\Cov(\Y)^{-1}\Cov(\Y,\X)u_i^\CCA &= \rho_i \Cov(\X)u_i^\CCA \label{eq:CCAeigU} \\
    \Cov(\Y,\X)\Cov(\X)^{-1}\Cov(\X,\Y)v_i^\CCA &= \rho_i \Cov(\Y)v_i^\CCA, \label{eq:CCAeigV}
\end{align} 
where the eigenvectors are ordered based on a descending order for the eigenvalues $\rho_i^2 \in [-1,1]$.

The next proposition shows that, for linear--Gaussian likelihood models (and potentially a non-Gaussian prior), 
our dimension reduction approach using whitening (see Section~\ref{sec:WhiteningGaussianLik}) is the same as CCA. The proof of this result is provided in Appendix~\ref{app:proofs}. 
\begin{proposition} \label{prop:CCA_connection}
Let $\Y = G\X + \varepsilon$, where $G \in \R^{m \times d}$ and $\varepsilon $ is independent of $\X$, with mean $\mathbb{E}[\varepsilon] = 0$ and covariance $\Cov(\varepsilon) = \Cobs$. Then, the solution to~\eqref{eq:CCA_maxcorr} is given by
$$
 \BasisX_r^{CCA} = \Cpr^{-1/2}\overline\BasisX_r,
 \qquad \text{and}\qquad
 \BasisY_r^{CCA} = \Cobs^{-1/2}\overline\BasisY_r ,
$$
where $\overline\BasisX_r$ and $\overline\BasisY_r$ are the matrices containing the first $r$ eigenvectors of the diagnostic matrices $H_{\overline\X},H_{\overline\Y}$ defined in \eqref{eq:Gaussian_Hx} and~\eqref{eq:Gaussian_Hy}, respectively.
Furthermore, we have $\rho_i = \lambda_i/(1+\lambda_i)$ where $\lambda_i=\lambda_i(H_{\overline \X})=\lambda_i(H_{\overline \Y})$.
\end{proposition}
The method proposed in the present paper can thus be seen as a generalization of CCA to nonlinear models. To do so, we use gradient information from the log-likelihood, 
whereas CCA uses only covariance information. %
We will show in Section~\ref{sec:NumericalExp} that our approach yields more accurate posterior approximations than CCA for the same reduced dimensions. Lastly, we note that CCA can only identify subspaces of the parameter and data of the \emph{same} dimension, i.e., $r = s$. In contrast, our proposed approach can trade off these two dimensions while meeting a desired error tolerance (see Section~\ref{sec:SelectingTheDimensions}). %

\section{Algorithms} \label{sec:Algorithms}
We now present algorithms to identify and exploit the low-dimensional subspaces for the informed parameters and informative data. 
Algorithm~\ref{alg:computeBasis} shows how to identify those subspaces using Monte Carlo estimation of the diagnostic matrices $H_\X$ and $H_\Y$. To do this, we assume we have access to the mixed partial derivatives of the log-likelihood function $\log\Like$, and that we can sample from the joint density $\Joint$. A sample $(\xs^i,\ys^i)$ from the joint density $\Joint$ is typically obtained by first sampling $\xs^i \sim \Prior$ and then sampling $\ys^i \sim \Like(\cdot|\xs^i)$. 

Once the matrices $\BasisX_r$ and $\BasisY_s$ are identified, sampling from the approximate posterior $\OptPost$ in~\eqref{eq:OptPost} requires samples from the reduced posterior $\pi_{\X_r|\Y_s}$ and from the conditional prior $\pi_{\X_\perp|\X_r}$. More specifically, given a realization of the data $\y$, we need to
\begin{enumerate}
    \item Project the data $\y_s=\BasisY_s^T \y$,
    \item Draw a sample from the reduced posterior $\widetilde\X_r^i\sim \pi_{\X_r|\Y_s=\y_s}$,
    \item Draw a sample from the conditional prior $\widetilde\X_\perp^i\sim \pi_{\X_\perp|\X_r=\widetilde\X_r^i}$,
    \item Assemble $\widetilde\X^i= \BasisX_r \widetilde\X_r^i + \BasisX_\perp \widetilde\X_\perp^i$.
\end{enumerate}
Thus, by construction we have $\widetilde\X^i\sim\OptPost$.
Step 2 of the above procedure (i.e., drawing samples $\widetilde\X_r^i\sim \pi_{\X_r|\Y_s=\y_s}$) is the key challenge, which requires using a dedicated inference algorithm. In the following subsections we propose two classes of inference algorithms. The first is based on evaluations of the likelihood function and the prior density, whereas the second only requires samples from the joint density $\Joint$.

\begin{algorithm}[!ht]
    \caption{Identify decomposition of parameter and data spaces \label{alg:computeBasis}}
    \DontPrintSemicolon
    \SetKwInOut{Input}{Input}\SetKwInOut{Output}{Output}
\BlankLine
    \Input{Prior density $\pi_{\X}$, Likelihood $\pi_{\Y|\X}$, Sample size $n$, Reduced dimensions $r,s$}
    \Output{Matrices with Orthonormal columns $\BasisX_r,\BasisY_s$}
    Draw $n$ i.i.d.\thinspace samples $\{\xs^{i}\}_{i=1}^{n} \sim \pi_{\X}$ and then $\{\ys^{i}\}_{i=1}^{n} \sim \pi_{\Y|\xs^{i}}$\;
    Compute $\nabla_{\X}\nabla_{\Y} \log \pi_{\Y|\X}(\ys^i|\xs^i)$ for $i=1,\dots,n$\;
    Assemble the Monte Carlo estimates 
    \begin{align*}
    \widehat H_{\X}&=\frac{1}{n}\sum_{i=1}^n \left(\nabla_{\X}\nabla_{\Y}\log \pi_{\Y|\X}(\ys^i|\xs^i)\right)^T \left(\nabla_{\X}\nabla_{\Y}\log \pi_{\Y|\X}(\ys^i|\xs^i)\right) \\
    \widehat H_{\Y}&=\frac{1}{n}\sum_{i=1}^n \left(\nabla_{\X}\nabla_{\Y}\log \pi_{\Y|\X}(\ys^i|\xs^i)\right) \left(\nabla_{\X}\nabla_{\Y}\log \pi_{\Y|\X}(\ys^i|\xs^i)\right)^T
    \end{align*}
    
    \textbf{Optimal rotation}: Solve eigenvalue problems $\widehat H_{\X}u_{i} = \lambda_i(\widehat H_{\X})u_{i}$, $\widehat H_{\Y}v_{i} = \lambda_i(\widehat H_{\Y})v_{i}$ for the eigenvectors corresponding to the $r,s$ leading eigenvalues, \textbf{or} \;
    \textbf{Optimal permutation}: Identify indices of largest diagonal entries in $\widehat H_\X$ and $\widehat H_\Y$ and set $u_1,\dots,u_r$ and $v_1,\dots,v_s$ to those canonical unit vectors\;
    Assemble $U_r = [u_1,\dots,u_r]$, $V_s = [v_1,\dots,v_s]$\;
\end{algorithm}

\subsection{Inference methods based on likelihood evaluations} \label{subsec:LikelihoodBasedMethods}

Markov chain Monte Carlo (MCMC) algorithms are popular methods for sampling from posterior distributions. They require the ability to evaluate the posterior density (up to a normalizing constant), and hence to evaluate the likelihood function (and in general also the prior density), to accept or reject a proposed move. In the present setting, the likelihood function is the reduced likelihood given by
\begin{equation} \label{eq:RedLikelihood}
    \pi_{\Y_s|\X_r}(\y_s|\x_r) = \int \pi_{\Y_s|\X}(\y_s|\BasisX_r\x_r + \BasisX_\perp\x_\perp)\pi_{\X_\perp|\X_r}(\x_\perp|\x_r) \d \x_\perp,
\end{equation}
where the data-marginalized likelihood $\pi_{\Y_s|\X}$ above %
is given by
\begin{equation}\label{eq:tmp385}
    \pi_{\Y_s|\X}(\y_s|\x) = \int \pi_{\Y|\X}(\BasisY_s\y_s + \BasisY_\perp\y_\perp|\x) \d\y_\perp.
\end{equation}
The next example shows that when the likelihood is Gaussian, one can analytically compute the integral in~\eqref{eq:tmp385} so that $\pi_{\Y_s|\X}$ is accessible in closed form. 
\begin{example} \label{ex:AnalyticalMargLikelihood}
For the (whitened) Gaussian-likelihood model in Section~\ref{sec:WhiteningGaussianLik},
we have the rotated data model 
\begin{align*}
    \Y_s = \BasisY_s^T \Cobs^{-1/2}\Y = \BasisY_s^T \Cobs^{-1/2} G(\x) + \BasisY_s^T \Cobs^{-1/2} \varepsilon\\
    \Y_\perp = \BasisY_\perp^T \Cobs^{-1/2}\Y = \BasisY_\perp^T \Cobs^{-1/2} G(\x) + \BasisY_\perp^T \Cobs^{-1/2} \varepsilon.
\end{align*}
Given that the observational noise components $V_s^T \Cobs^{-1/2} \varepsilon$ and $V_\perp^T \Cobs^{-1/2} \varepsilon$ are independent and have identity covariance, the data-marginalized likelihood is Gaussian with the form
$$\pi_{\Y_s|\X}(\y_s|\x) = (2\pi)^{-s/2}\exp \left(-\frac{1}{2}\|\y_s - \BasisY_s^T\Cobs^{-1/2}G(\x)\|_2^2 \right).$$
\end{example}

While the integral in \eqref{eq:tmp385} can be computed analytically, there is in general no closed form expression for the integral in~\eqref{eq:RedLikelihood}. Thus, the reduced likelihood $\pi_{\Y_s|\X_r}$ needs to be estimated numerically.
We consider here the Monte-Carlo estimator 
\begin{equation} \label{eq:MCestimator_RedLikelihood}
    \widehat{\pi}_{\Y_s|\X_r}(\y_s|\x_r) =\frac{1}{\ell} \sum_{i=1}^{\ell} \pi_{\Y_s|\X}(\y_s|\BasisX_r\x_r + \BasisX_\perp\xs_\perp^i),
    \quad \xs_\perp^i \sim \pi_{\X_\perp|\X_r}(\cdot|\x_r) .
\end{equation}
We refer to~\citet{cui2020data} for an intensive discussion on different sampling strategies, and on the impact of the sample size $\ell$ versus the truncated dimension $r$. As shown in~\citet{zahm2018certified,cui2021unified}, the variance of the estimator \eqref{eq:MCestimator_RedLikelihood} is low when the error bound $\sum_{i=r+1}^d \lambda_i(H_\X)$ is small. In practice, it is sufficient to use few samples (e.g., $\ell = 1$) or even deterministic approximations (e.g., by setting $\xs_\perp^i$ to the conditional prior mean). More interestingly, taking the perspective of pseudo-marginal MCMC \citep{andrieu2009pseudo}, it is shown in \citet{cui2020data} that redrawing fresh samples $\xs_\perp^i$ in \eqref{eq:MCestimator_RedLikelihood} at each MCMC iteration permits sampling from the \emph{exact} reduced posterior.

\subsection{Inference methods based on joint samples} \label{subsec:SampleBasedMethods}

Transportation of measure underpins another broad class of algorithms for generating conditional samples \citep{Marzouk2016, kovachki2020conditional}. These methods require having access to samples from the joint distribution $\Joint$ in order to construct an invertible map (e.g., using invertible neural networks as in~\citet{radev2020bayesflow} or polynomial expansions as in~\citet{baptista2020adaptive}) that transforms samples from the joint distribution to samples from the standard normal distribution. This map is then used to draw samples from the conditional distribution $\pi_{\X \vert \Y=y}$ for any value $y$, thereby amortizing the cost of inference for multiple realizations of the data. 

Recalling the overall scheme presented at the start of Section~\ref{sec:Algorithms}, we need to sample from the reduced posterior $\pi_{\X_r|\Y_s=\y_s}$ for some $y_s$. We begin by considering the reduced joint distribution $\pi_{X_r,Y_s}$. Samples from $\pi_{X_r,Y_s}$ are obtained by projecting samples $(X^i,Y^i)\sim\Joint$ as follows: $(X_r^i,Y_s^i)=(U_r^T X^i,V_s^T Y^i)$.
Using these samples, we then build a (block)-triangular map $S\colon\R^{r+s}\rightarrow\R^{r+s}$ such that
$$
 S(Y_s,X_r)=
 \left(\begin{array}{l}
  S^\mathcal{Y}(Y_s) \\ S^\mathcal{X}(Y_s,X_r)
 \end{array}\right)
 \sim\mathcal{N}(0,\Id_{r+s}) ,
$$
where both $y_s\mapsto S^\mathcal{Y}(y_s)$ and $x_r\mapsto S^\mathcal{X}(y_s,x_r)$ are invertible functions. Once $S$ is built, sampling from $\pi_{\X_r|\Y_s=y_s}$ requires solving the equation $S^\mathcal{X}(y_s,\widetilde\X_r) = Z_r$ for $\widetilde\X_r\in\R^r$, where $Z_r$ is a sample from $\mathcal{N}(0,\Id_r)$. By construction, we have $\widetilde\X_r\sim \pi_{X_r|Y_s=y_s}$; see \citet{Marzouk2016} for a proof. Furthermore, the map $S^\mathcal{X}$ enables evaluations of the conditional density via the change of variables formula, $\pi_{\X_r|\Y_s}(\x_r|\y_s) = \eta \circ S^\mathcal{X}(\y_s,\x_r) |\nabla_{\X_r} S^\mathcal{X}(\y_s,\x_r)|$ where $\eta$ denotes the density of the standard Gaussian distribution $\mathcal{N}(0,I_r)$. Let us remark that this procedure does not utilize the first map component $S^\mathcal{Y}$, and so it is unnecessary to construct it in practice. 

In this setting, reducing the dimensions of both the parameter and the data alleviates the computational burden of the map construction: $S^\mathcal{X}$ becomes a function of $s+r$ variables, rather than of $m+d$ variables. It is also worth noting that this reduced-dimensional approach to amortized inference is feasible only when the projection $V_s$ is independent of the data realization $y$, which is the case in our approach. %

\section{Numerical experiments} \label{sec:NumericalExp}
Code to reproduce the following numerical experiments is freely available at \url{www.github.com/baptistar/BayesianDimRed}.

\subsection{Linear elasticity inverse problem} \label{sec:Exp_EllipticPDE}

Our first numerical example is to infer the inhomogeneous Young's modulus of a (wrench-shaped) physical body $\mathcal{D}\subset\R^2$ given some measurements of the displacement on its boundary~\citep{lam2020multifidelity, smetana2020randomized}. This is a challenging inverse problem as both the Young's modulus and observed displacements are spatially distributed quantities and hence high-dimensional vectors after discretization; moreover, they are indirectly related via a partial differential equation that induces a nonlinear forward model.

Let $u\colon \mathcal{D} \rightarrow \R^2$ represent the displacement field given an external force $f$ applied on a subset of $\partial\mathcal{D}$. The displacement field $u$ satisfies the coupled elliptic PDE $\text{div}(K : \epsilon(u)) = 0$ everywhere on $\mathcal{D}$, where $\epsilon(u) = \frac{1}{2}(\nabla u + \nabla u^T)$ is the strain field and $K$ is the Hooke tensor, such that
\begin{equation}
    K : \epsilon(u) \coloneqq \frac{E}{1 + \nu}\epsilon(u) + \frac{\nu E}{1-\nu^2}\Tr(\epsilon(u))\Id_2.
\end{equation}
Here, $\nu = 0.3$ is Poisson's ratio and $E\colon \mathcal{D} \rightarrow \R_{>0}$ is the Young's modulus. The displacement field is also subject to a Dirichlet boundary condition, i.e., $u = 0$ on the right hand side of the wrench; see the dashed lines in Figure~\ref{fig:wrench_realizations_a}.
We model the Young's modulus field with a log-normal prior, i.e., $\log E \sim \mathcal{N}(0,C)$ where $C(s,s') = \varsigma^2 \exp(-\|s-s'\|_2^2/\ell^2)$ is a squared exponential covariance kernel on $\mathcal{D} \times \mathcal{D}$ with correlation length $\ell = 1$ and marginal variance $\varsigma^2 = 1$. %

To solve the PDE numerically, we apply the finite element method \citep{zienkiewicz2000finite}. We first discretize the domain using a mesh with $925$ elements and we approximate the stochastic field $\log(E)$ with a piecewise constant field whose values are gathered in a random vector $X$ of dimension $d = 925$. We denote by $u^h(X)$ the Galerkin projection of $u(X)$ onto the space of piecewise affine functions. 
We then extract the vertical displacements of $u^h(X)$ at the $m = 48$ nodes located along the line where the force is applied; see Figure~\ref{fig:wrench_realizations_b}. 
Denoting the corresponding (linear) extraction operator by $L$, the forward model $G\colon\R^{925}\rightarrow\R^{48}$ is written as $G(\X) = Lu^h(\X)$.
The data $\Y=G(\X)+\varepsilon$ are perturbed with a zero-mean Gaussian noise $\varepsilon$ which is independent of $X$. %
The covariance of $\varepsilon$ is defined as $\Cobs=LR^{-1}L^T$, where $R$ is the Riesz map associated with the $H^1(\mathcal{D)}$-norm such that $\|u^h\|_{R}^2 = \int_\mathcal{D} (u^h(s))^2+\|\nabla u^h(s)\|^2 \d s$.
This way, the norm $\|\cdot\|_{\Cobs^{-1}}$ corresponds to the standard trace norm on $H^{1/2}(\partial\mathcal{D})$; see~\citep[Chapter 5]{zahm2015model}. 
This example falls into the framework of Gaussian error models, discussed in Section~\ref{sec:GaussianLik}, with a Gaussian prior distribution.

\begin{figure}[!ht]
    \centering
    \begin{subfigure}[c]{0.45\textwidth}
     \centering
     \includegraphics[ width=\textwidth]{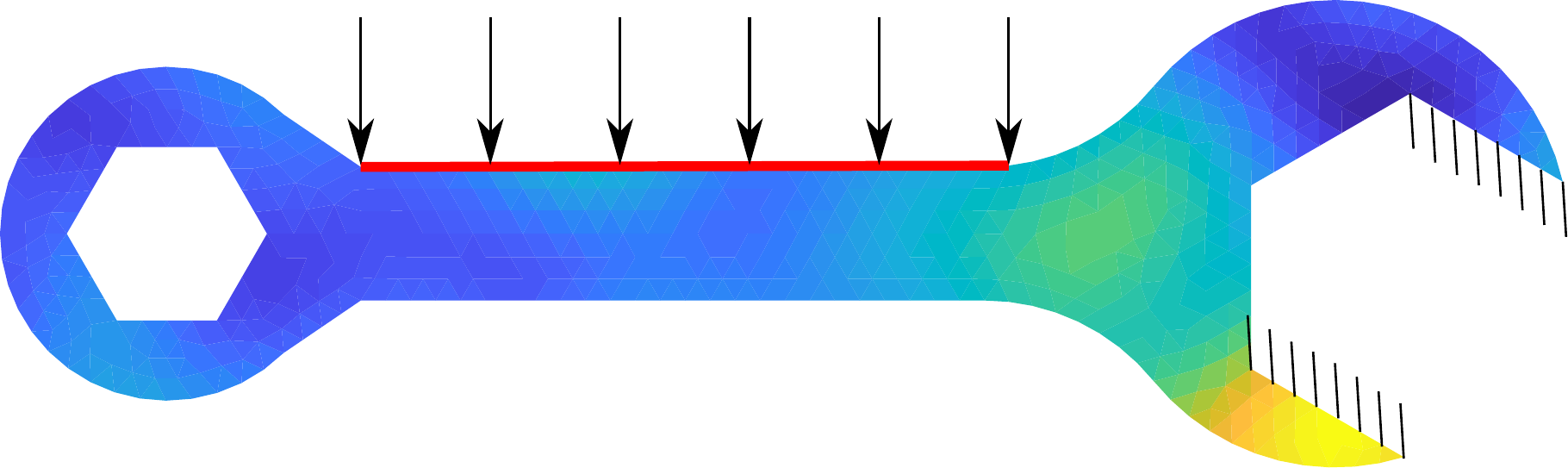}
     \caption{}
     \label{fig:wrench_realizations_a}
    \end{subfigure}
    \hspace{1cm}
    \begin{subfigure}[c]{0.45\textwidth}
     \centering
     \includegraphics[ width=\textwidth]{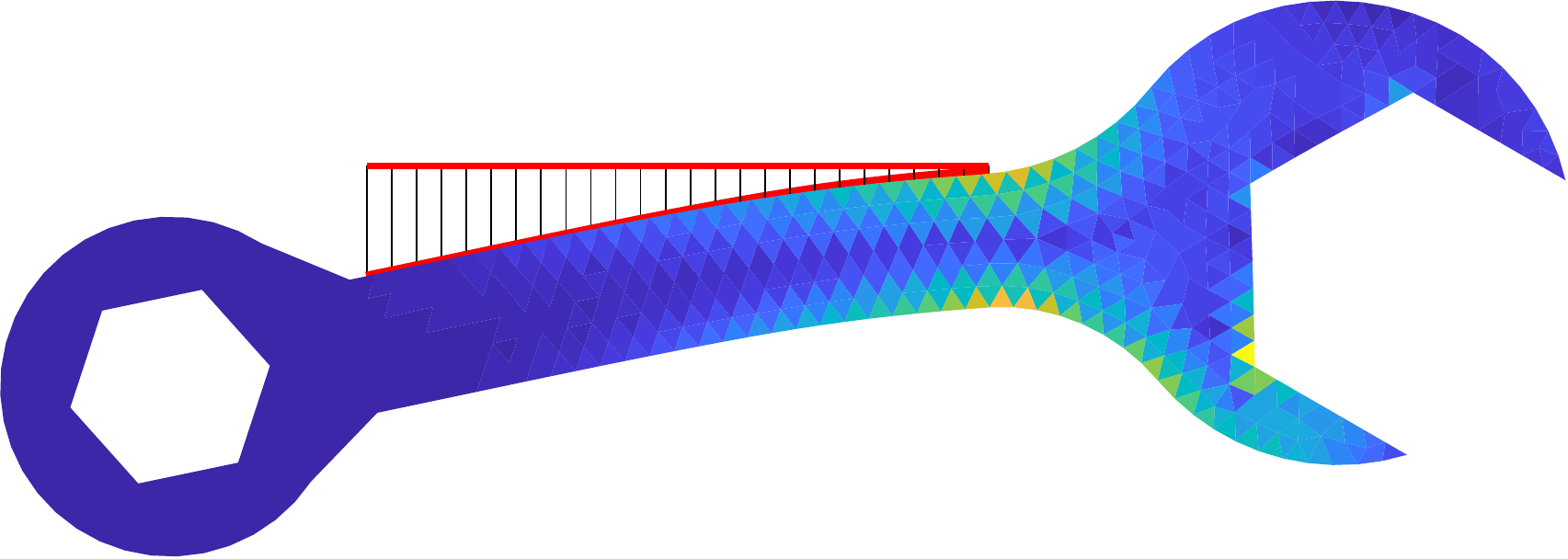}
     \caption{}
     \label{fig:wrench_realizations_b}
    \end{subfigure}
    \caption{Settings of the linear elasticity example: (a) A realization of the Young's modulus field, i.e., the parameter. The black arrows represent the force applied to the body, and the dashed lines represent the imposed boundary condition. (b) The von-Mises stress \citep{zienkiewicz2000finite} of the displacement field $u$ and the observed data given by the vertical displacement along the red line.}  \label{fig:wrench_realizations}
\end{figure} %

We compute 500 realizations of the gradients of the forward model and use these to estimate the matrices $H_{\overline{\X}}$ and $H_{\overline{\Y}}$ in~\eqref{eq:Gaussian_Hx} and~\eqref{eq:Gaussian_Hy}, using the whitening transformation; see Section \ref{sec:GaussianLik}. Figure~\ref{fig:wrench_param_eigs} plots sums of the trailing eigenvalues of $H_{\overline{\X}}$ and $H_{\overline{\Y}}$ for the parameter and data spaces, respectively; both are labeled as CMI in the plots (because our approach minimizes bounds for the conditional mutual information). The two sums of trailing eigenvalues correspond to the two terms in the upper bound for the expected KL divergence in~\eqref{eq:KL_UpperBound}. Fast decay of these eigenvalue sums indicates that linear dimension reduction can be used to accurately approximate the posterior distribution. We also evaluate the upper bound~\eqref{eq:KL_UpperBound} (up to the {same} unknown log-Sobolev constant) for parameter and data modes computed using either CCA or PCA. The approximation errors for subspaces computed using either of these strategies decay much more slowly than with our gradient-based dimension reduction approach. In this example, the number of computable CCA or PCA modes %
is also limited by the numerical rank of the covariance matrices of $\X$ and $\Y$.

\begin{figure}[!ht]
    \centering
    \begin{subfigure}[c]{0.45\textwidth}
        \includegraphics[width=\textwidth]{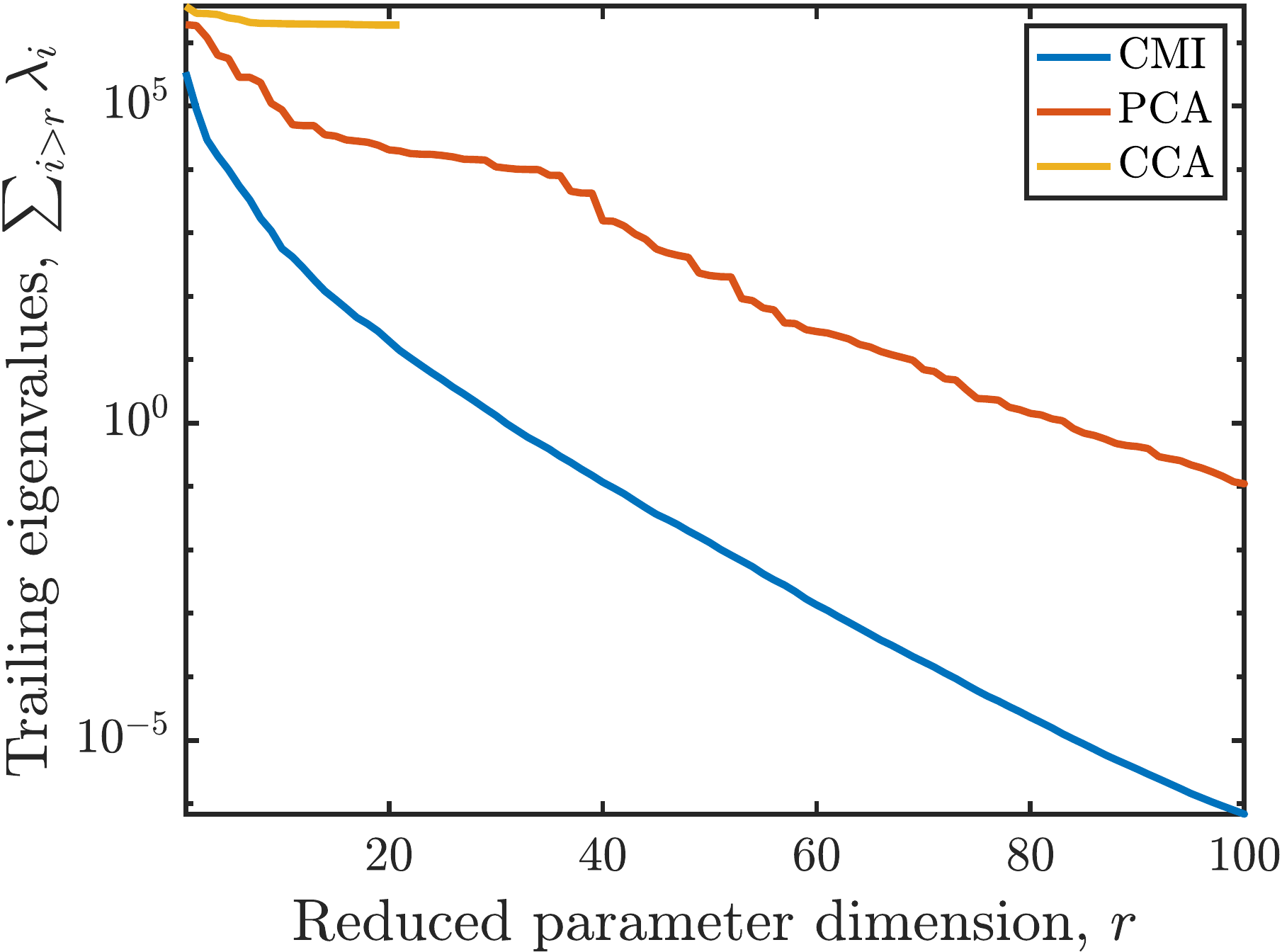}
        \caption{}
    \end{subfigure}
    \hspace{1cm}
    \begin{subfigure}[c]{0.45\textwidth}
        \includegraphics[width=\textwidth]{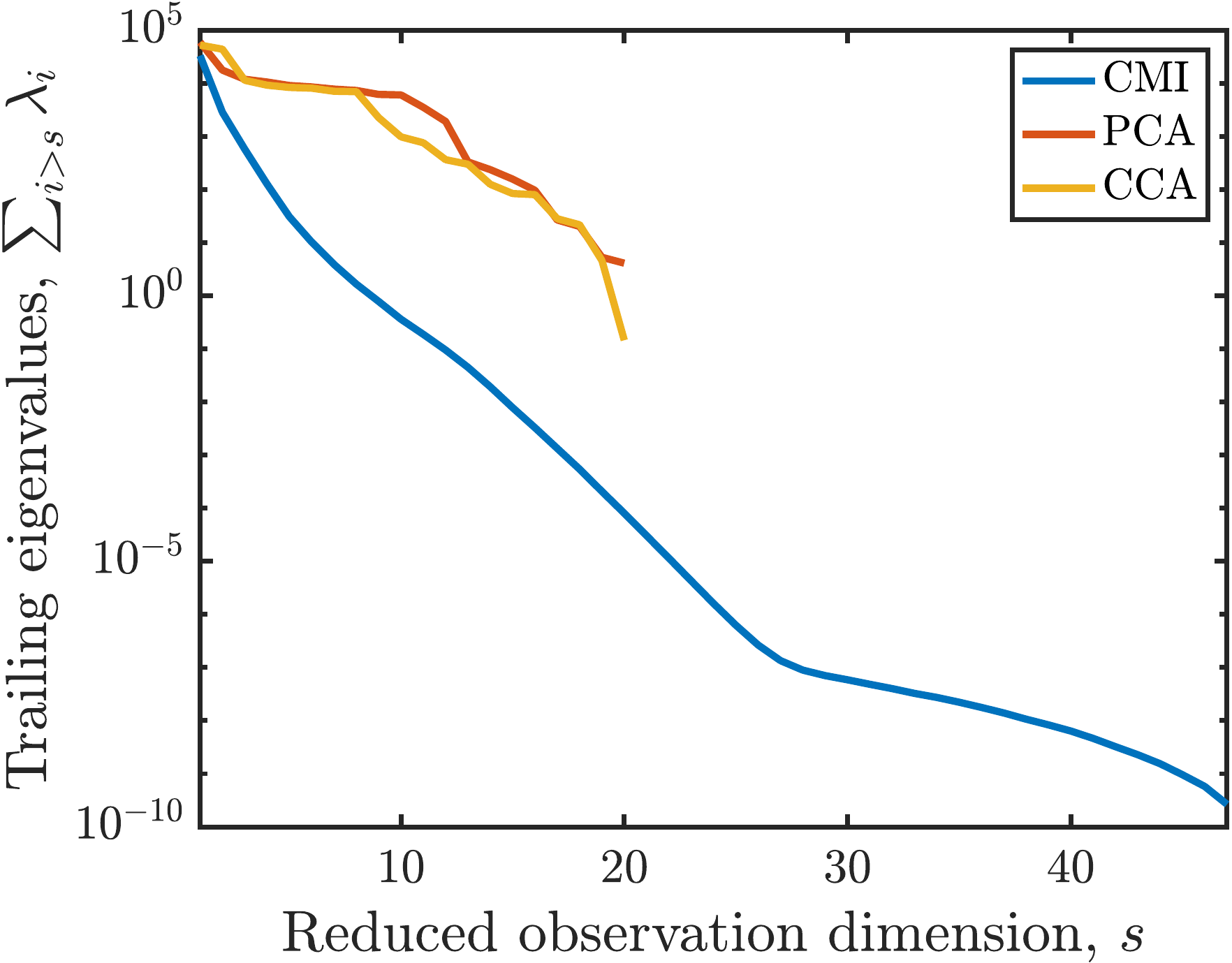}
        \caption{}
    \end{subfigure}
    \caption{Linear elasticity inverse problem: Upper bound of the expected KL divergence for three dimension reduction strategies, for increasing reduced dimensions $r$ or $s$: (a) parameter space reduction; (b) data space reduction.\label{fig:wrench_param_eigs}}
\end{figure}

Figure~\ref{fig:wrench_param_eigenvectors_CMI} plots the first three modes $(\Cpr^{1/2}\overline{u}_i)_{i=1}^3$ of the reduced parameter space, where $\overline{u}_i$ is an eigenvector of the diagnostic matrix $H_{\overline{X}}$. We observe that the informed part of the parameter is centered near the wrench's axis of rotation, where there is typically higher stress. In comparison, Figure~\ref{fig:wrench_param_eigenvectors_CCA} plots the first three parameter modes obtained using CCA, which display more global support. Analogously, Figure~\ref{fig:wrench_data_eigenvectors} plots the first five modes $(\Cobs^{1/2}\overline{v}_i)_{i=1}^5$ of the reduced data space, where $\overline{v}_i$ is an eigenvector of $H_{\overline{Y}}$. The first mode has a stronger dependence on the displacement at the left-most part of the wrench, which is also the point of highest vertical displacement. In comparison, we observe the first five modes obtained using CCA are more oscillatory, and hence capture higher-frequency components of the displacement field. %

\begin{figure}[!ht]
    \centering
    \begin{subfigure}[c]{0.45\textwidth}
        \includegraphics[trim={0 4.8cm 0 4.8cm}, clip,width=0.9\textwidth]{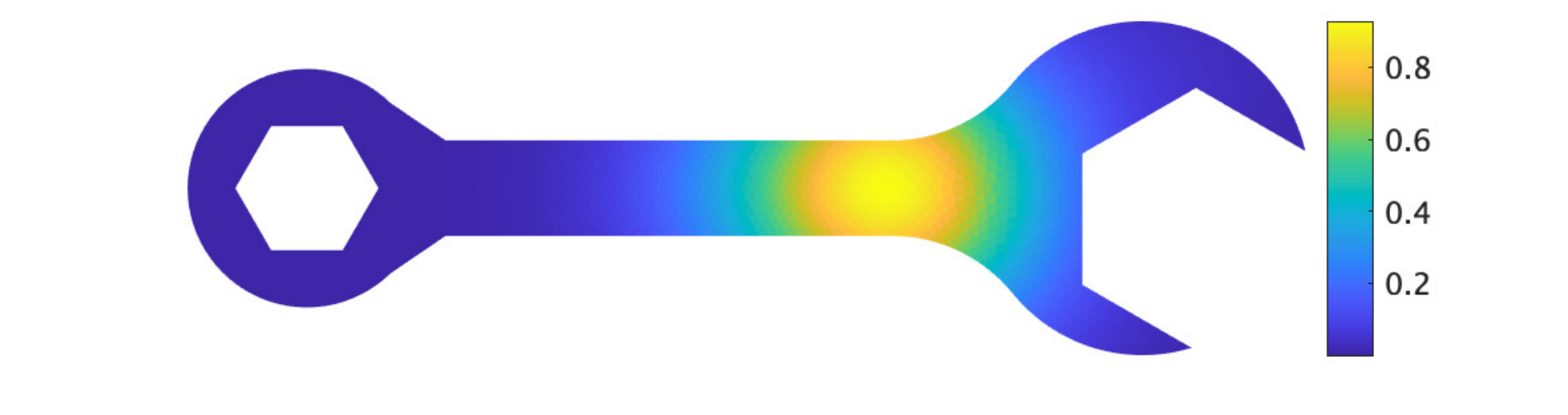}\\
        \includegraphics[trim={0 4.8cm 0 4.8cm}, clip,width=0.9\textwidth]{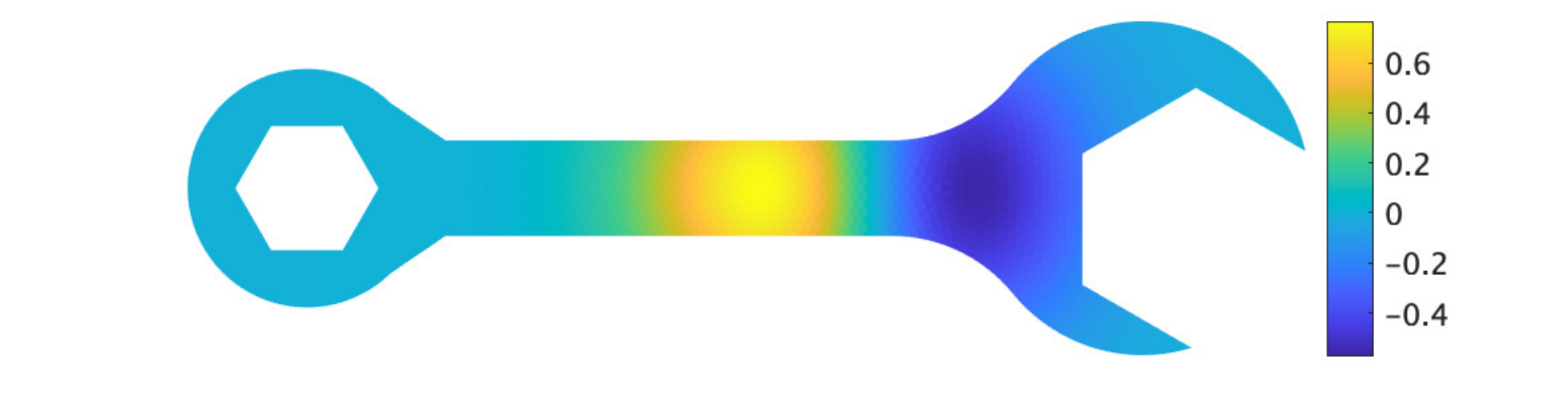}\\
        \includegraphics[trim={0 4.8cm 0 4.8cm}, clip,width=0.9\textwidth]{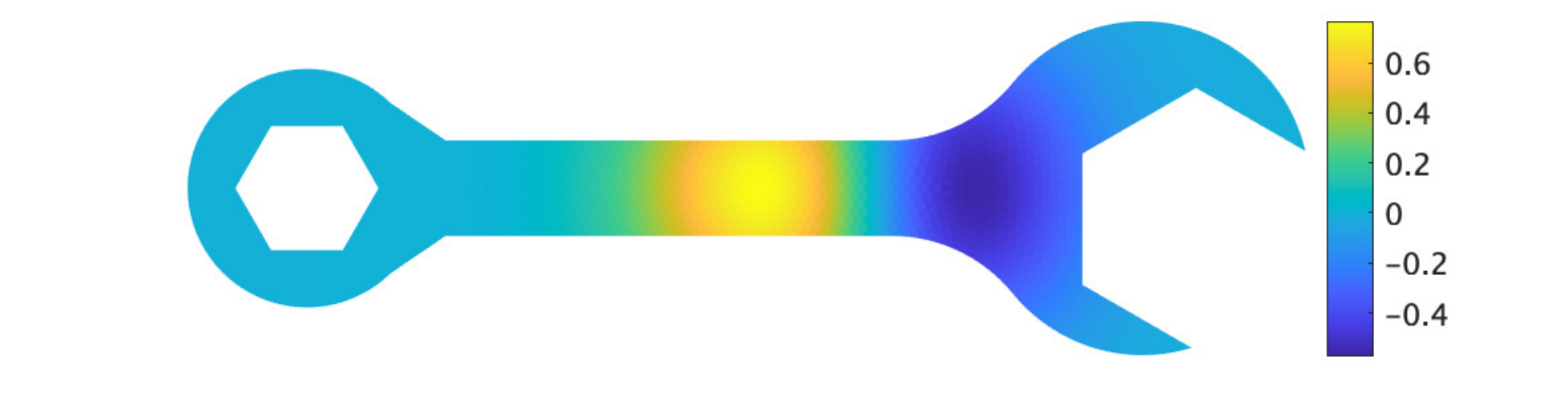}
        \caption{\label{fig:wrench_param_eigenvectors_CMI}}
    \end{subfigure}
    \hspace{1cm}
    \begin{subfigure}[c]{0.45\textwidth}
        \includegraphics[trim={0 4.8cm 0 4.8cm}, clip,width=0.9\textwidth]{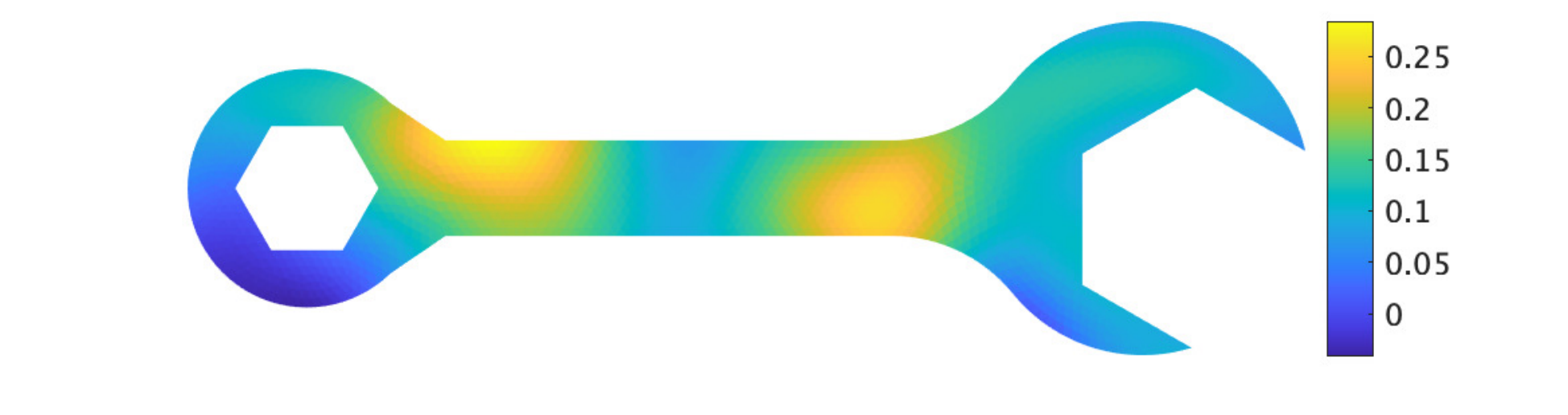}\\
        \includegraphics[trim={0 4.8cm 0 4.8cm}, clip,width=0.9\textwidth]{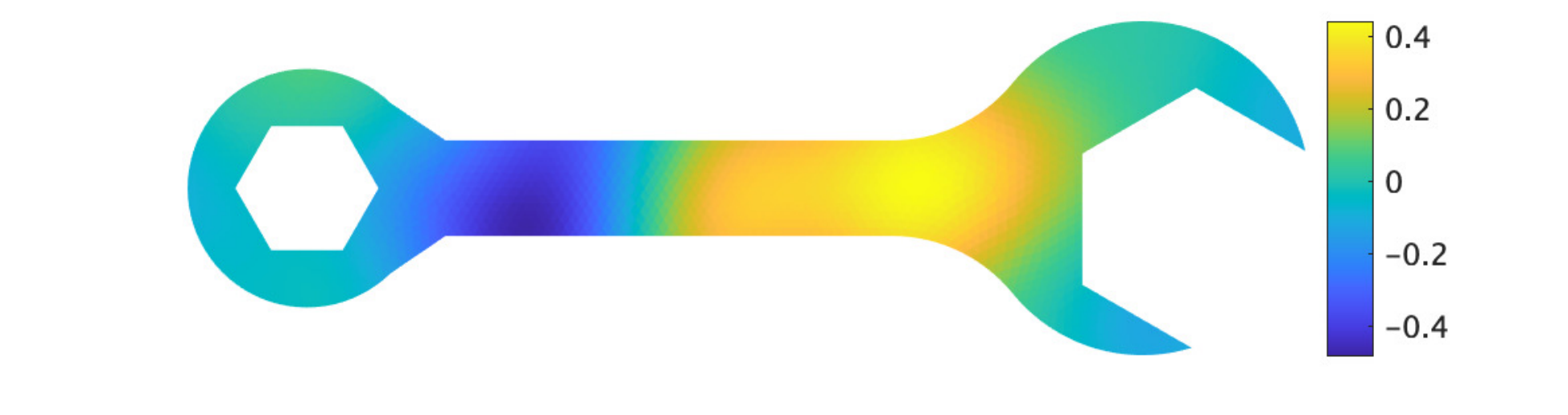}\\
        \includegraphics[trim={0 4.8cm 0 4.8cm}, clip,width=0.9\textwidth]{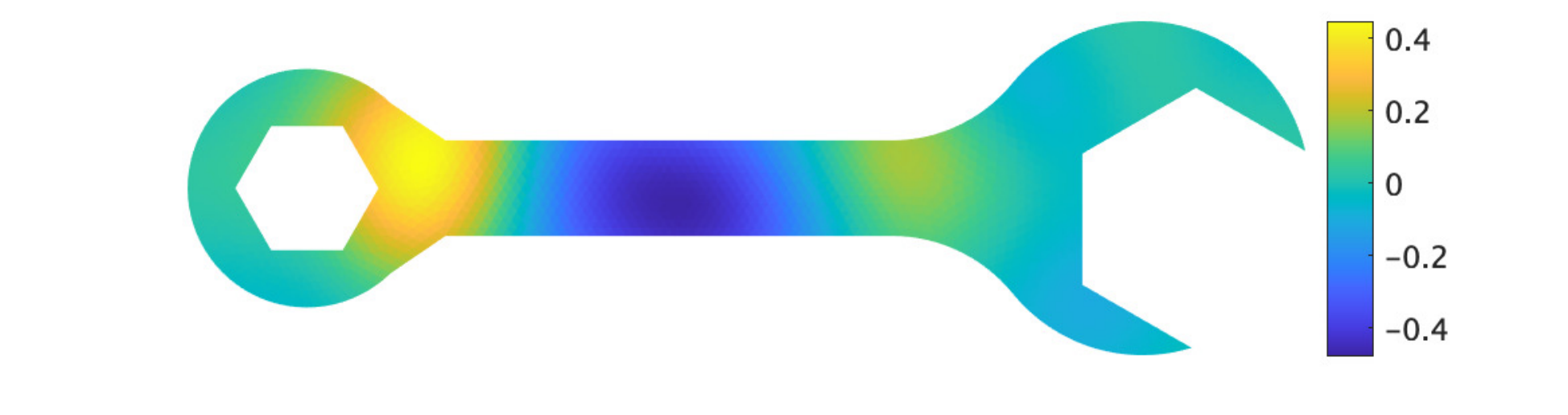}
        \caption{\label{fig:wrench_param_eigenvectors_CCA}}
    \end{subfigure}
    \caption{Linear elasticity inverse problem: The first three parameter-space modes obtained via (a) the gradient-based diagnostic matrix $H_{\overline{X}}$, and (b) canonical correlation analysis (CCA).}%
\end{figure}

\begin{figure}[!ht]
    \centering
    \begin{subfigure}[b]{0.45\textwidth}
        \includegraphics[width=\textwidth]{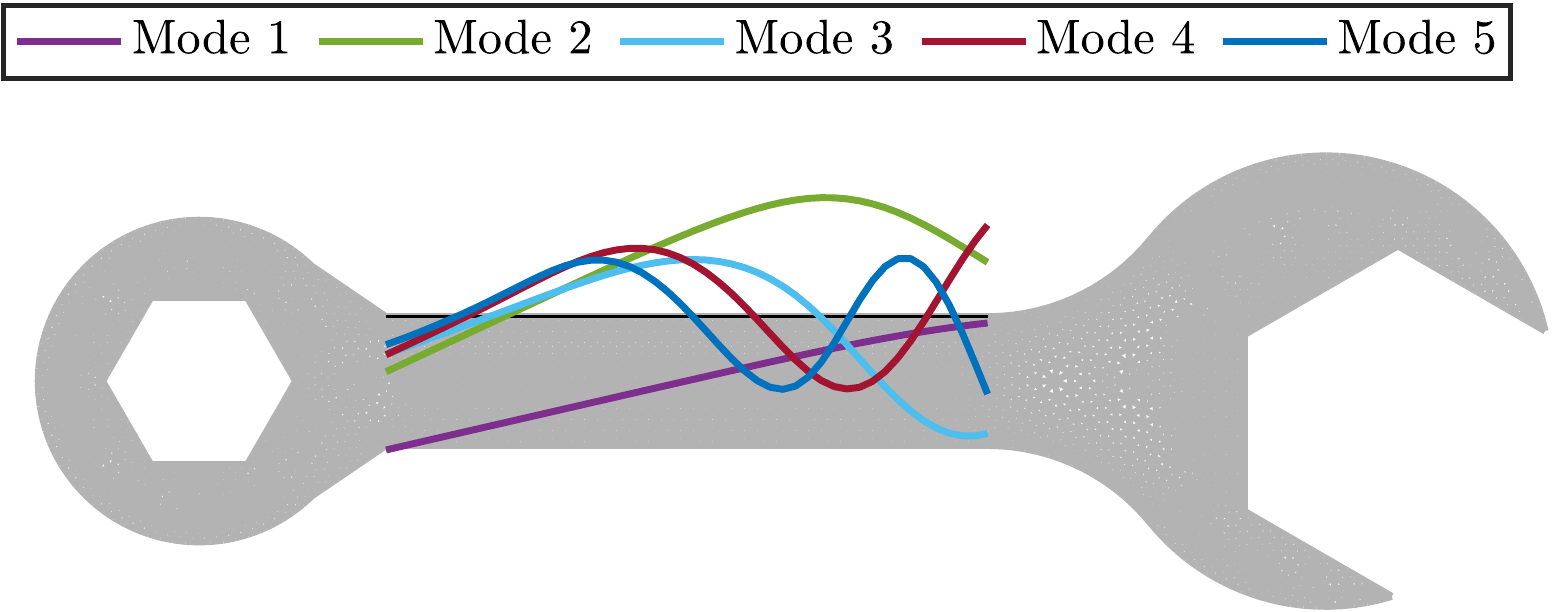}
        \caption{}
    \end{subfigure}
    \hspace{1cm}
    \begin{subfigure}[b]{0.45\textwidth}
        \includegraphics[width=\textwidth]{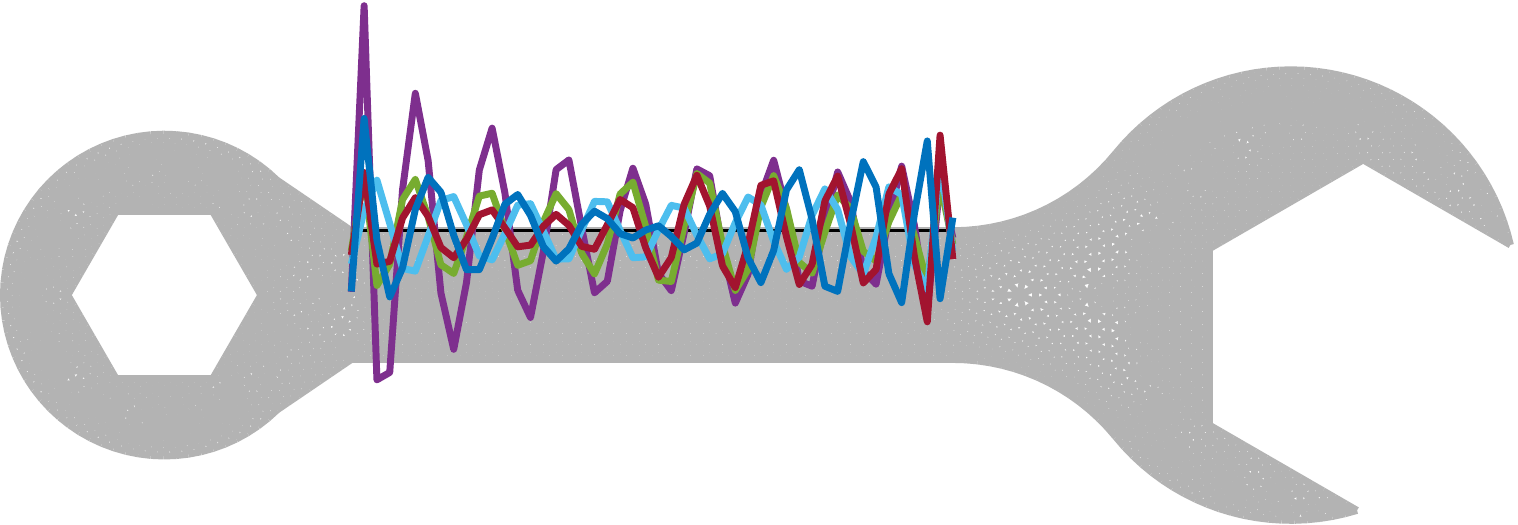}
        \caption{}
    \end{subfigure}
    \caption{Linear elasticity inverse problem: The first five data-space modes obtained from (a) the gradient-based diagnostic matrix $H_{\overline{Y}}$, and (b) CCA.} \label{fig:wrench_data_eigenvectors}
\end{figure}

Lastly, we show how to use the proposed dimension reduction technique to perform Bayesian inference using the measure transport approach from~\citet{baptista2020adaptive}. For this experiment we set the correlation length $\ell = 0.5$ and the marginal variance $\varsigma^2 = 9$. We follow the inference procedure described in Section~\ref{subsec:SampleBasedMethods}: after computing the matrices $H_{\overline{\X}}$ and $H_{\overline{\Y}}$ using $n = 500$ samples, we draw $n = 2000$ samples from the full joint distribution. Then we project those samples onto the reduced space in order to build the reduced transport map.
The transport maps are built using polynomials with adaptive degree; see the algorithm in~\citet{baptista2020adaptive}. 
In order to evaluate the quality of the resulting approximate posterior $\ApproxPost$, we decompose the expected KL divergence as
\begin{equation}
\mathbb{E}_{\Y}[\KLDiv(\Post||\ApproxPost)] = \mathbb{E}_{\X,\Y}[\log\Post(\X|\Y)] - \mathbb{E}_{\X,\Y}[\log\ApproxPost(\X|\Y)]
\end{equation}
and compute only the second term in the expression above, as the first term is both intractable and independent of $\ApproxPost$.
Table \ref{tab:nll_wrenchmark} presents sample estimates for $- \mathbb{E}_{\X,\Y}[\log\ApproxPost(\X|\Y)]$ using $5 \times 10^5$ independent samples from $\Joint$, which is often referred to as the average negative log-likelihood.
For each pair $(r,s)$, using the proposed dimension reduction method yields the best performance, i.e., the lowest value of the negative log-likelihood.

\begin{table}
    \caption{Average negative log-likelihoods $- \mathbb{E}_{\X,\Y}[\log\ApproxPost(\X|\Y)]$ for the approximate posteriors with reduced parameter and data of dimensions $(r,s)$. For readability, we add an arbitrary offset given by the average negative log-likelihood for CCA projections with $(r,s) = (1,1)$. \label{tab:nll_wrenchmark}}
    \begin{center}
    \begin{tabular}{|c|r r r r|}
     \hline
     $(r,s)$ & \multicolumn{1}{c}{$(1,1)$} & \multicolumn{1}{c}{$(2,1)$} & \multicolumn{1}{c}{$(3,1)$} & \multicolumn{1}{c|}{$(4,1)$} \\
     \hline \hline
     \text{CMI} & $-0.35 \pm 0.05$ & $-0.45 \pm 0.05$ & $-0.54 \pm 0.05$ & $-0.56 \pm 0.05$ \\
     \hline
     \text{PCA} & $-0.13 \pm 0.05$ & $-0.18 \pm 0.05$ & $-0.32 \pm 0.05$ & $-0.33 \pm 0.05$ \\
     \hline
     \text{CCA} & $0.00 \pm 0.05$ & $0.00 \pm 0.05$ & $0.00 \pm 0.05$ & $0.00 \pm 0.05$ \\
     \hline
    \end{tabular}
    \end{center}
    
    \begin{center}
    \begin{tabular}{|c|r r r r|}
     \hline
     $(r,s)$ & \multicolumn{1}{c}{$(1,2)$} & \multicolumn{1}{c}{$(2,2)$} & \multicolumn{1}{c}{$(3,2)$} & \multicolumn{1}{c|}{$(4,2)$} \\
     \hline \hline
     \text{CMI} & $-0.33 \pm 0.05$ & $-0.53 \pm 0.05$ & $-0.66 \pm 0.05$ & $-0.69 \pm 0.05$ \\
     \hline
     \text{PCA} & $-0.17 \pm 0.05$ & $-0.25 \pm 0.05$ & $-0.45 \pm 0.05$ & $-0.50 \pm 0.05$ \\
     \hline
     \text{CCA} & $0.00 \pm 0.05$ & $0.00 \pm 0.05$ & $0.00 \pm 0.05$ & $0.00 \pm 0.05$ \\ 
     \hline
    \end{tabular}
    \end{center}
\end{table}

\subsection{High-dimensional image data} \label{sec:Exp_2DImage}

We consider next an inference problem with a non-Gaussian likelihood. The goal is to infer the location of a feature in a high-dimensional image as well as an image hyperparameter; see~\citet{lueckmann2019likelihood}. 
The feature is described by its horizontal and vertical position in the image $-16\leq \x_1,\x_2\leq 16$. The hyperparameter $0.25\leq \gamma\leq 5$ defines the contrast. Thus, the parameter $\X=(\x_1,\x_2,\gamma)$ is a three-dimensional random vector endowed with uniform prior on $[-16,16]\times [-16,16]\times [0.25,5]$. Conditioned on $\X=\x$, the data $\Y|\X=\x$ is a $32\times32$ matrix drawn from the following continuous Bernoulli distribution,
\begin{equation} \label{eq:intensity_image}
    \pi_{\Y|\X}(\y|\x)\propto 
    \prod_{i,j=1}^{32} p_{ij}(\x)^{\y_{ij}} (1-p_{ij}(\x))^{1 - \y_{ij}} \ ,
\end{equation}
where
\begin{align*}
    p_{ij}(\x) &= 0.9 - 0.8\exp\left (-\frac12\left(\frac{(x^{(i)}_{1} - x_1)^2 + (x^{(j)}_2  - x_2)^2}{\sigma^2}\right)^\gamma\right),
\end{align*}
and $\{x^{(i)}_1\}_{i=1}^{32}$ and $\{x^{(j)}_2\}_{j=1}^{32}$ are the vertical and horizontal discretizations of $[-16,16]^2$.
In our experiment we always set $\sigma = 3$. 
In contrast to the setting in~\citet{lueckmann2019likelihood}, which uses a discrete (instead of a continuous) Bernoulli distribution for $Y|X$, we employ a continuous and differentiable likelihood model so that one can compute
$$
 \nabla_X\nabla_{Y_{ij}} \log\pi_{Y|X}(y|x) = \frac{\nabla_X p_{ij}(x)}{p_{ij}(x)(1 - p_{ij}(x))}.
$$

Our goal here is to reduce the dimension of the data $\Y \in \R^{1024}$ without projecting the already low-dimensional parameters $\X = (\x_1,\x_2,\gamma) \in \R^3$. 
Figure~\ref{fig:2d_sample_images} displays three realizations of $\Y$. 
Data dimension reduction can be interpreted as defining summary statistics for the data that are linear projections of $\Y$ such that $\pi_{\X|\Y} \approx \pi_{\X|\Y_s}$. Automatic methods for defining summary statistics are relevant for many likelihood-free inference procedures based on approximate Bayesian computation~\citep{fearnhead2012constructing}, whose performance is affected by the dimension of the data.

\begin{figure}[!ht]
    \centering
    \begin{subfigure}[c]{0.25\textwidth}
     \centering
     \includegraphics[width=\textwidth]{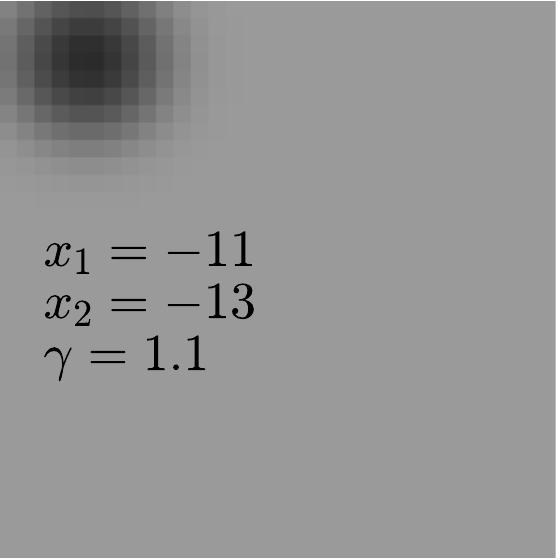}
    \end{subfigure}
    \begin{subfigure}[c]{0.25\textwidth}
     \centering
     \includegraphics[width=\textwidth]{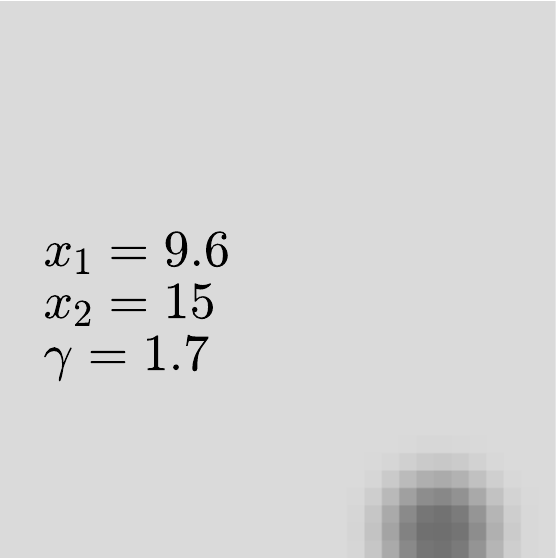}
    \end{subfigure}
    \begin{subfigure}[c]{0.25\textwidth}
     \centering
     \includegraphics[width=\textwidth]{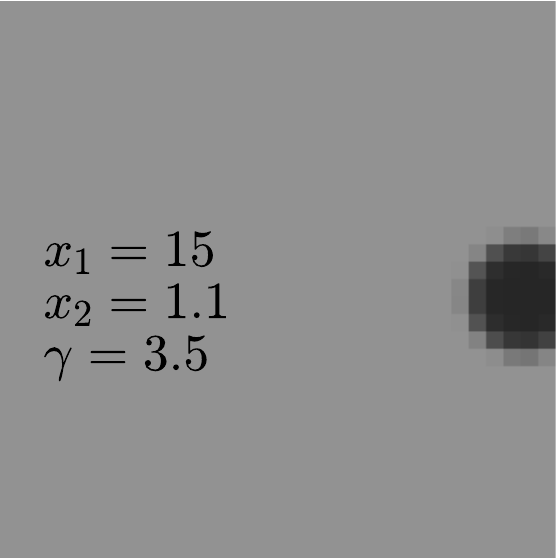}
    \end{subfigure}
    \caption{Imaging example: Three samples of the image intensities $\Y$. %
    \label{fig:2d_sample_images}}
\end{figure}

\begin{figure}[!ht]
    \centering
    \begin{subfigure}[c]{0.45\textwidth}
        \includegraphics[width=\textwidth]{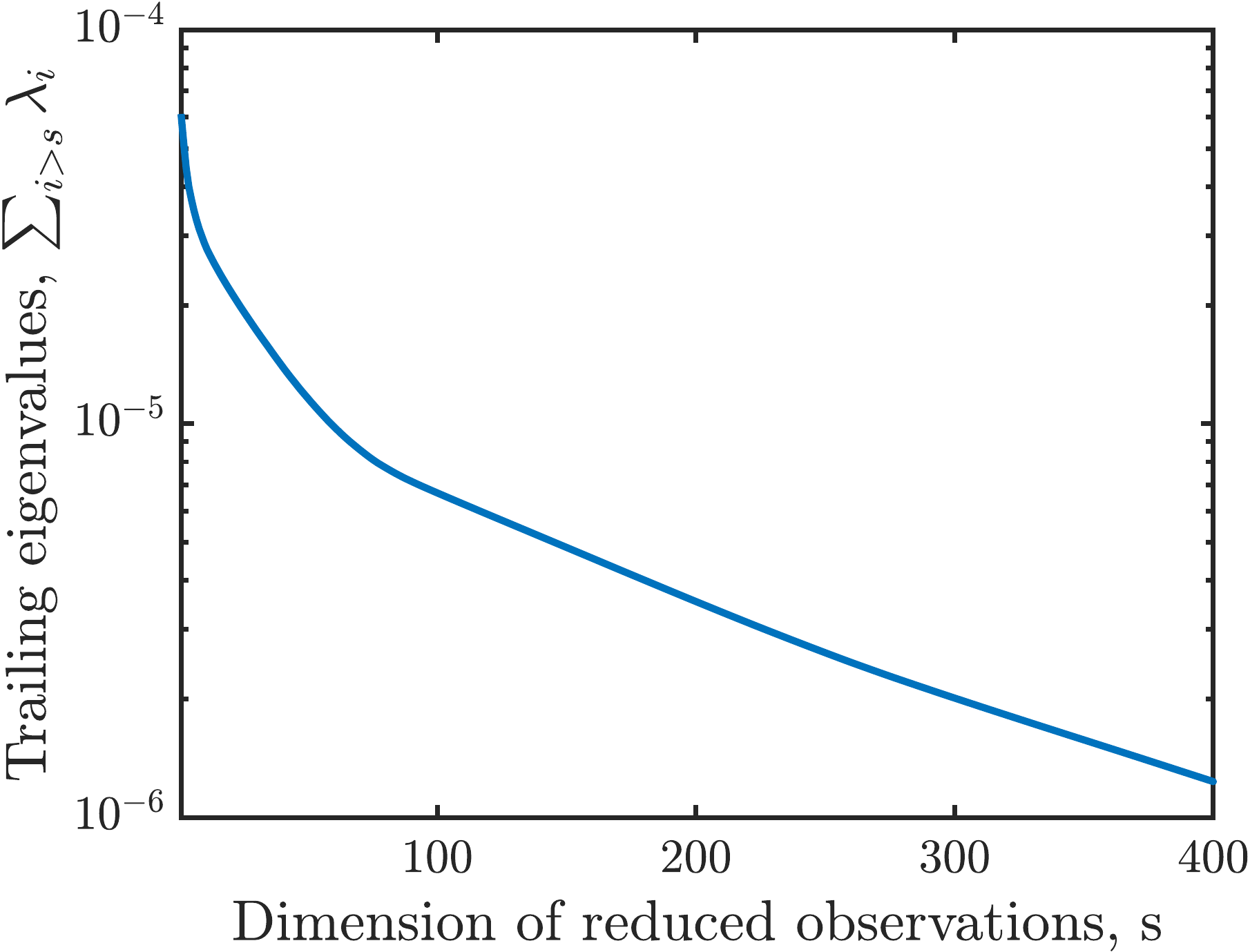}
        \caption{}
        \label{fig:2d_eigenvalues}
    \end{subfigure}
    \hspace{1cm}
    \begin{subfigure}[c]{0.45\textwidth}
        \includegraphics[width=\textwidth]{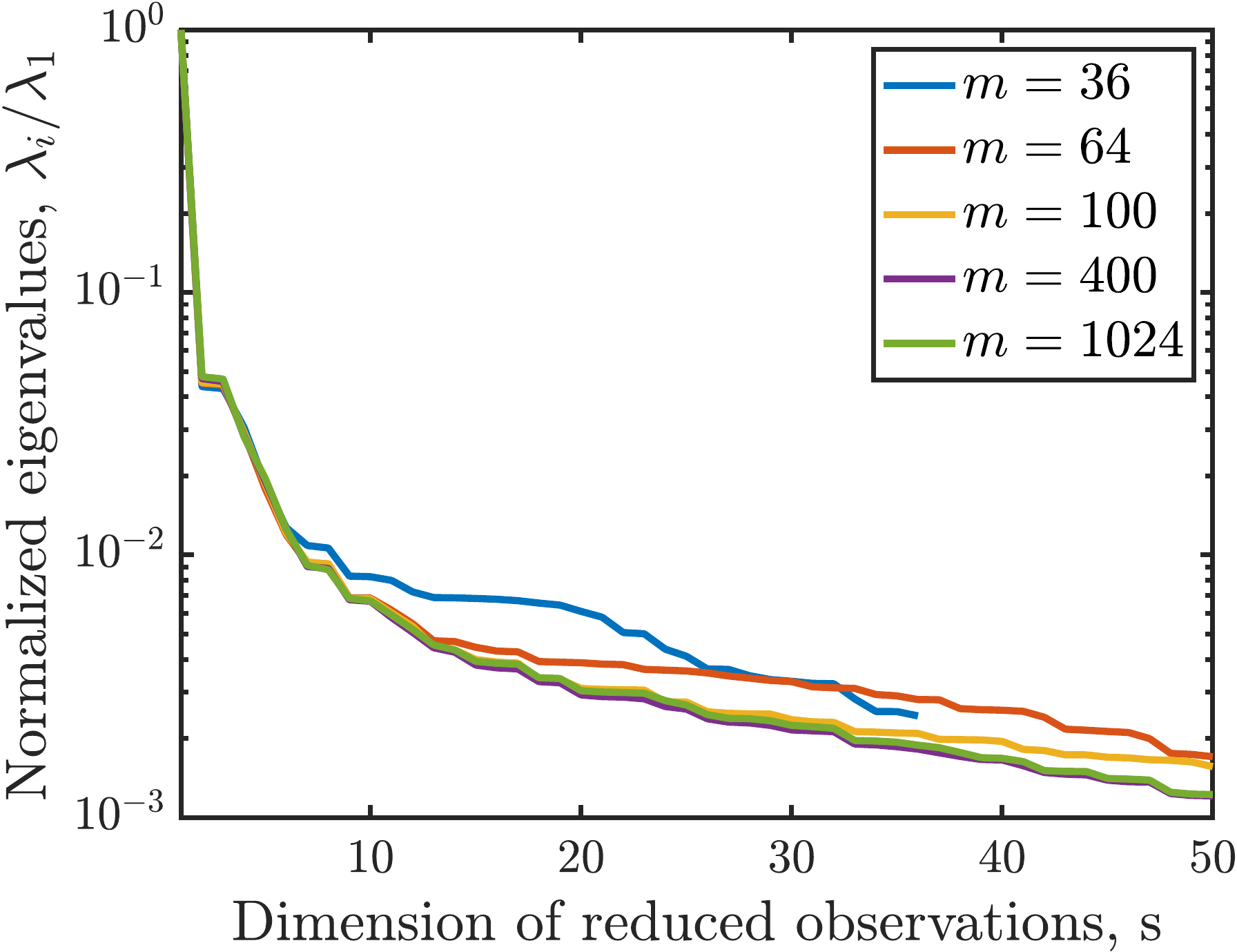}
        \caption{\label{fig:2d_eigenvalues_vs_gridresolution}}
    \end{subfigure}
    \caption{Imaging example: (a) Sum of the trailing eigenvalues of $H_\Y$ for projecting the data. (b) First 100 eigenvalues with increasing image resolution.} %
\end{figure} %

To reduce the dimension of $\Y$, we evaluate the mixed partial derivatives of the log-likelihood $\nabla_\X \nabla_\Y \log\pi_{\Y|\X} \in \R^{3 \times 1024}$ at $n = 10^5$ samples $(\X,\Y) \sim \pi_{\X,\Y}$ and assemble a Monte Carlo estimate for the matrix $H_\Y$ in~\eqref{eq:DiagnosticMatrixY}. Figure~\ref{fig:2d_eigenvalues} plots the error indicator $\sum_{i=s+1}^m \lambda_i(H_\Y)$ which, up to the unknown constant $\overline{C}(\Joint)$, corresponds to the right-hand side of~\eqref{eq:PostErr_Subspaces} with $r = d$. We observe fast decay in this sum, which indicates that low-dimensional projections of the image may be sufficient to update the parameters. In Figure~\ref{fig:2d_eigenvalues_vs_gridresolution} we demonstrate that this decay is unaffected by the grid resolution.

Figure~\ref{fig:2d_image_eigenvectors} displays the ten leading eigenvectors of $H_{\Y}$. We observe low-frequency oscillations in the first eigenvectors, which are sufficient to approximately determine parameters such as the location of the circular blob, while higher-order eigenvectors distinguish finer features of the image. 

\begin{figure}[!ht]
    \centering
    \includegraphics[width=0.19\textwidth]{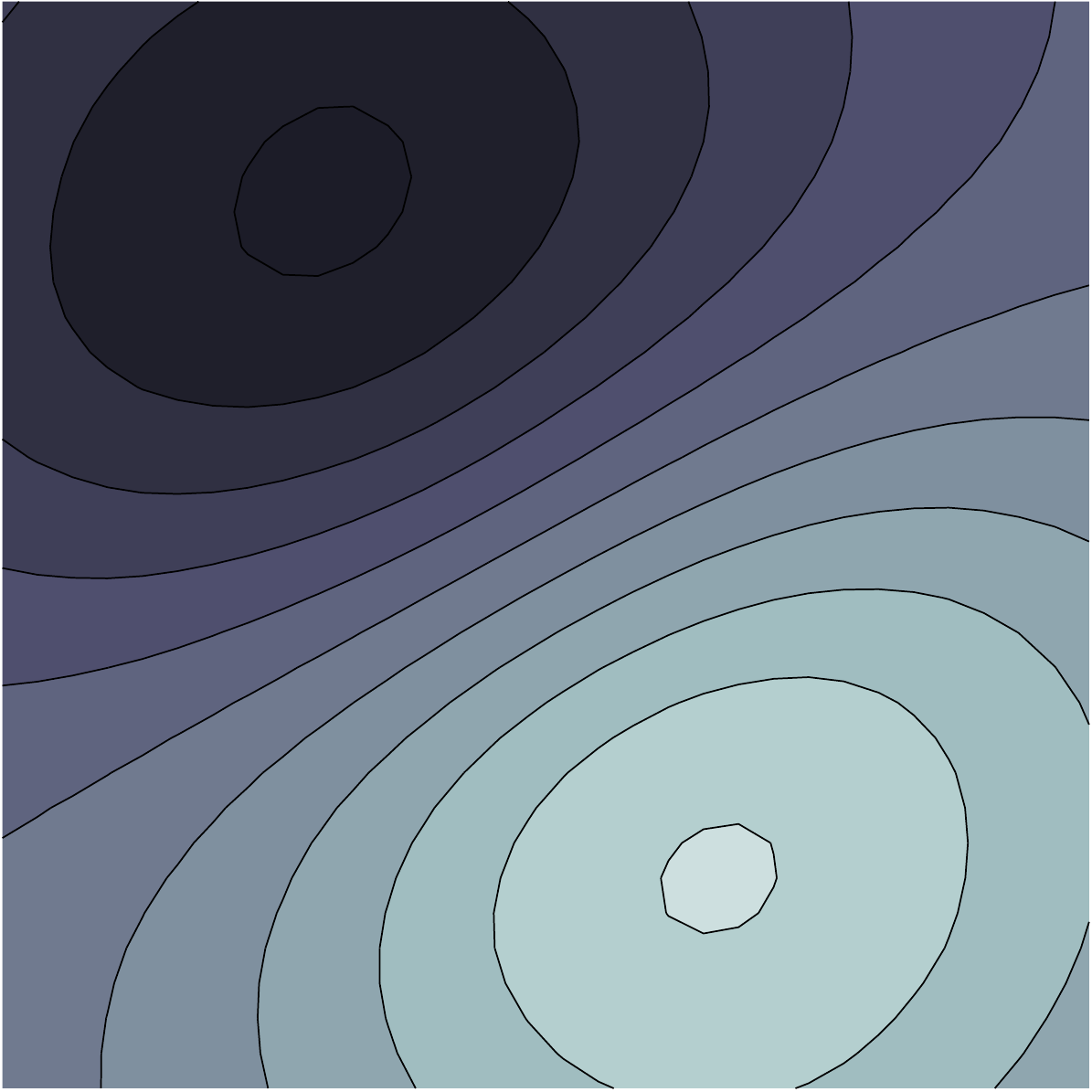}
    \includegraphics[width=0.19\textwidth]{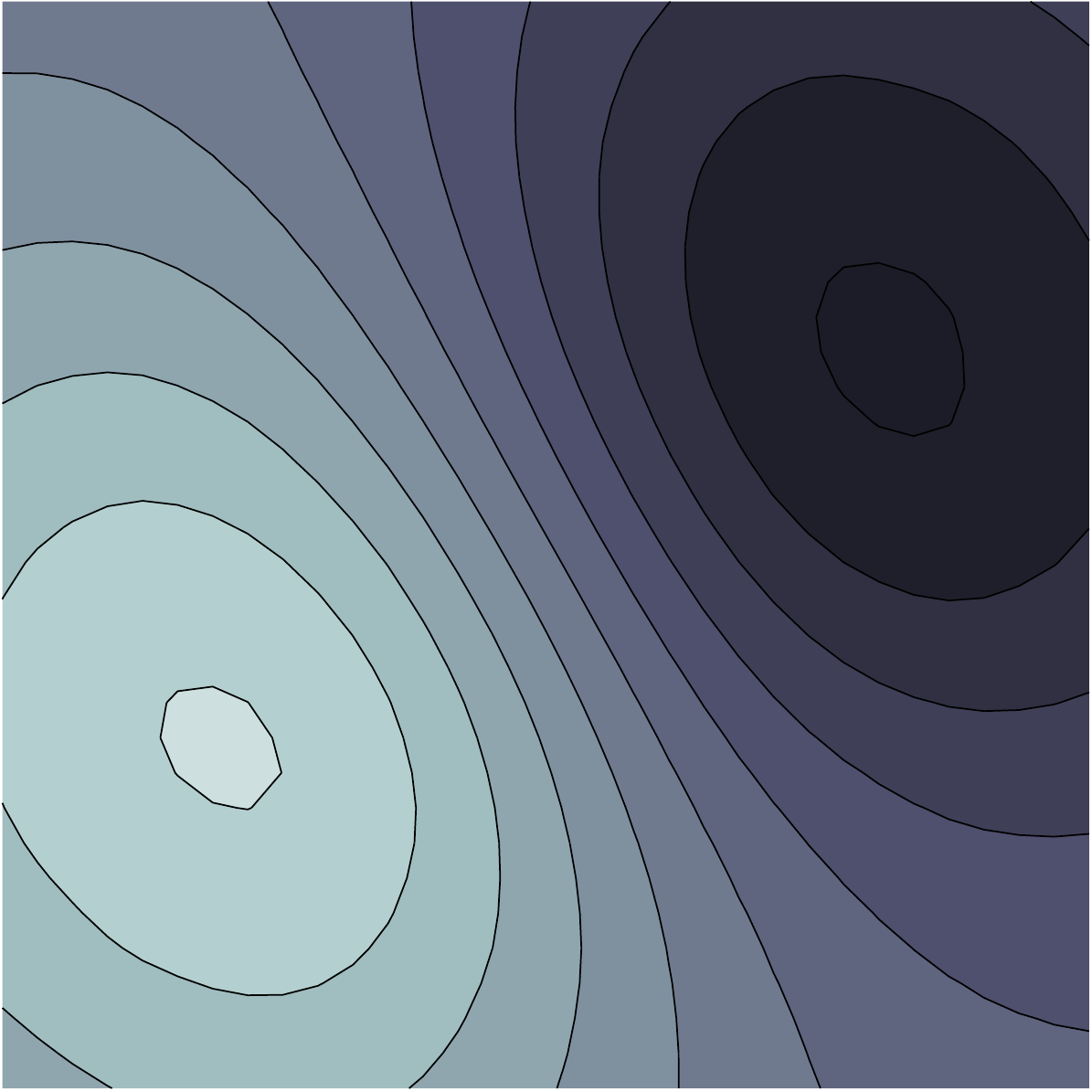}
    \includegraphics[width=0.19\textwidth]{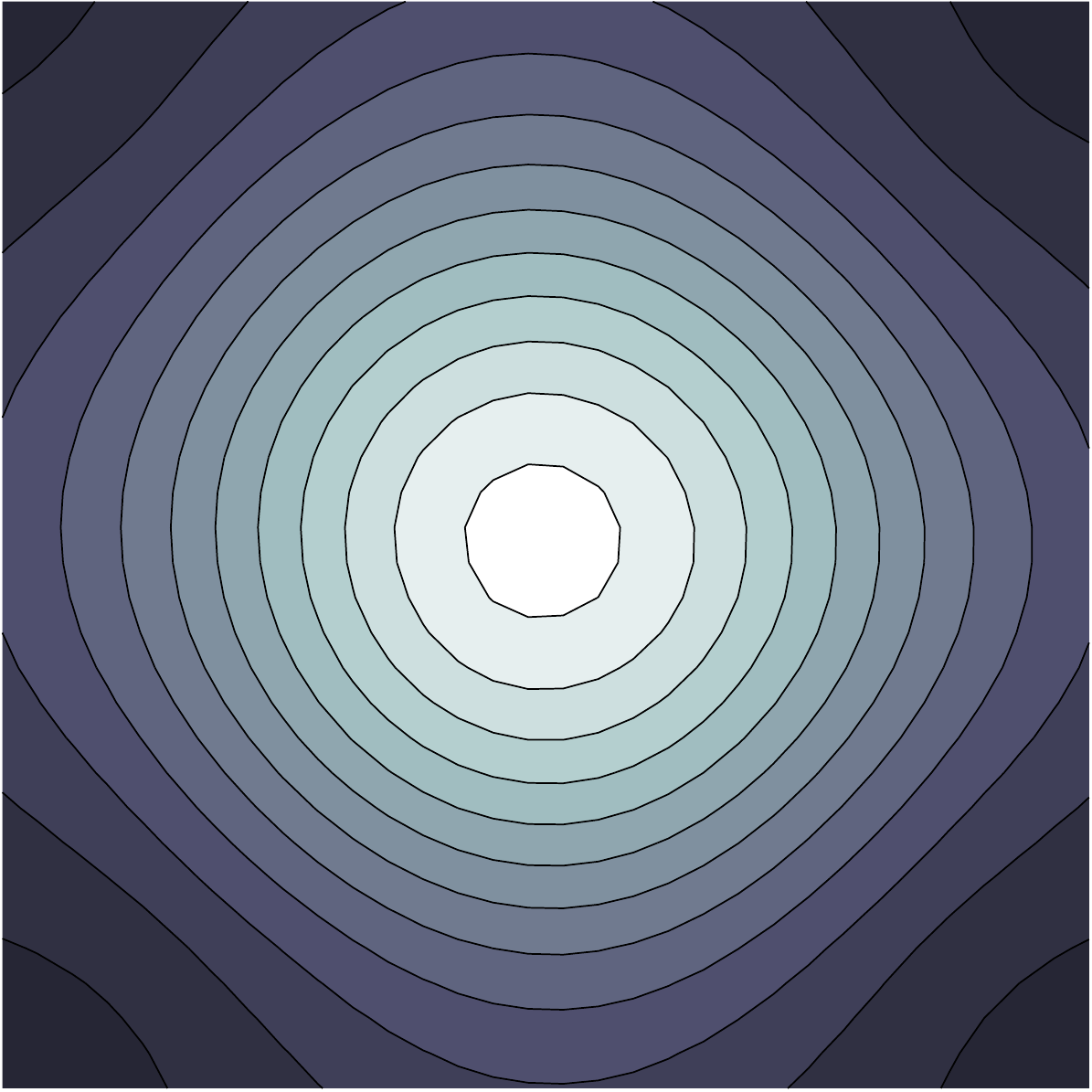}
    \includegraphics[width=0.19\textwidth]{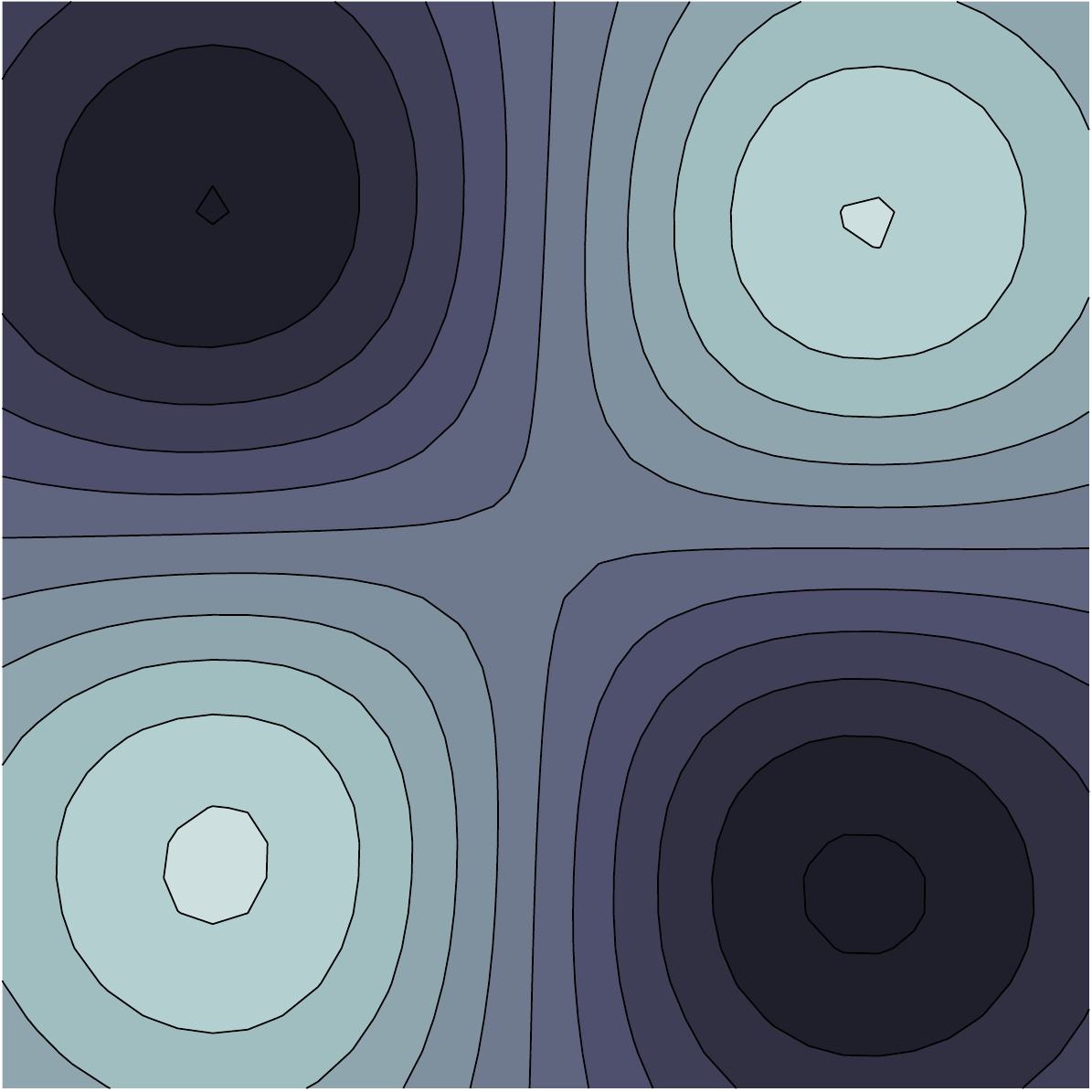}\\
    \includegraphics[width=0.19\textwidth]{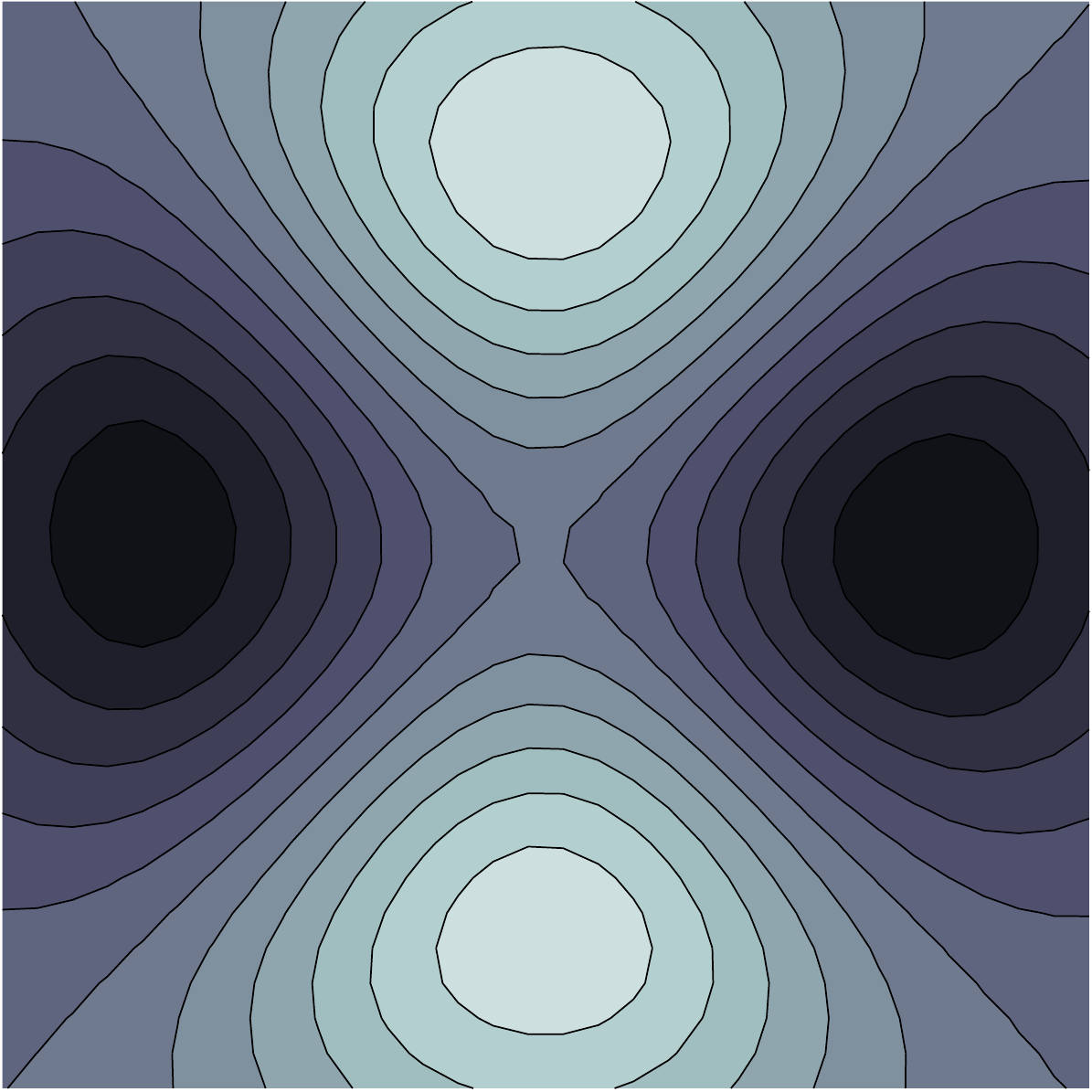}
    \includegraphics[width=0.19\textwidth]{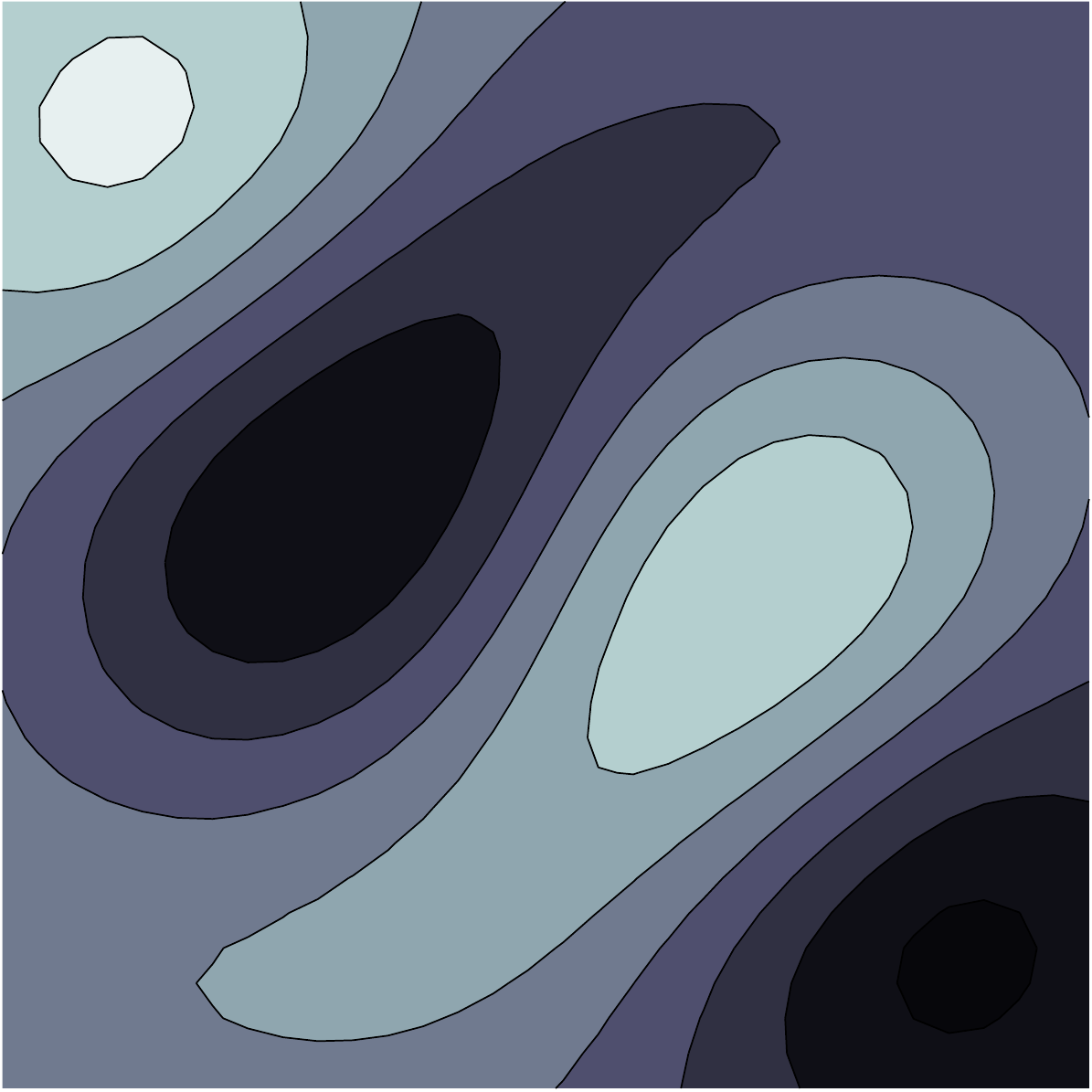}
    \includegraphics[width=0.19\textwidth]{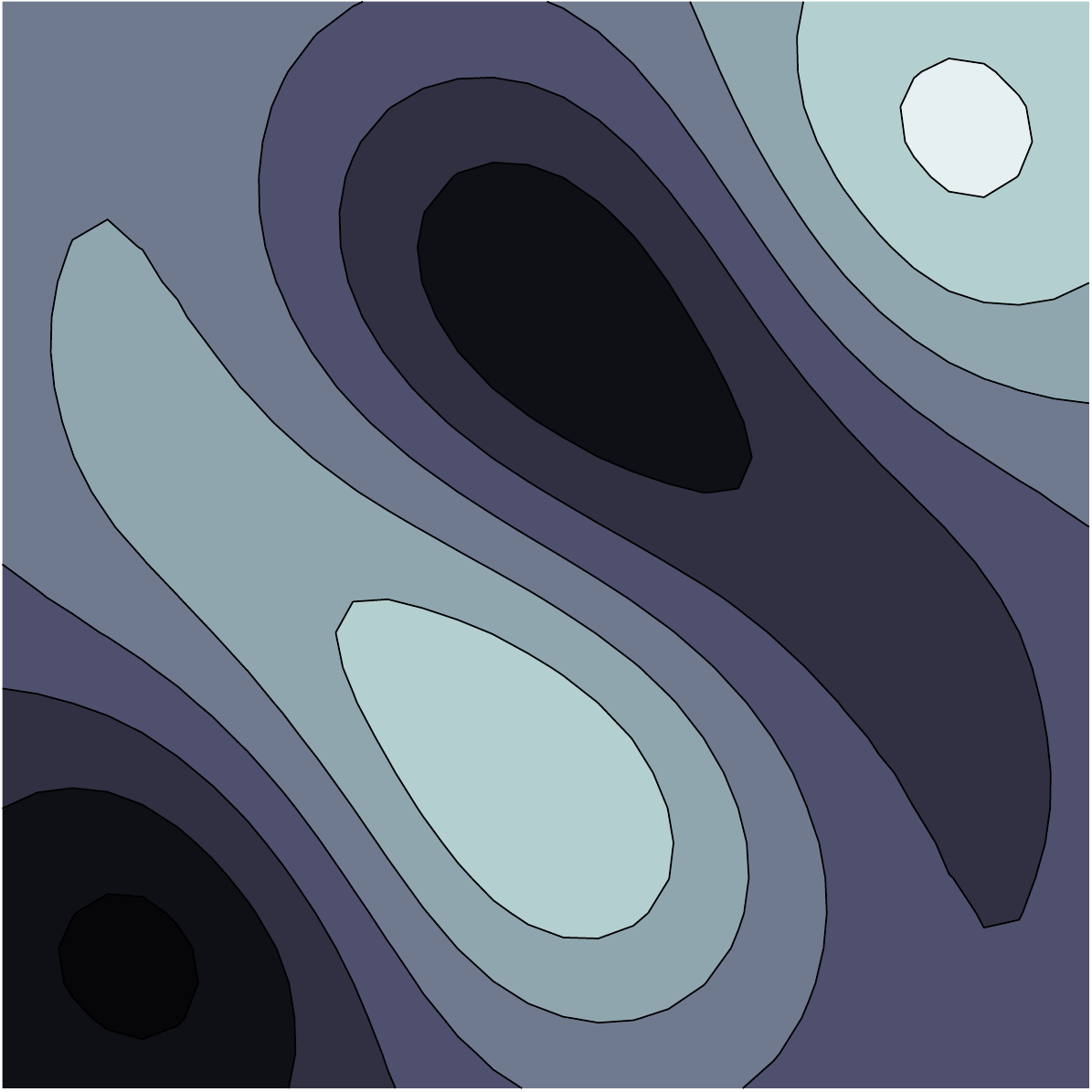}
    \includegraphics[width=0.19\textwidth]{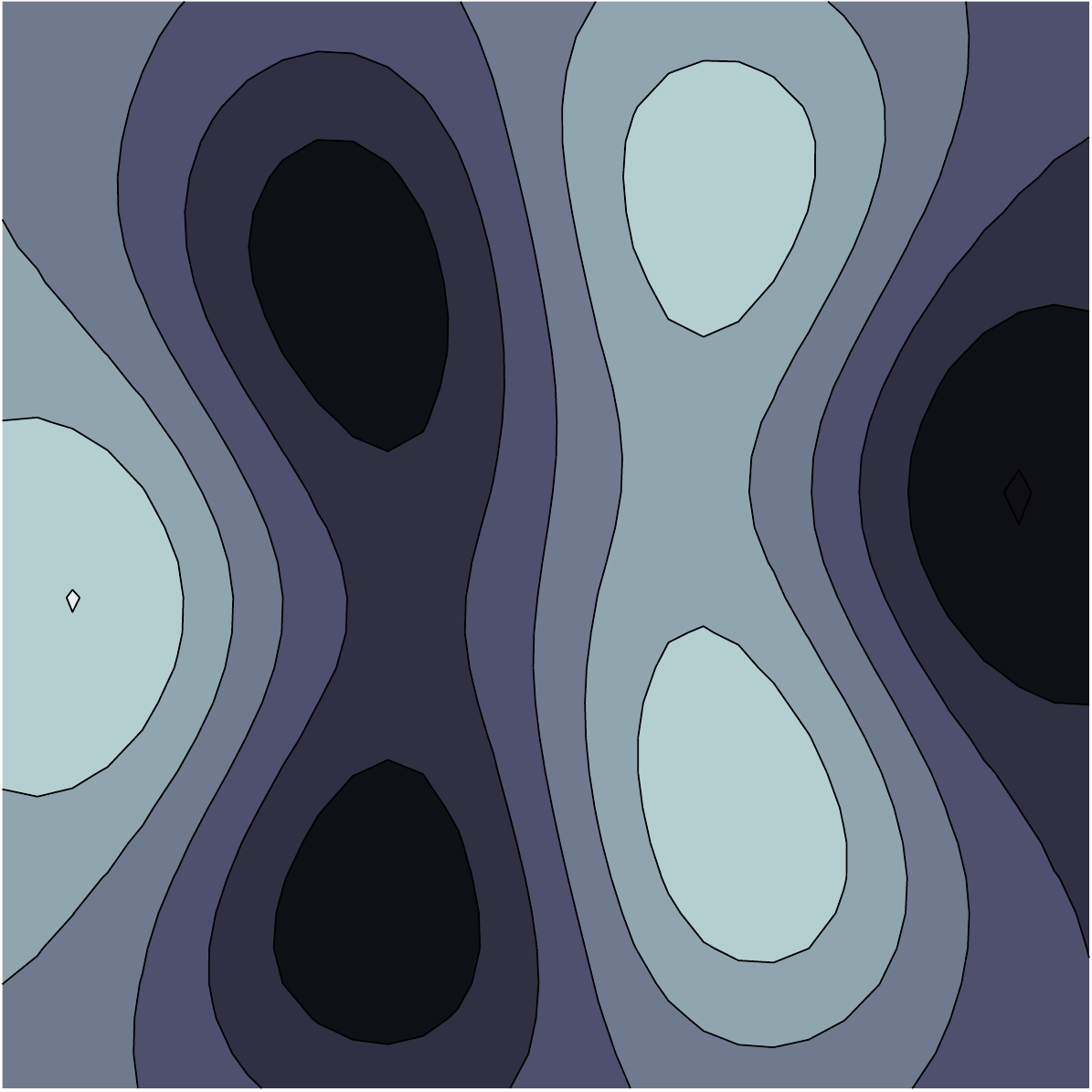}
    \caption{Imaging example: The first eight leading non-constant eigenvectors of the data diagnostic matrix $H_\Y$ \label{fig:2d_image_eigenvectors}}
\end{figure}

Gradients of the log-likelihood are also useful to identify \emph{goal-oriented} subspaces of the data, i.e., subspaces informative for a chosen subset of the parameters. Suppose that the vertical position $x_2$ and the contrast parameter are not of interest, so that $\X \equiv x_1$ is the only parameter we wish to infer. By doing this, we modify the diagnostic matrix $H_{\Y}$ by using only the single row $\nabla_{\X_1} \nabla_\Y \log\pi_{\Y|\X}$, i.e., we ignore the other parameters and only consider approximating the distribution of $\X_1 | \Y$.
Figure~\ref{fig:2d_image_eigenvectors_goal_oriented_xpos} plots the leading eigenvectors of the resulting diagnostic matrix. We see that the first four eigenvectors of $H_{\Y}$ capture horizontal variations in the image, while remaining nearly constant along the vertical axis. Analogous structure is observed if the $x_1$ position and $\gamma$ are not of interest, so that $\X \equiv x_2$. The first four eigenvectors of $H_{\Y}$ for this case are displayed in Figure~\ref{fig:2d_image_eigenvectors_goal_oriented_ypos}. 

\begin{figure}[!ht]
    \centering
    \begin{subfigure}{\textwidth}
        \centering
        \includegraphics[width=0.19\textwidth]{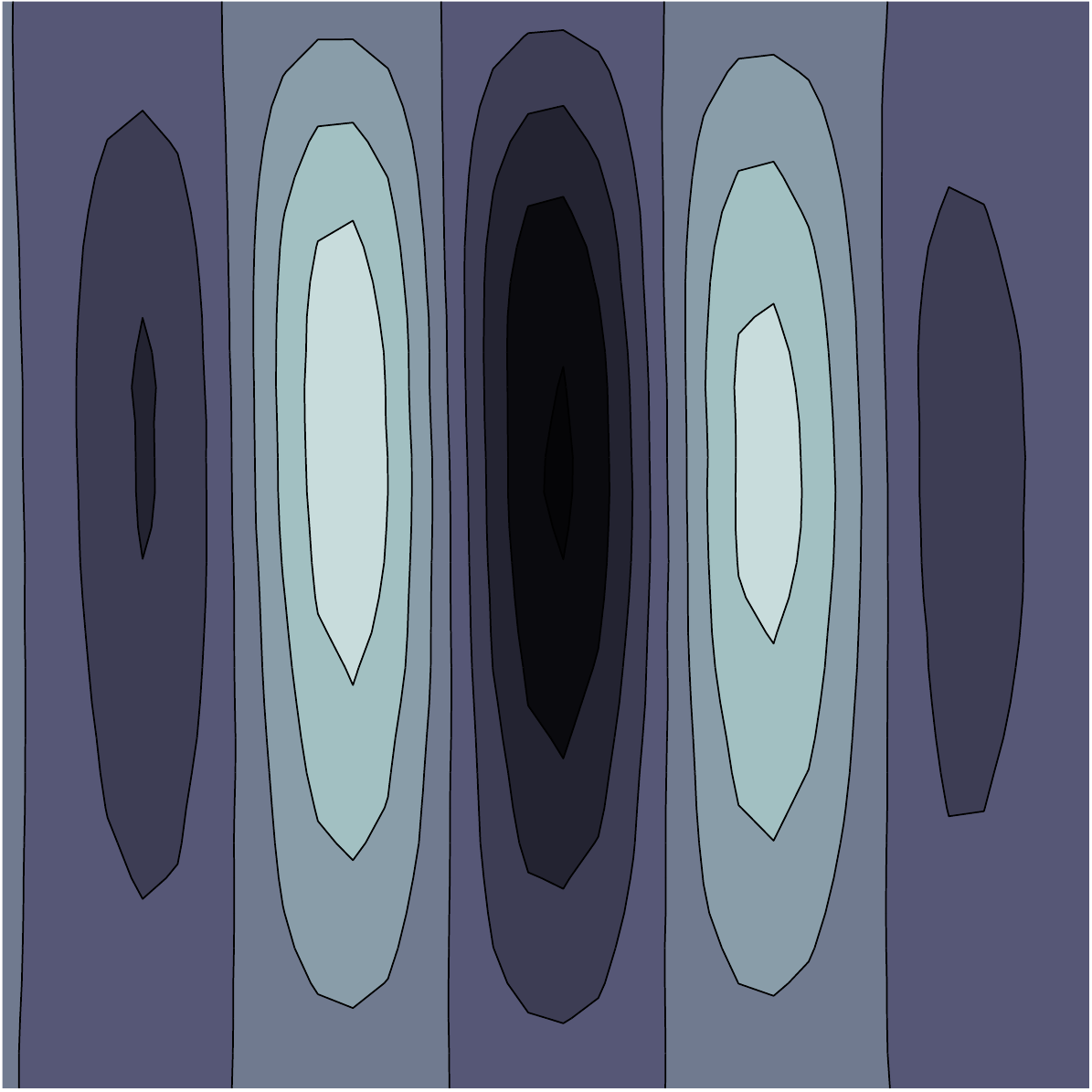}
        \includegraphics[width=0.19\textwidth]{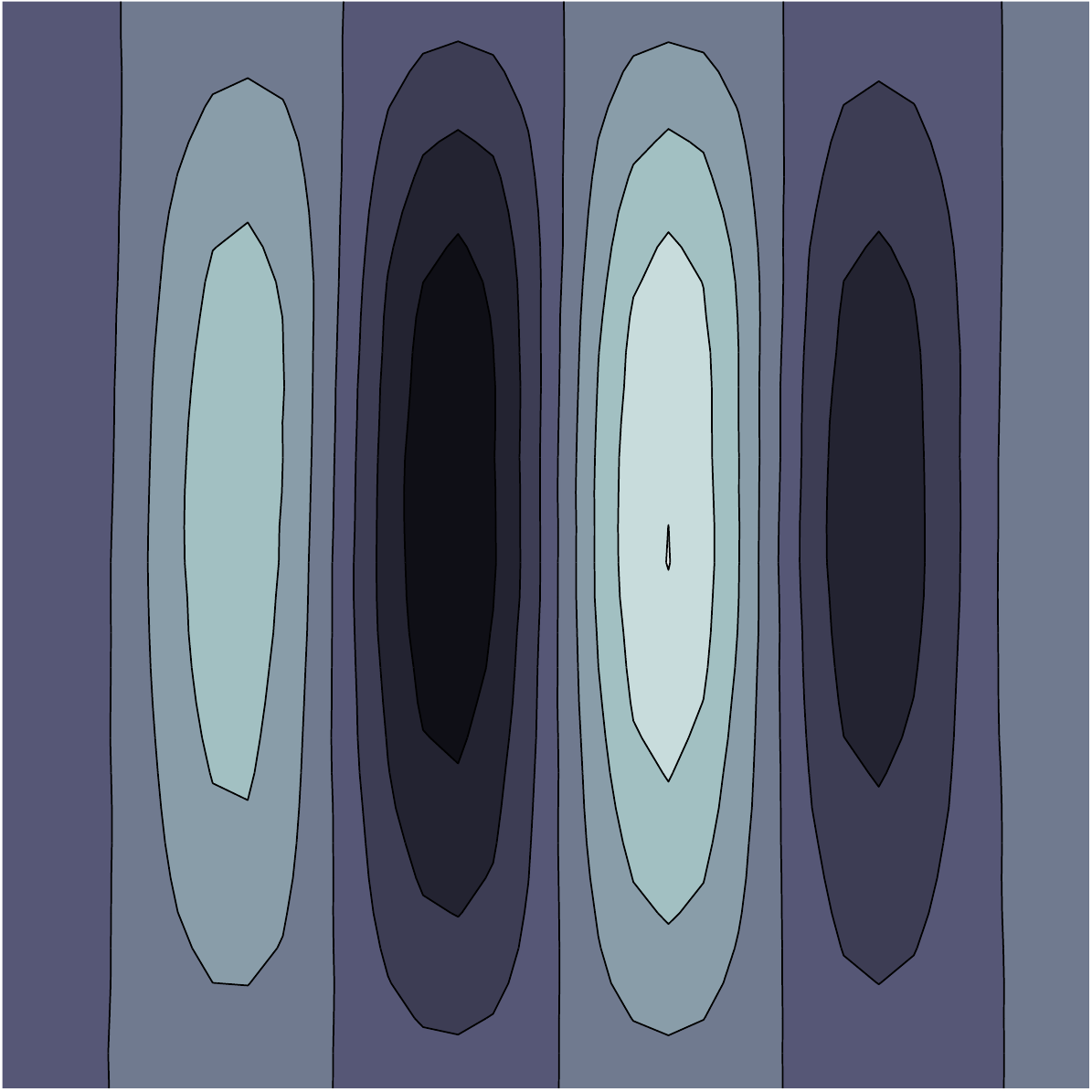}
        \includegraphics[width=0.19\textwidth]{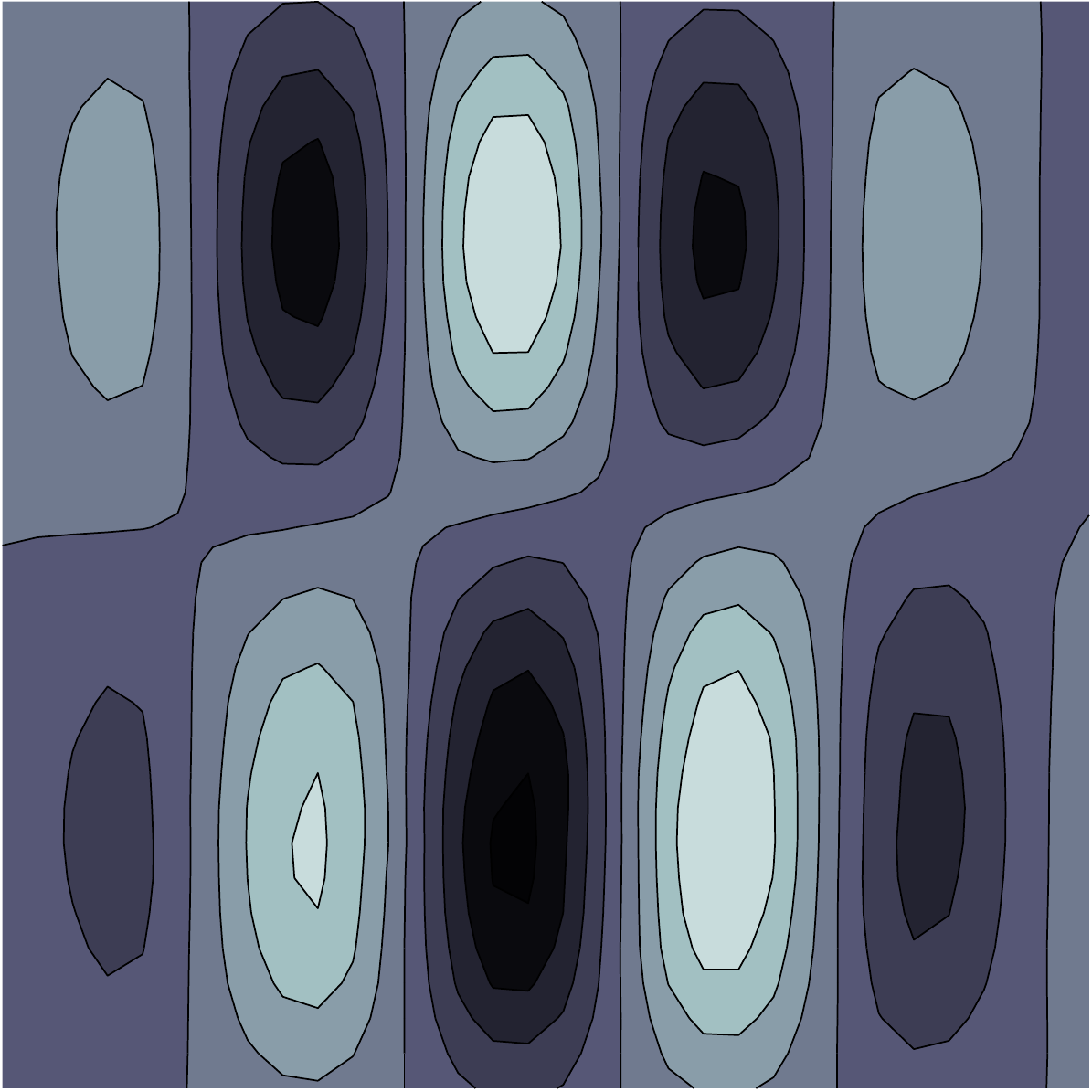}
        \includegraphics[width=0.19\textwidth]{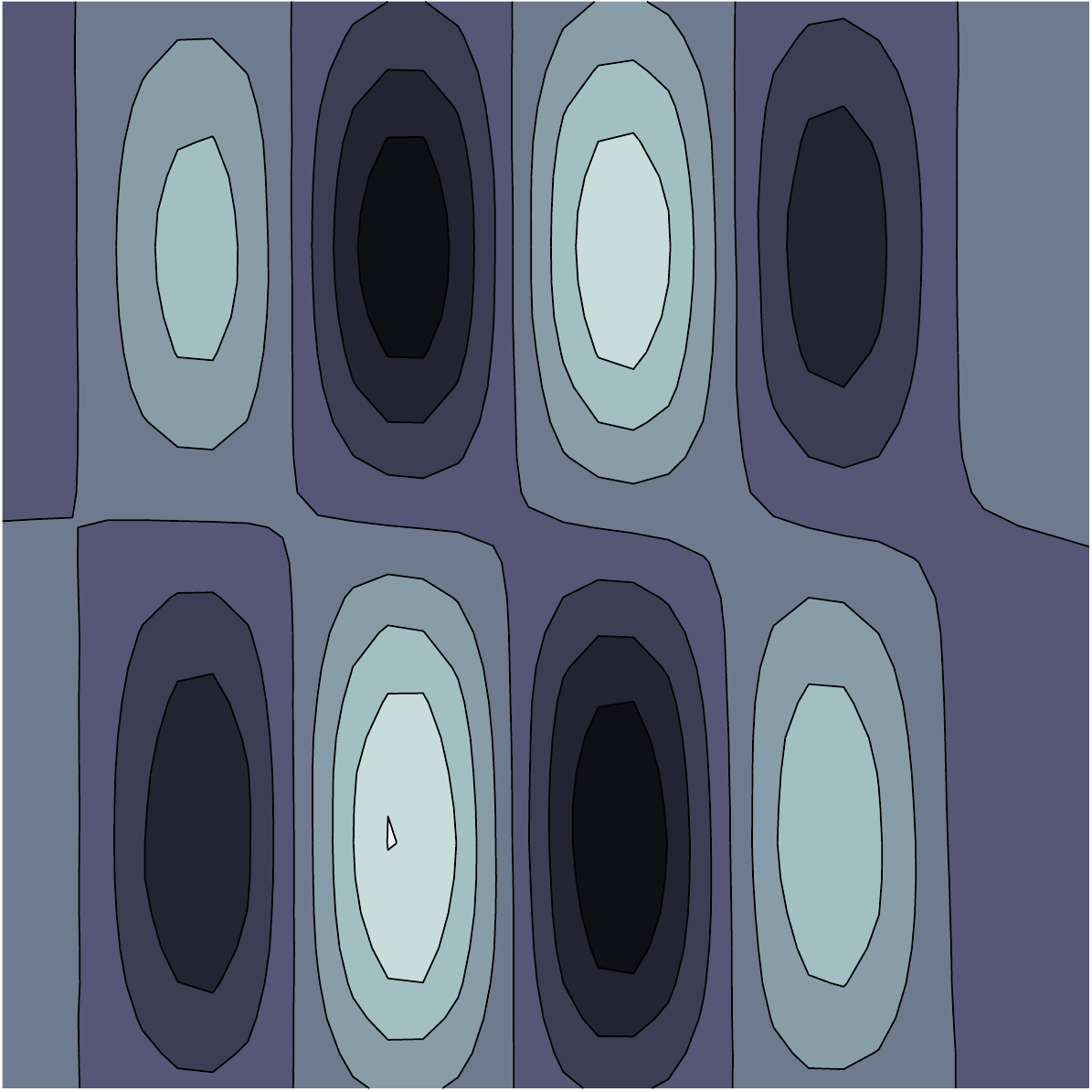}
        \caption{\label{fig:2d_image_eigenvectors_goal_oriented_xpos}}
    \end{subfigure}
    \begin{subfigure}{\textwidth}
        \centering
        \includegraphics[width=0.19\textwidth]{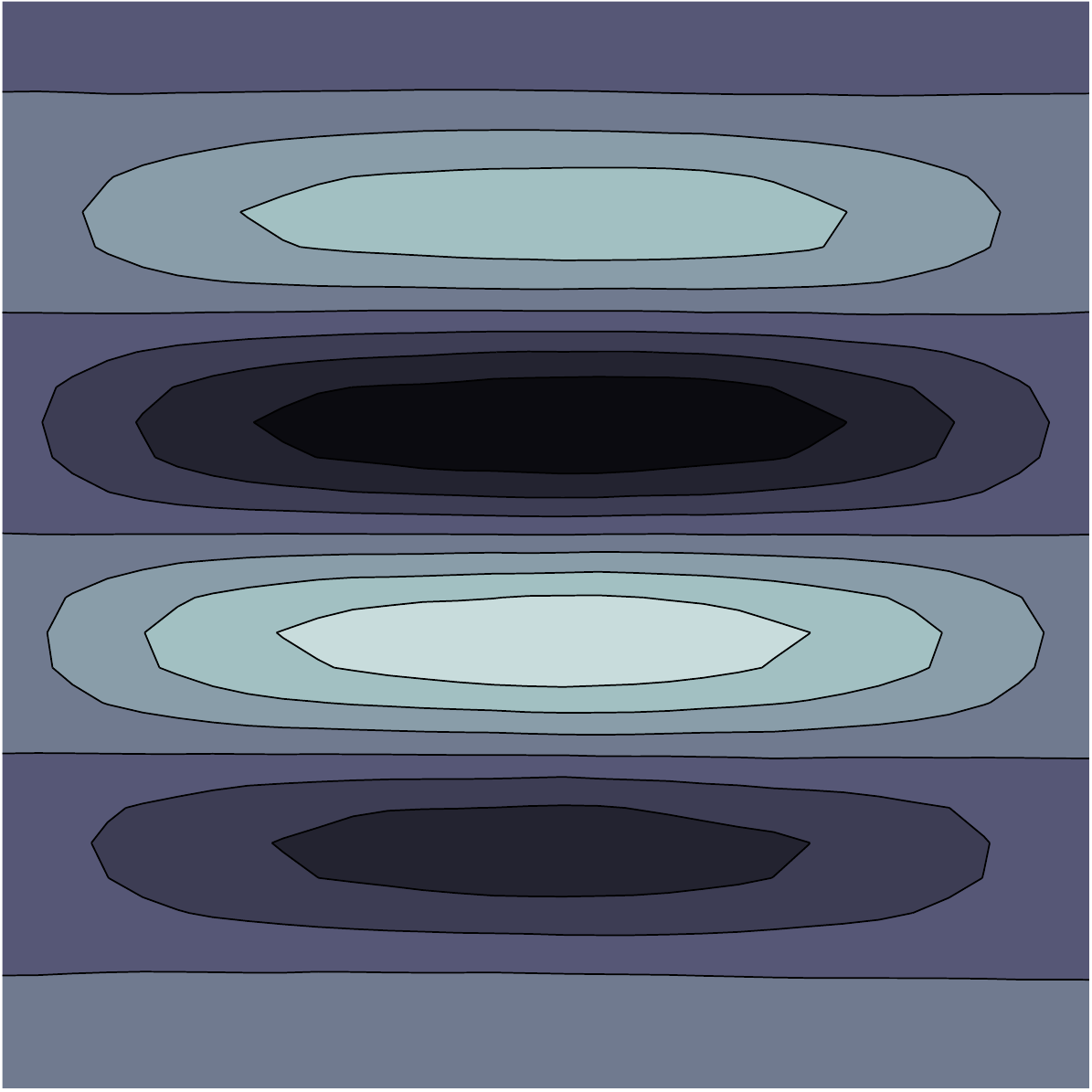}
        \includegraphics[width=0.19\textwidth]{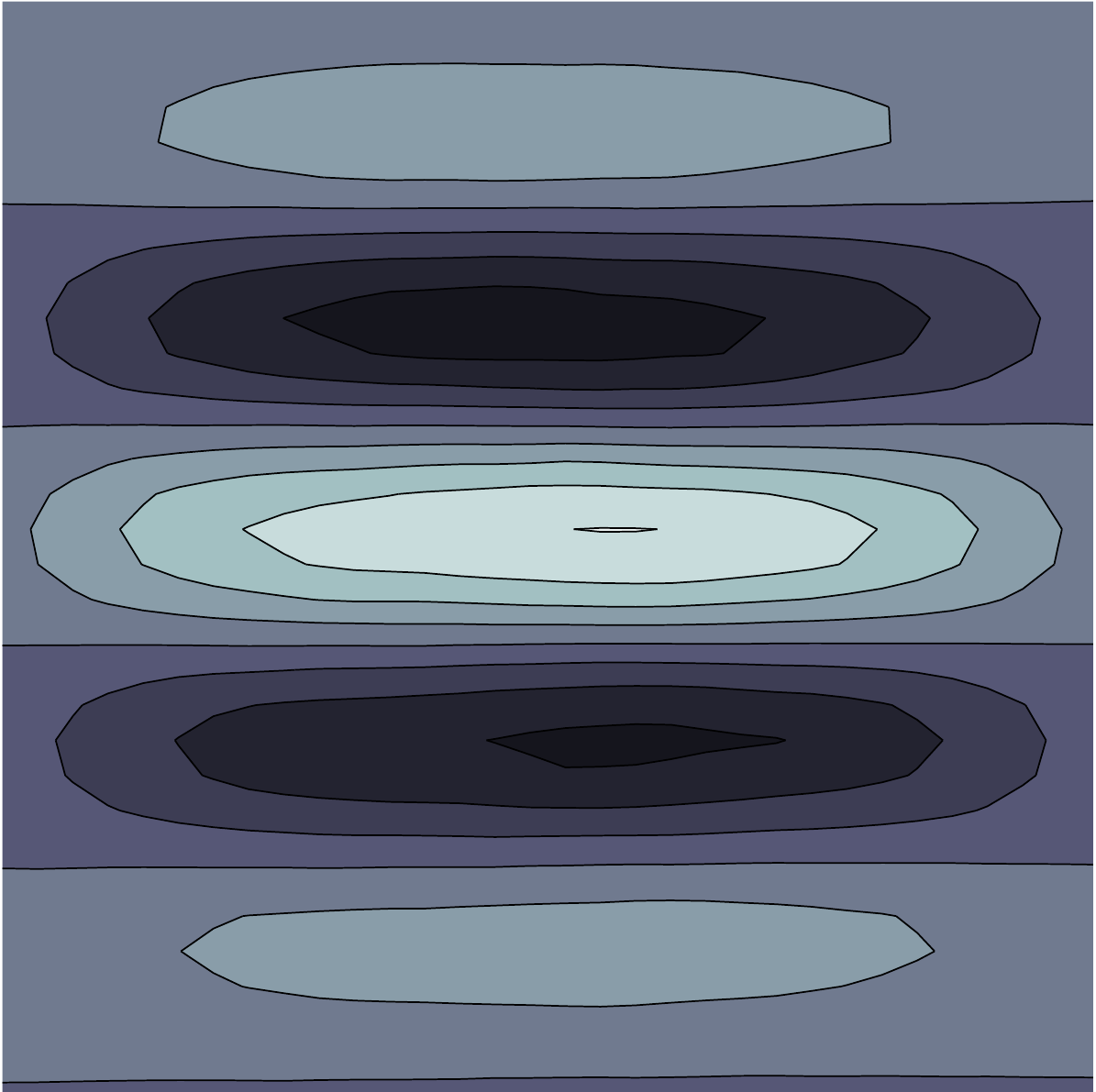}
        \includegraphics[width=0.19\textwidth]{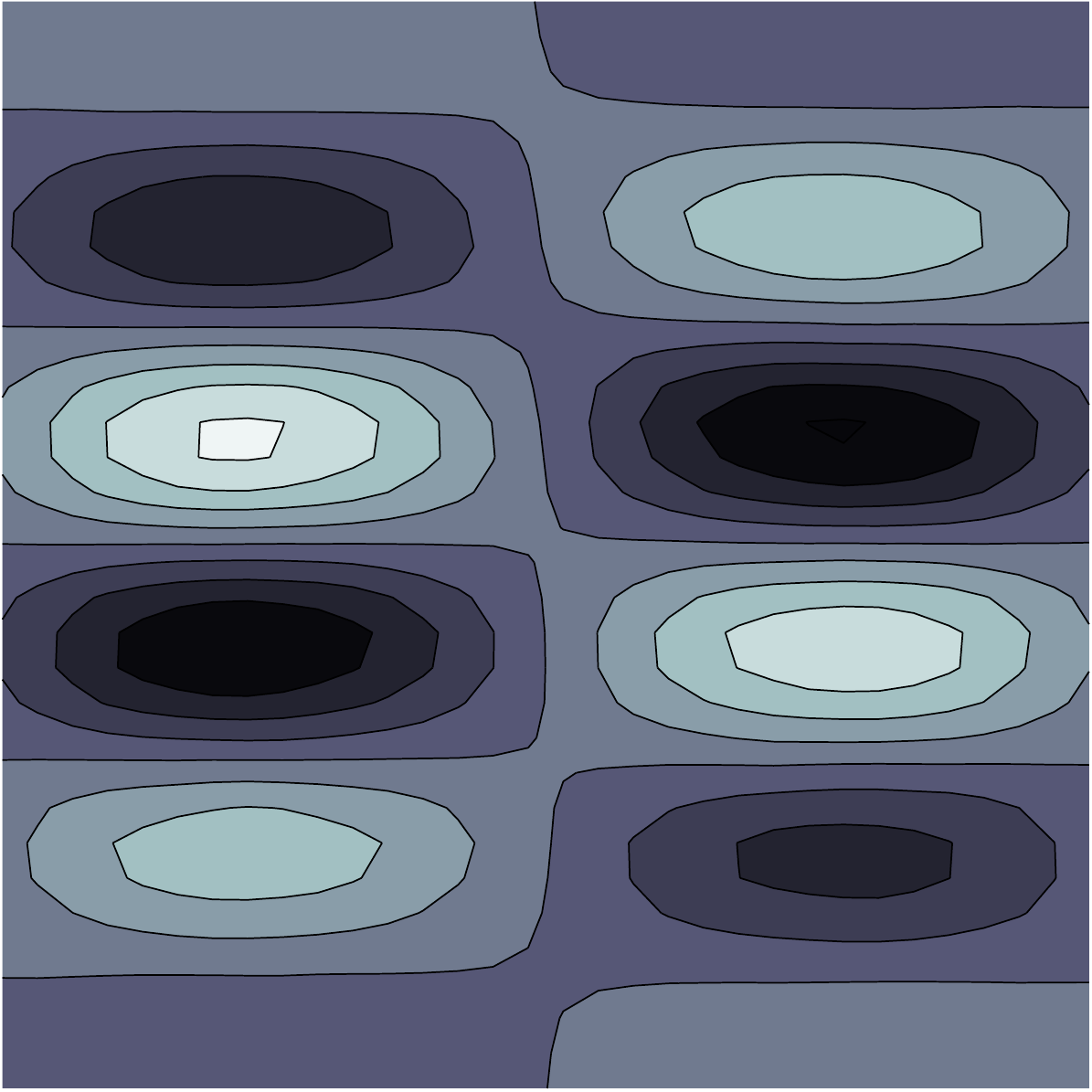}
        \includegraphics[width=0.19\textwidth]{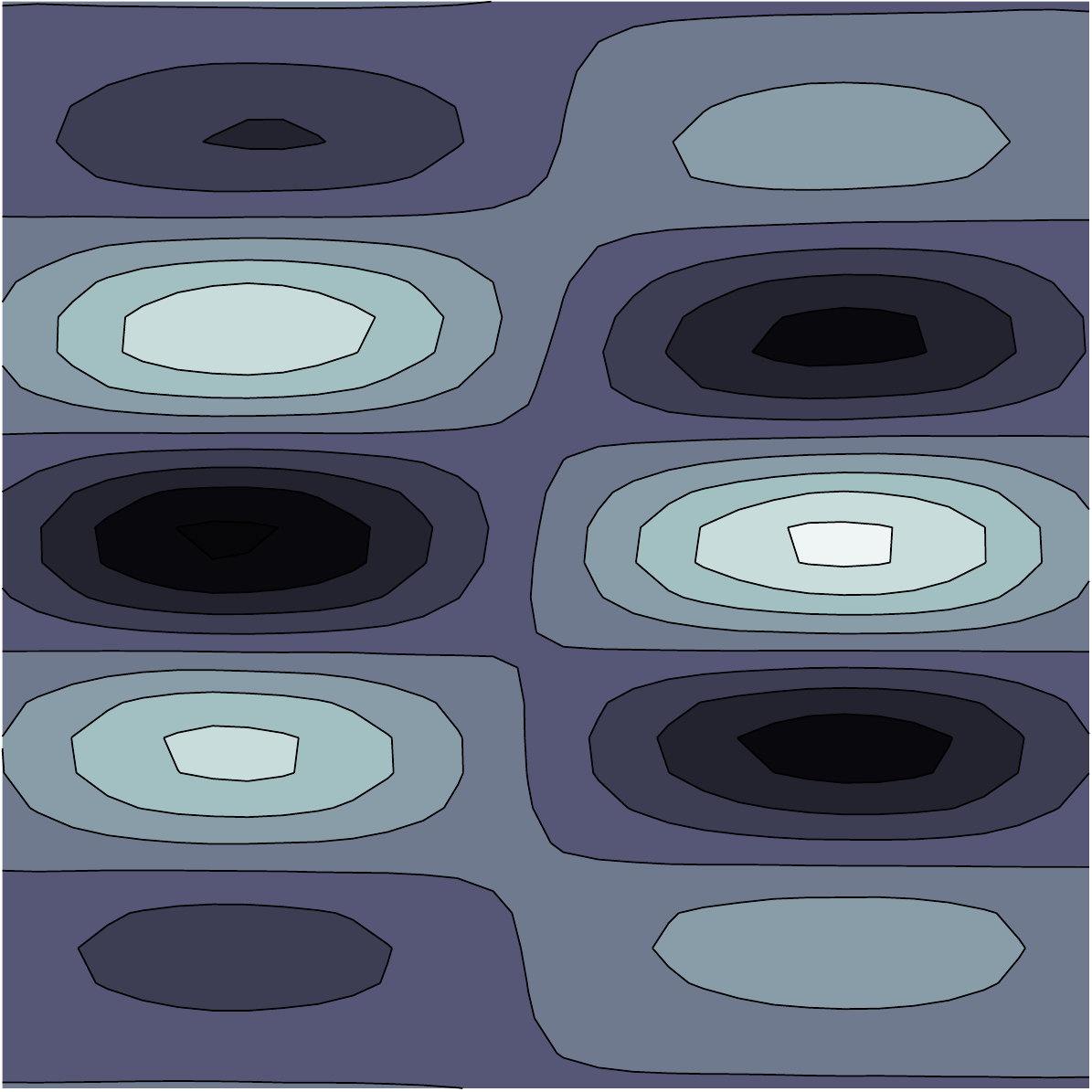}
        \caption{\label{fig:2d_image_eigenvectors_goal_oriented_ypos}}
    \end{subfigure}
    \caption{{Imaging example:} The four leading eigenvectors of the data diagnostic matrix $H_\Y$ for informing chosen parameters: (a) $X_1$, and (b) $X_2$. \label{fig:2d_image_eigenvectors_goal_oriented}}
\end{figure}

\subsection{Conditioned diffusion} \label{sec:Exp_CondDiffusion}

Now we consider a high-dimensional non-Gaussian inference problem motivated by applications in molecular dynamics. Our goal is to infer the driving force on a particle diffusing in a double-well potential, given a noisy observation of its path; see~\citet{cui2016dimension}. The particle's position is described by a function $u\colon [0,1] \rightarrow \mathbb{R}$ which solves the stochastic differential equation
\begin{equation}
\d u_t = f(u_t)\d t + \d X_t, \;\;\; u_0 = 0 .
\end{equation}
Here, $f\colon \mathbb{R} \rightarrow \mathbb{R}$ is the nonlinear drift function $f(u) = \beta u (1 - u^2)/(1 + u^2)$ for $\beta > 0$ and $\d X_t$ is an increment of the Brownian motion $X \sim \mathcal{N}(0,C)$ with covariance function $C(t,t') = \min(t,t')$. 
We set $\beta = 1$ and discretize the ODE using an Euler-Maruyama scheme with time step $\Delta t = 10^{-2}$, so that $d=100$. At $m$ equispaced times $t_1,\dots,t_{m}$ in the interval $[0,1]$, we observe the noisy position of the particle,
\begin{equation}
y_{t_i} = u_{t_i} + \varepsilon_i,
\end{equation}
where $\varepsilon \sim \mathcal{N}(0,\sigma^2 \Id_m)$ is independent of $u$ and $\sigma = 0.1$.
In other words, we have $\Y = G(\X) + \varepsilon$ where $G\colon\R^d\rightarrow\R^m$ is the nonlinear forward model that maps a realization of the noise $\x \in \mathbb{R}^{d}$ to the path $(u_{t_1},\hdots,u_{t_{m}})$. {In our experiments we set $m = d = 100$.} %
Figure~\ref{fig:cd_observations} shows 200 realizations of $\Y$.
We then compute the matrices $H_{\overline\X}$ and $H_{\overline\Y}$ using Algorithm \ref{alg:computeBasis} with $n = 10^6$ samples and plot their leading eigenvalues in Figure~\ref{fig:cd_eigenvalues}. Due to the nonlinear forward model, the eigenvalues of the two matrices are different (in contrast with the linear--Gaussian setting described in Section~\ref{subsec:LinearGaussian}).

\begin{figure}[!ht]
     \centering
     \begin{subfigure}[b]{0.45\textwidth}
         \centering
         \includegraphics[width=\textwidth]{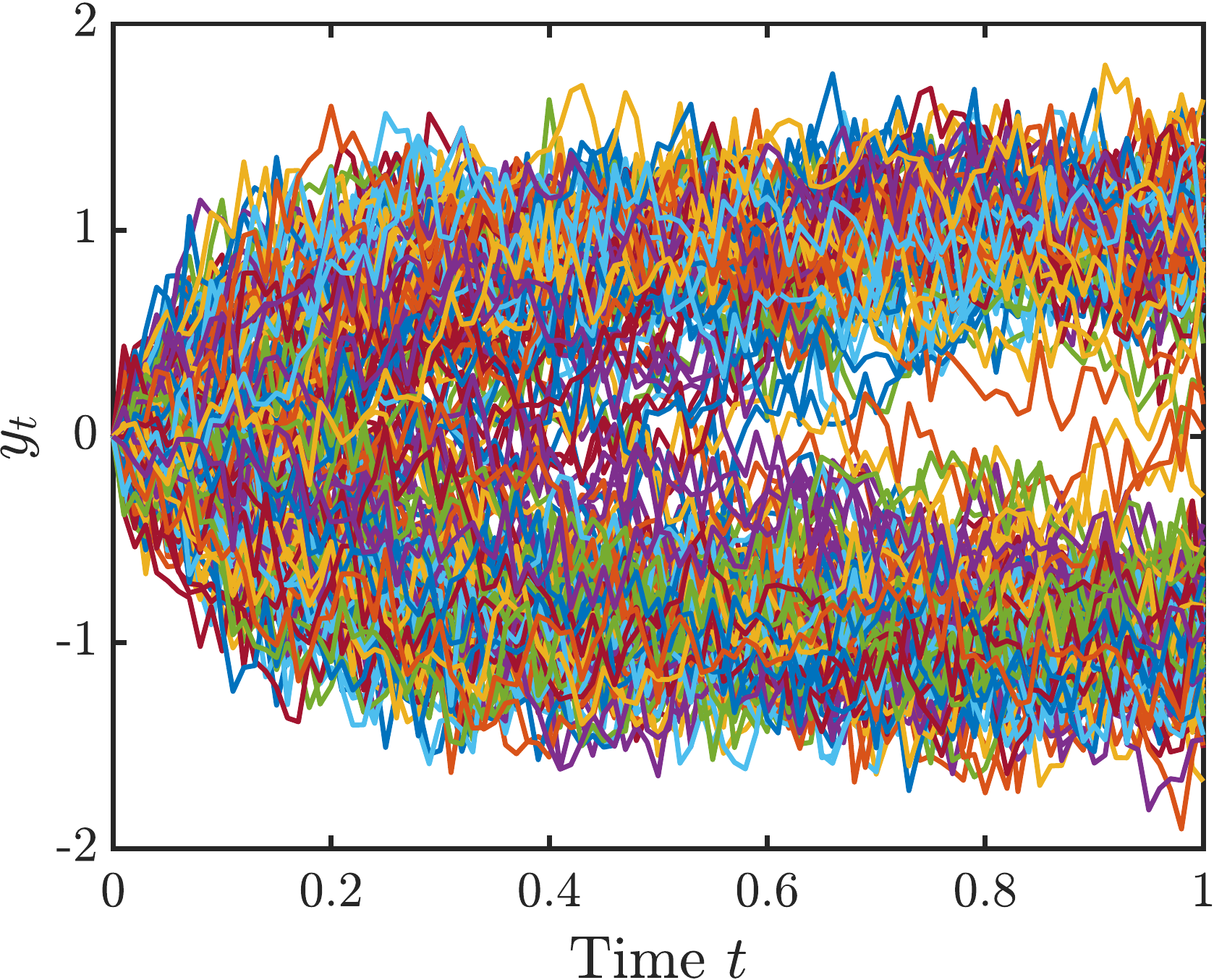}
         \caption{\label{fig:cd_observations}}
     \end{subfigure}
     \hspace{1cm}
     \begin{subfigure}[b]{0.45\textwidth}
         \centering
         \includegraphics[width=1.05\textwidth]{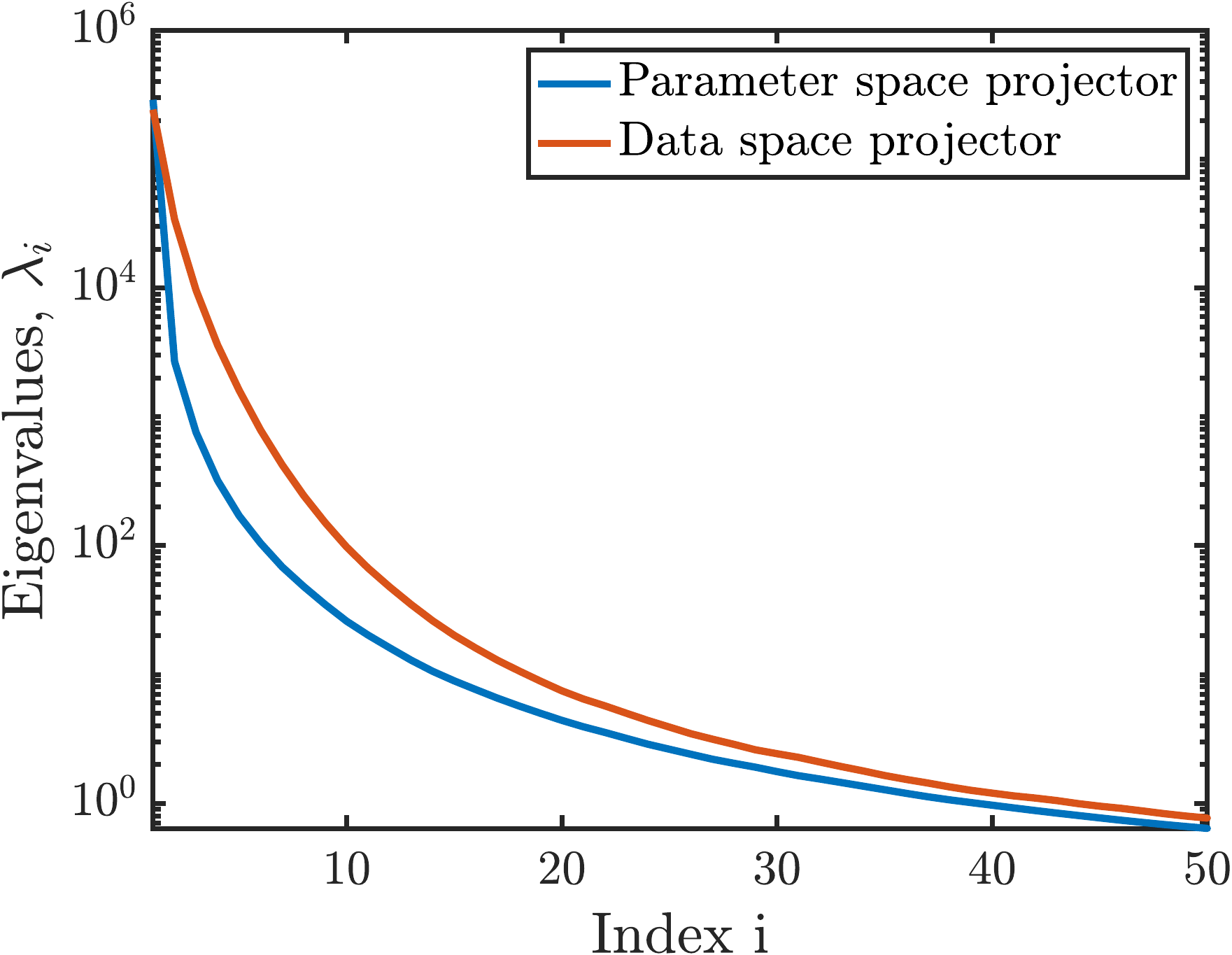}
         \caption{\label{fig:cd_eigenvalues}}
     \end{subfigure}
     \caption{{Conditioned diffusion model: (a) 200 samples of data $\ys$. (b) Leading $50$ eigenvalues of the diagnostic matrices $H_{\overline\X}$ and $H_{\overline\Y}$.}}
\end{figure} %

Figure~\ref{fig:cd_eigv_CMI} plots the parameter and data space eigenvectors corresponding to the five leading eigenvalues of $H_{\overline\X}$ and $H_{\overline\Y}$, respectively. The parameter-space eigenvectors capture more of the sample path behavior near $t=0$, while the data-space eigenvectors capture the data behavior near $t = 1$. {From the realizations in Figure~\ref{fig:cd_observations}, we also see that the particle tends to settle in one well or the other relatively early in time; hence the data seem to be most informative about earlier portions of the force trajectory, while the driving force is most informed by the particle's position near the final time $t = 1$.}
For contrast, we plot the parameter- and data-space eigenvectors obtained via PCA and CCA 
(again using $n=10^6$ samples)
in Figures~\ref{fig:cd_data_eigv_PCA} and~\ref{fig:cd_data_eigv_CCA}, respectively. 
We observe that PCA modes are more globally supported than those obtained from the CMI bound---i.e., less focused on early or late portions of the trajectory. The CCA modes are more irregular. %

\begin{figure}[!ht]
     \centering
     \begin{subfigure}[b]{0.45\textwidth}
         \centering
         \includegraphics[width=\textwidth]{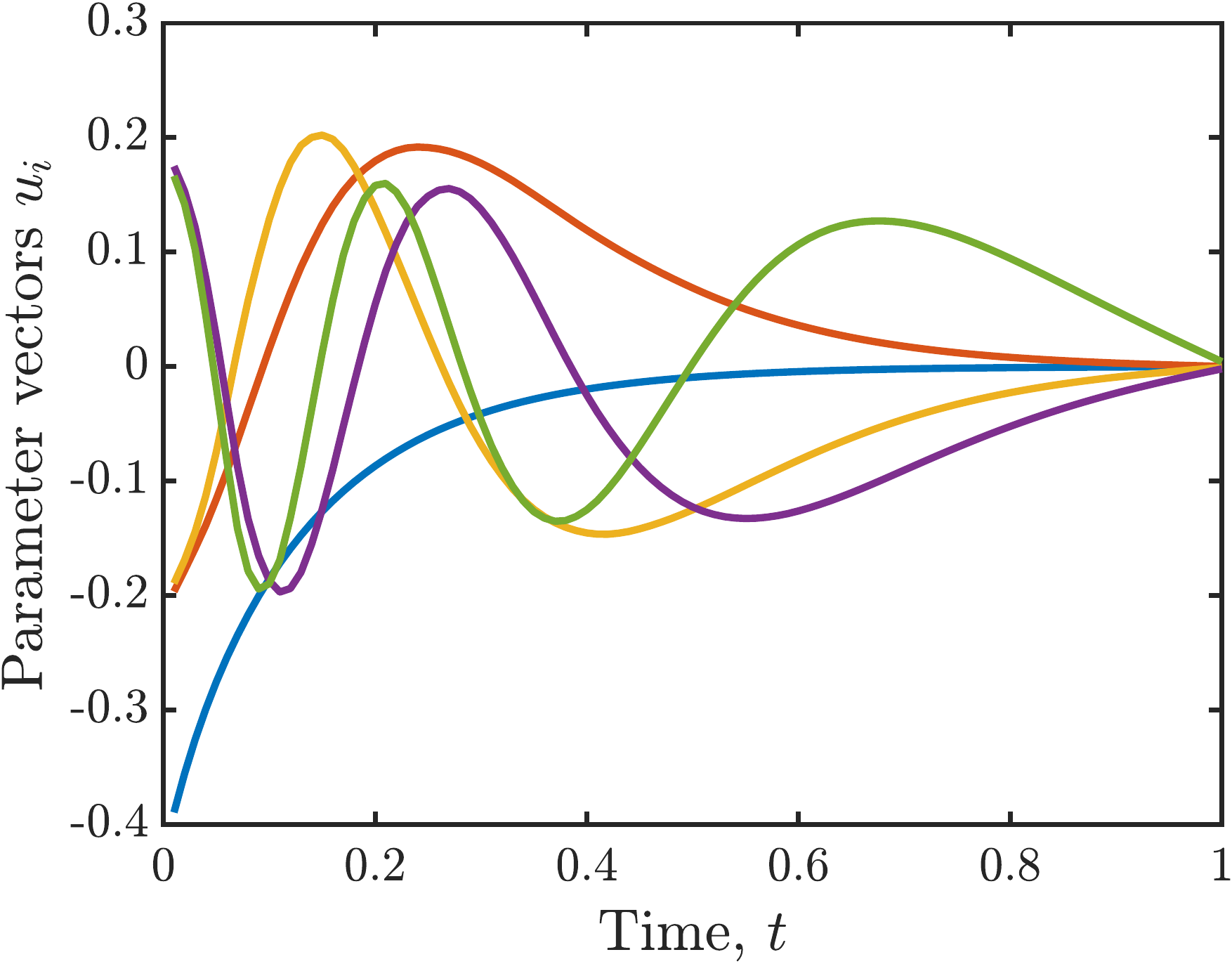}
         \caption{}
         \label{fig:cd_param_eigv_CMI}
     \end{subfigure}
     \hspace{1cm}
     \begin{subfigure}[b]{0.45\textwidth}
         \centering
         \includegraphics[width=0.98\textwidth]{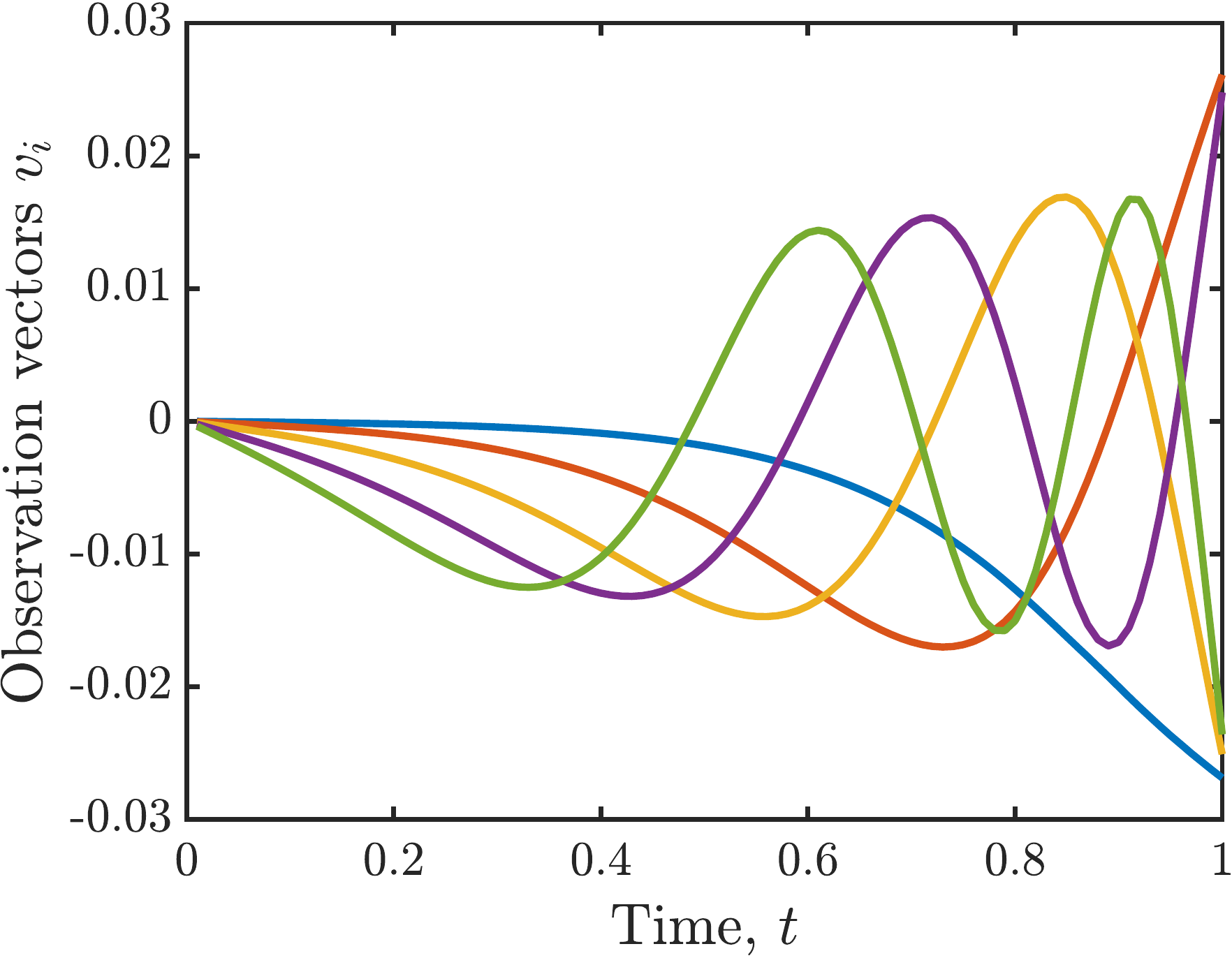}
         \caption{}
         \label{fig:cd_data_eigv_CMI}
     \end{subfigure}
     \caption{{Conditioned diffusion model}: (a) Parameter space and (b) data space eigenvectors of the diagnostic matrices $H_{\overline\X}$ and $H_{\overline\Y}$ \label{fig:cd_eigv_CMI}}
\end{figure}

\begin{figure}[!ht]
     \centering
     \begin{subfigure}[t]{0.45\textwidth}
         \centering
         \includegraphics[width=\textwidth]{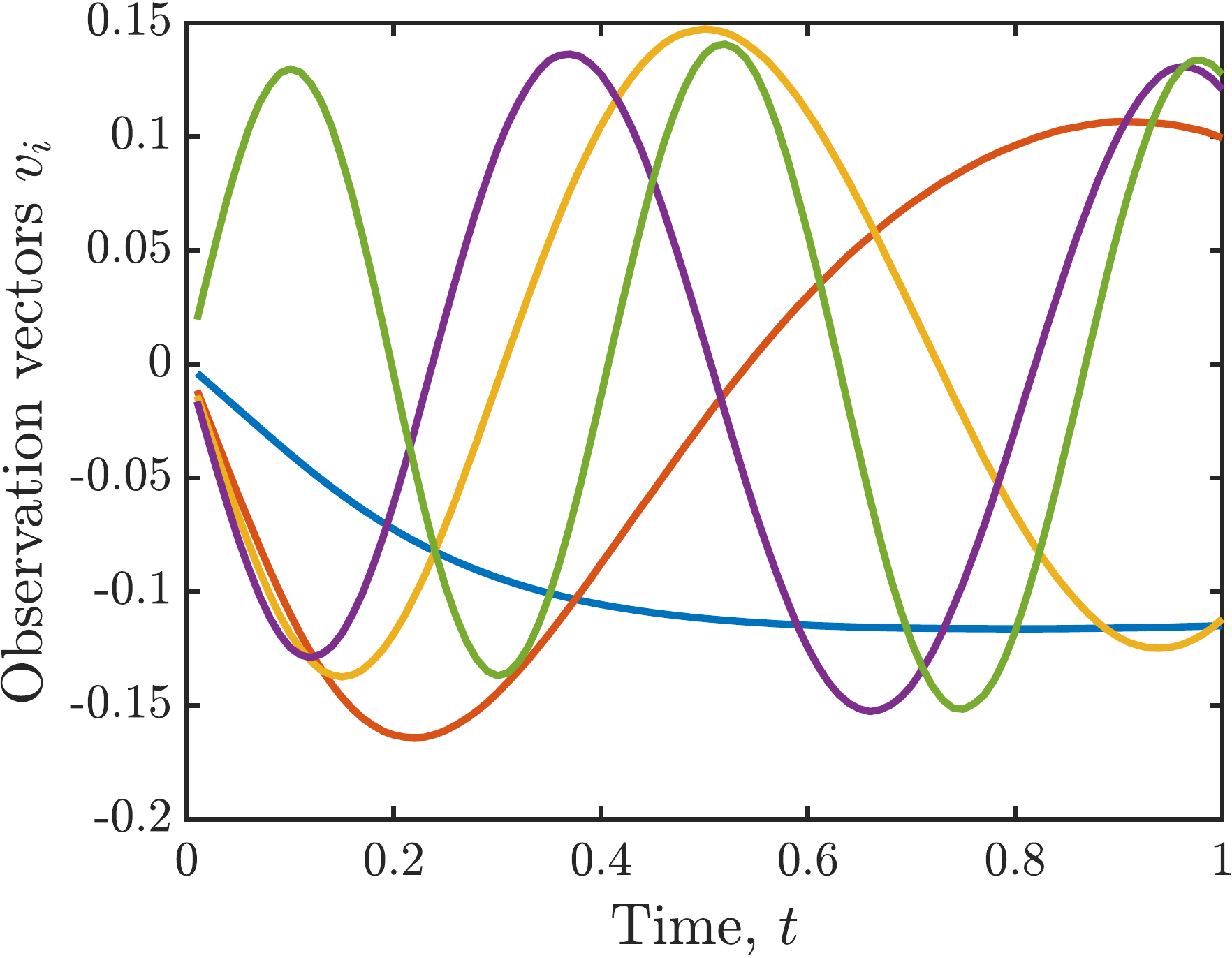}
         \caption{}
         \label{fig:cd_data_eigv_PCA}
     \end{subfigure}
     \hspace{1cm}
     \begin{subfigure}[t]{0.45\textwidth}
         \centering
         \includegraphics[width=\textwidth]{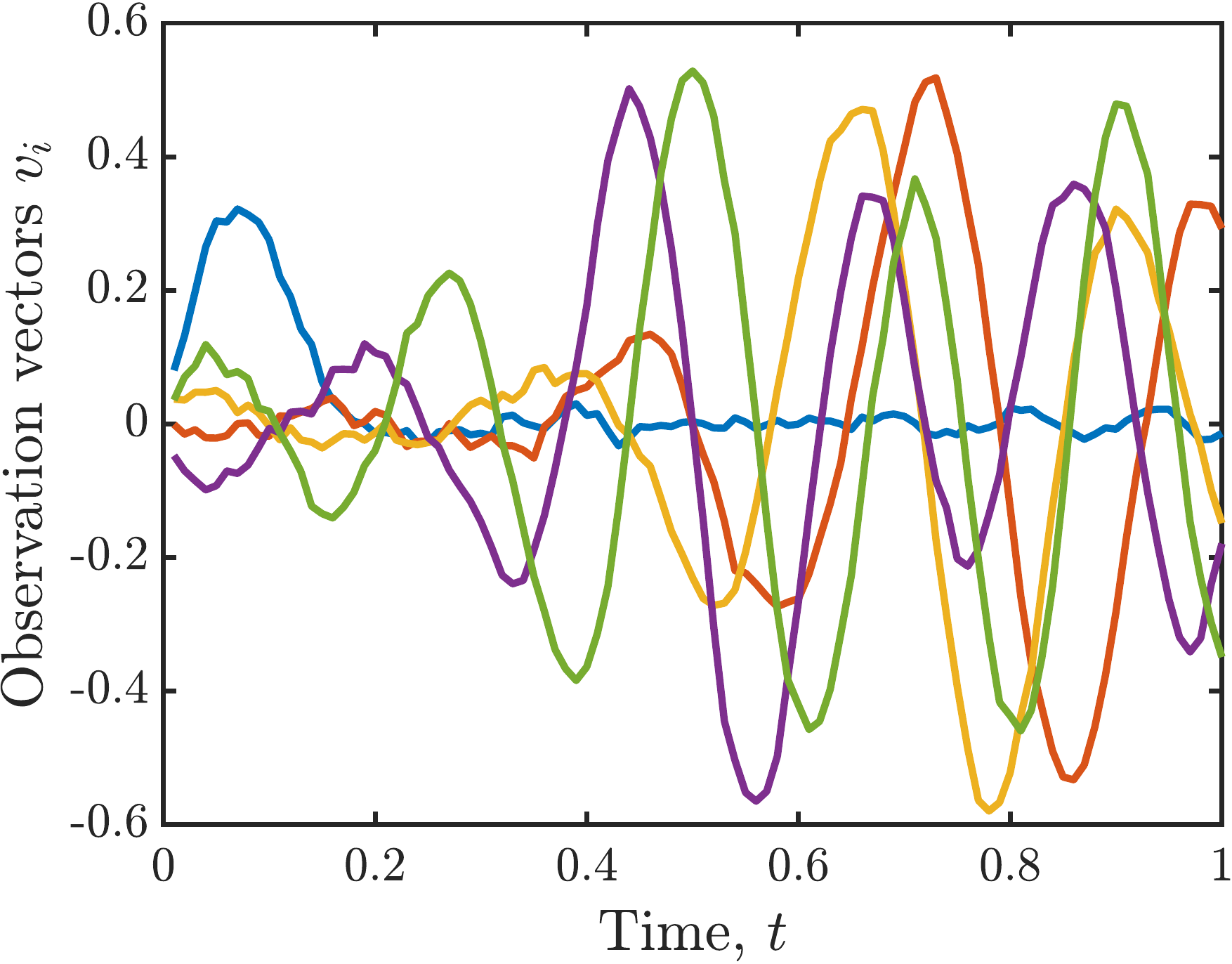}
         \caption{}
         \label{fig:cd_data_eigv_CCA}
     \end{subfigure}
     \caption{{Conditioned diffusion model}: Data eigenmodes from (a) PCA and (b) CCA.}
\end{figure}

To evaluate the approximation errors resulting from these parameter and data-space projectors, we estimate the conditional mutual information for projected parameters and data. To estimate the conditional mutual information $I(\X_\perp;\Y|\X_r)$, we generate $n = 10^4$ samples $(\xs^i,\ys^i) \sim \Joint$ and construct the Monte Carlo estimator
\begin{equation} \label{eq:MCest_MIparam}
\widehat{I}(\X_\perp;\Y|\X_r) = \frac{1}{n} \sum_{i=1}^{n} \log \frac{\pi_{\Y|\X}(\ys^i|\xs^i)}{\pi_{\Y|\X_r}(\ys^i|\xs_r^i)},
\end{equation}
where $\pi_{\Y|\X_r}(\y|\x_r) = \int \pi_{\Y|\X}(\y|\BasisX_r\x_r + \BasisX_\perp\x_\perp)\pi_{\X_\perp|\X_r}(\x_\perp|\x_r)\d \x_\perp$. Analogously to the estimator in~\eqref{eq:MCestimator_RedLikelihood} for the reduced likelihood $\pi_{\Y_s|\X_r}$, we estimate $\pi_{\Y|\X_r}$ using the following Monte Carlo estimator, given $\ell$ %
samples from the conditional prior
\begin{equation} \label{eq:cond_exp_likelihood}
\widehat\pi_{\Y|\X_r}(\y|\x_r) = \frac{1}{\ell} \sum_{j=1}^{\ell} \pi_{\Y|\X}(\y|\BasisX_r \x_r + \BasisX_\perp \xs_\perp^j), \quad \xs_\perp^j \sim \pi_{\X_\perp|\X_r}(\cdot|\x_r).
\end{equation}
To check the impact of $\ell$ on estimating the conditional mutual information, Figure~\ref{fig:cd_parameter_errors} plots the estimator with $\ell \in \{10,100,1000\}$ as well as a single sample, i.e., $l = 1$, at the prior mean $\int \x_\perp \d\pi_{\X_\perp|\X_r}$. %
We observe a convergence of the estimators in~\eqref{eq:cond_exp_likelihood} with increasing $m$. 
Furthermore, the CMI closely matches the trend for the upper bound (up to the subspace log-Sobolev constant), which indicates that for this example the bound can be used as a good error indicator even without knowing $\overline{C}(\pi_{\overline\X,\overline\Y})$.

Next, to estimate the conditional mutual information $I(\Y_\perp;\X|\Y_s) = I(\X;\Y) - I(\X;\Y_s)$, we generate $n$ samples $(\xs^i,\ys^i) \sim \Joint$ and construct the Monte Carlo estimator
\begin{equation} \label{eq:MCest_MIdata}
\widehat{I}(\Y_\perp;\X|\Y_s) = \frac{1}{n} \sum_{i=1}^{n} \log \frac{\pi_{\Y|\X}(\ys^i|\xs^i)}{\pi_{\Y}(\ys^i)} - \log \frac{\pi_{\Y_s|\X}(\ys_s^i|\xs^i)}{\pi_{\Y_s}(\ys_s^i)},
\end{equation}
where the marginal likelihoods $\pi_{\Y}$ and $\pi_{\Y_s}$ are estimated using $\ell$  %
prior samples for each sample $\ys^i$, e.g., $\widehat\pi_{\Y}(\ys^i) = \frac{1}{\ell} \sum_{j=1}^{\ell} \pi_{\Y|\X}(\ys^i|\xs^j)$ for $\xs^j \sim \Prior$. To compute the likelihood for the reduced data $\Y_s$, we analytically marginalize the Gaussian likelihood by projecting the mean and covariance of the observational noise using the formula in Example~\ref{ex:AnalyticalMargLikelihood}. This avoids an additional numerical integration. Figure~\ref{fig:cd_data_errors} plots the estimates of the mutual information for the projected data with increasing reduced dimension, along with the upper bound in~\eqref{eq:PostErr_Subspaces} (up to the subspace log-Sobolev constant). We observe that the estimators converge with increasing sample size $\ell$. Furthermore, the upper bound closely matches the trend for the true approximation error, especially for larger $s$. %

\begin{figure}[!ht]
     \centering
     \begin{subfigure}[t]{0.45\textwidth}
         \centering
         \includegraphics[width=\textwidth]{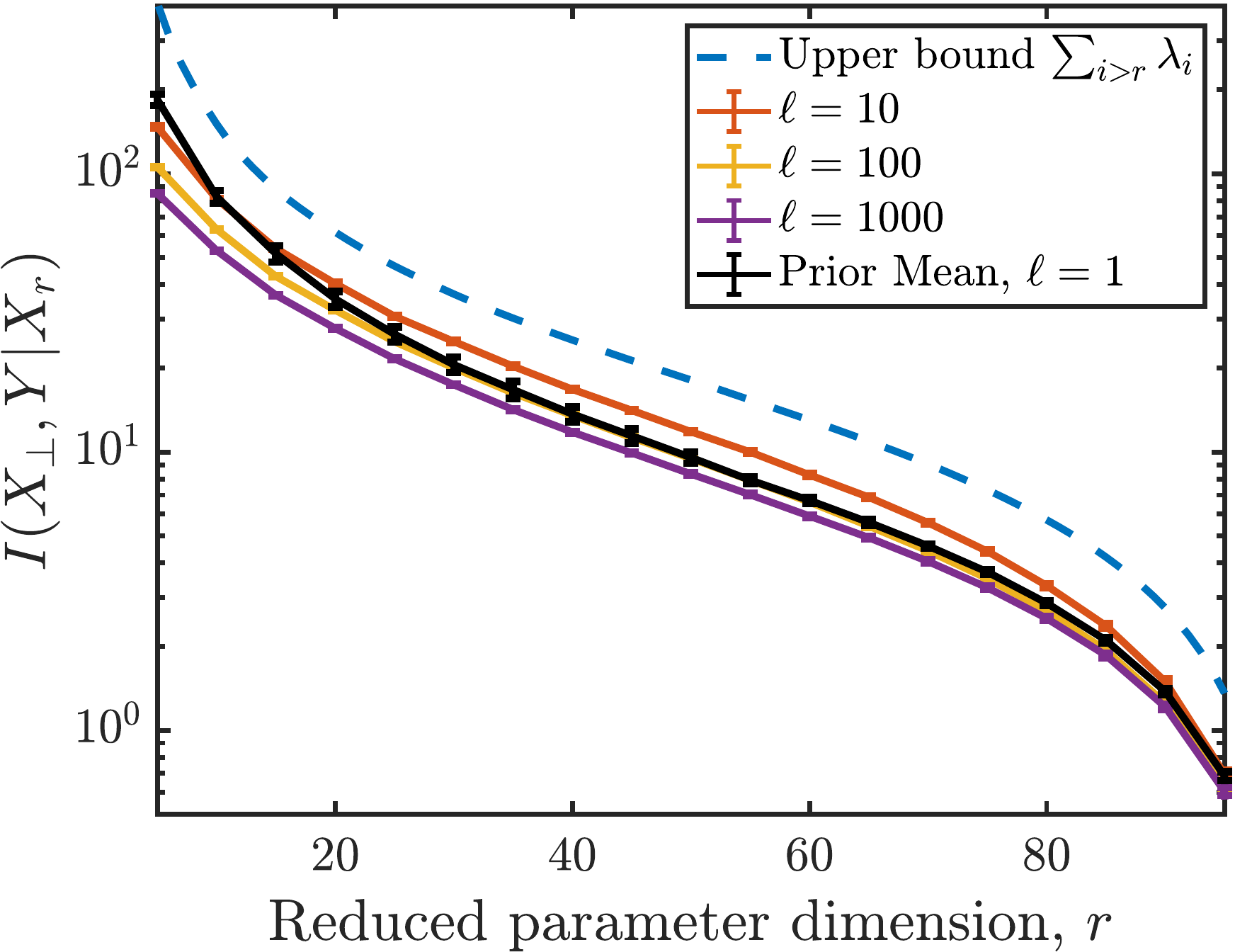}
         \caption{}
         \label{fig:cd_parameter_errors}
     \end{subfigure}
     \hspace{1cm}
     \begin{subfigure}[t]{0.45\textwidth}
         \centering
         \includegraphics[width=\textwidth]{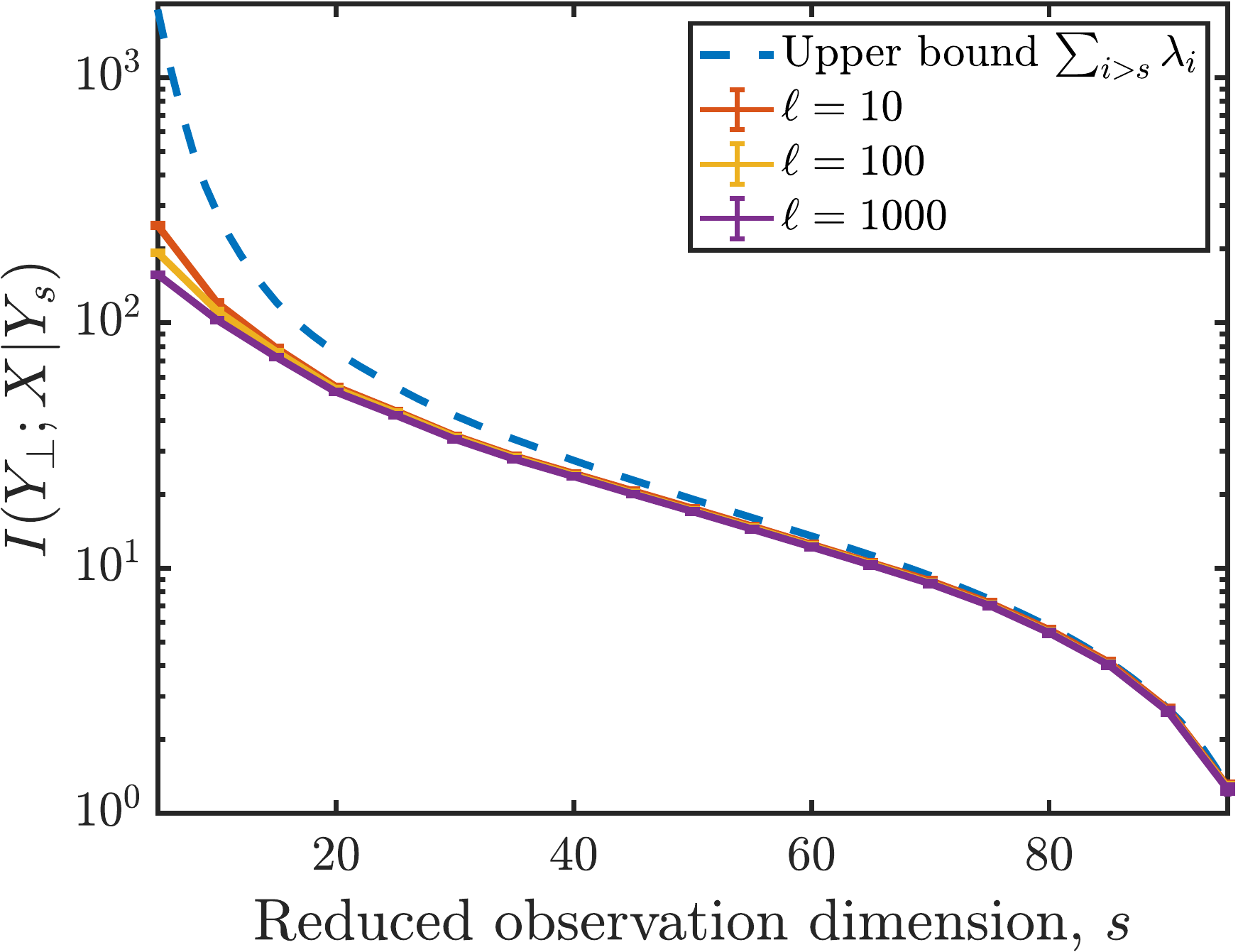}
         \caption{}
         \label{fig:cd_data_errors}
     \end{subfigure}
     \caption{{Conditioned diffusion model: Convergence of the estimators for the conditional mutual information with projections of (a) the parameter and (b) the data.}}
\end{figure}

We also compare the optimal {subspaces}
identified from the upper bound in~\eqref{eq:KL_UpperBound} to the subspaces resulting using PCA and CCA. Figures~\ref{fig:cd_parameter_compare_projector} and~\ref{fig:cd_data_compare_projectors} plot the mutual information representing the posterior approximation error in expected KL divergence for increasing reduced dimensions of the parameter and data, respectively. The mutual information is computed for each dimension $r,s$ using the Monte Carlo estimators in~\eqref{eq:MCest_MIparam} and~\eqref{eq:MCest_MIdata} with $n = 10^4$ and $\ell = 100$. The CMI bounds present the lowest error for the projection of the  parameters. We note that PCA performs similarly to the CMI bound for data reduction in this example, despite having very different modes; see Figures~\ref{fig:cd_data_eigv_CMI} and~\ref{fig:cd_data_eigv_PCA}. Together with the parameters, however, the {subspaces} identified from the CMI bound provide the lowest posterior approximation error.

\begin{figure}[!ht]
     \centering
     \begin{subfigure}[t]{0.45\textwidth}
         \centering
         \includegraphics[width=\textwidth]{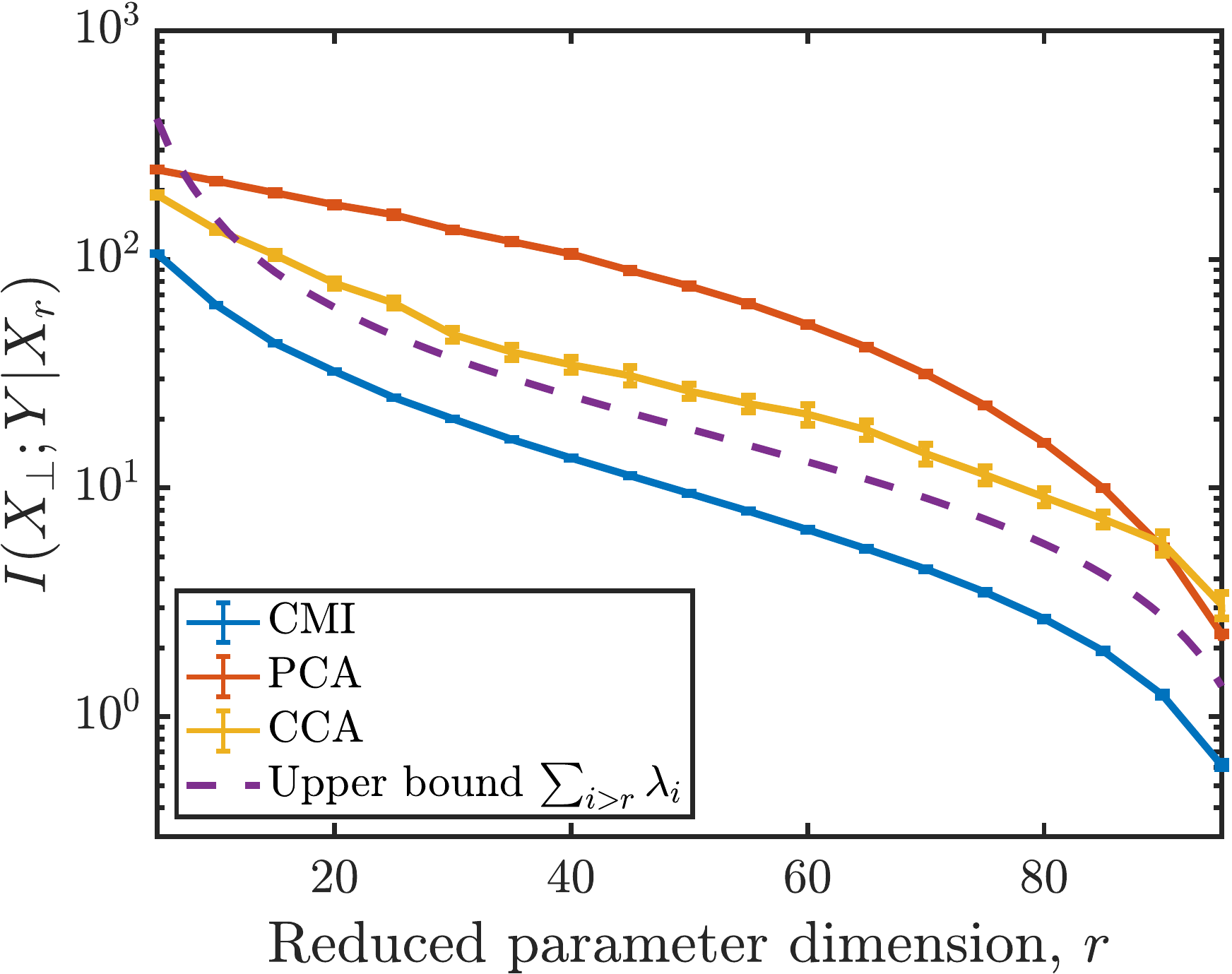}
         \caption{}
         \label{fig:cd_parameter_compare_projector}
     \end{subfigure}
     \hspace{1cm}
     \begin{subfigure}[t]{0.45\textwidth}
         \centering
         \includegraphics[width=\textwidth]{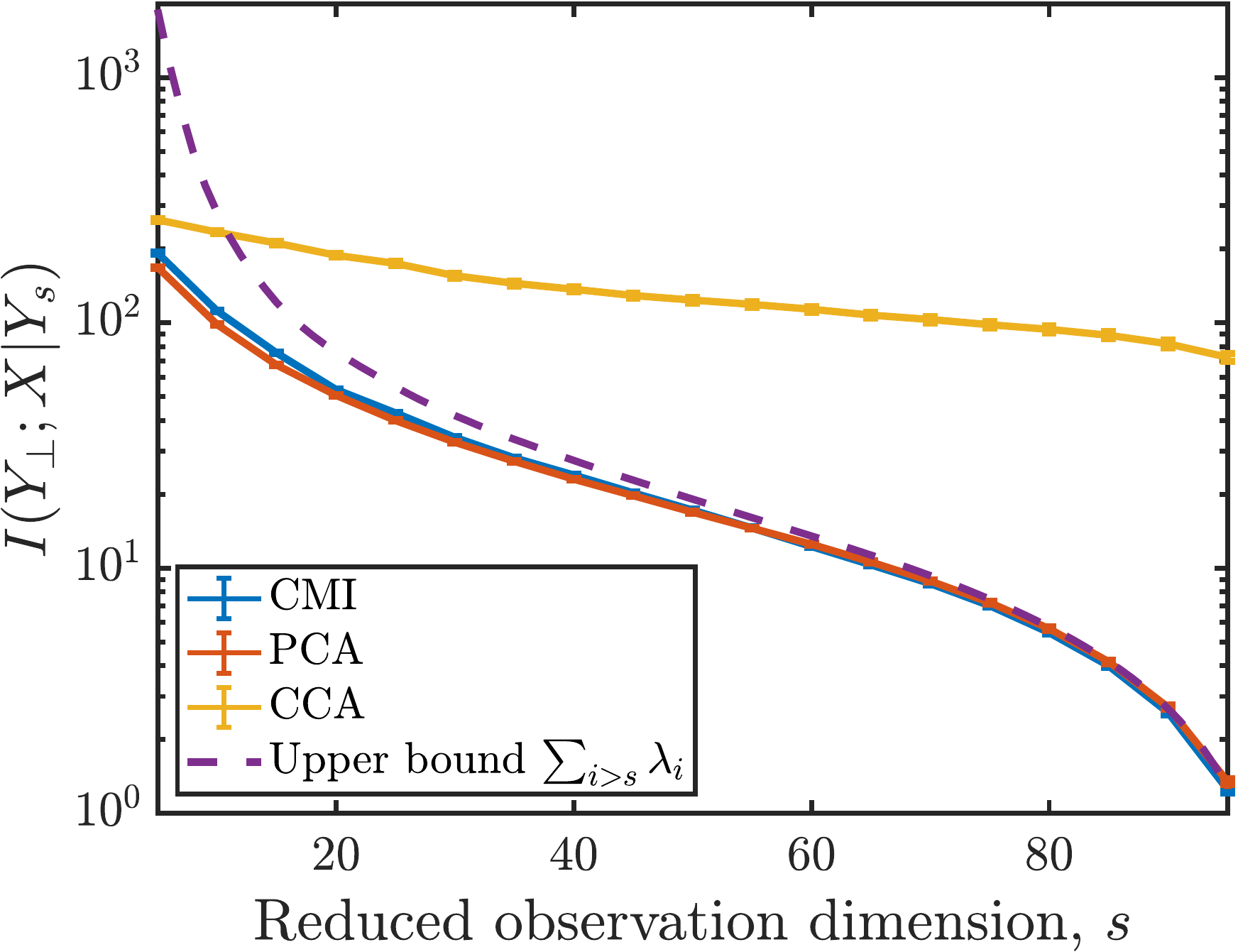}
         \caption{}
         \label{fig:cd_data_compare_projectors}
     \end{subfigure}
     \caption{{Conditioned diffusion model: Comparison of three dimension reduction strategies for reducing the dimension of the (a) parameter and (b) data using the expected KL divergence, or equivalently the conditional mutual information.}}
\end{figure}

We also consider the problem of parameter and data \emph{coordinate selection}. As presented in Section~\ref{subsec:OptPermutation}, this corresponds to sorting the diagonal entries of the diagnostic matrices $H_{\X}$ and $H_{\Y}$ in decreasing order. Figure~\ref{fig:cd_diagnostic_diagonal_entries} plots the values of these diagonal entries in their canonical ordering, i.e., with increasing time. The forcing at the initial time, i.e., $t \approx 0$, is most informed component, while the most informative data are observations of the particle position near the final time $t \approx 1$. To compare the accuracy of posterior approximations obtained via coordinate selection to posterior approximations built on the subspaces found above,  Figure~\ref{fig:cd_diagnostic_comp_upper_bounds} plots the upper bounds for the expected KL divergence with rotated or selected data (up to the subspace log-Sobolev constant). For this example, we observe that optimal rotations yield an improvement of at least two orders of magnitude and converge at a faster rate than optimal coordinate selections, particularly for lower dimensional data $\Y_s$. %

\begin{figure}[!ht]
     \centering
     \begin{subfigure}[t]{0.45\textwidth}
        \centering
        \includegraphics[width=\textwidth]{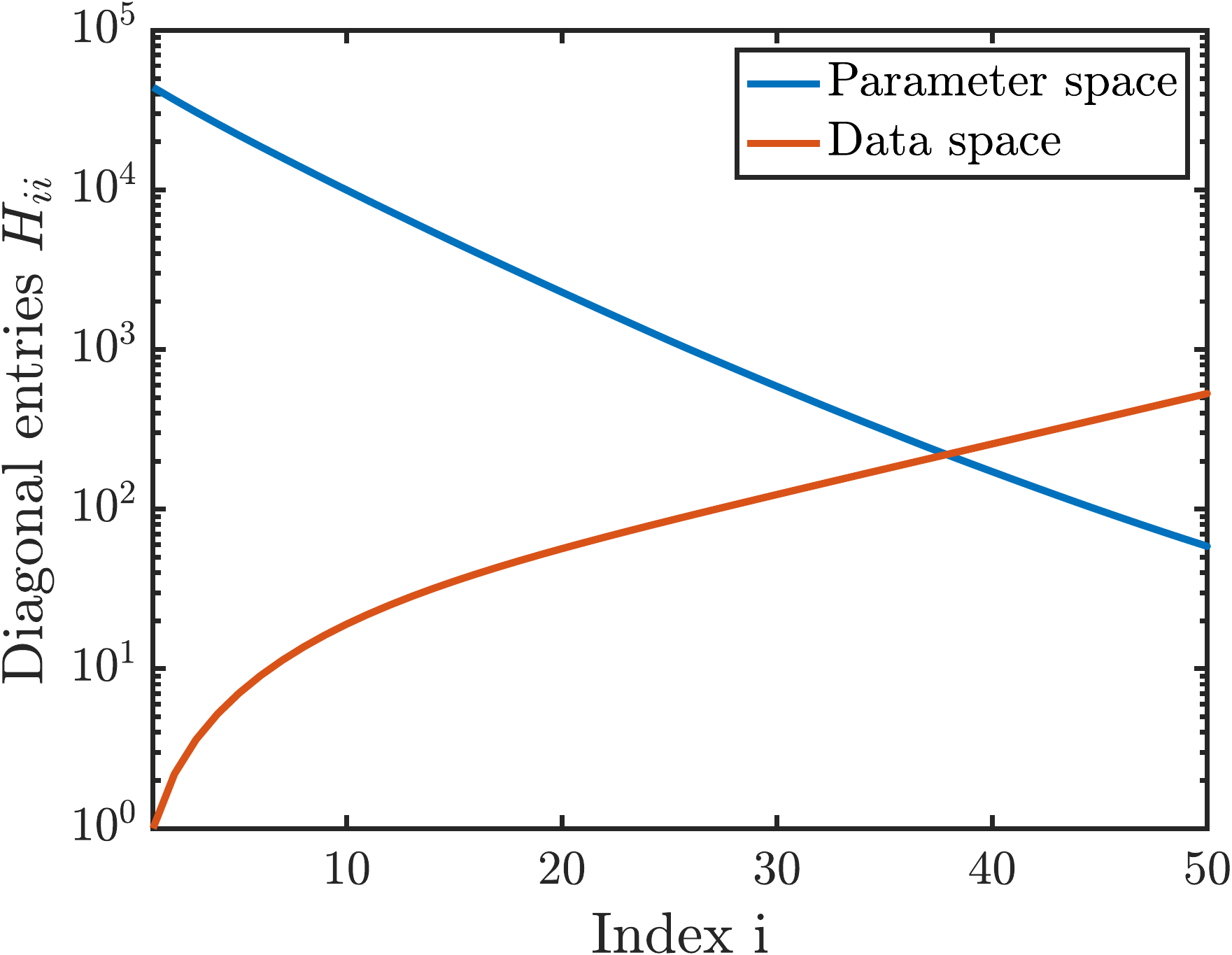}
        \caption{}
        \label{fig:cd_diagnostic_diagonal_entries}
    \end{subfigure}
    \hspace{1cm}
    \begin{subfigure}[t]{0.45\textwidth}
        \centering
        \includegraphics[width=\textwidth]{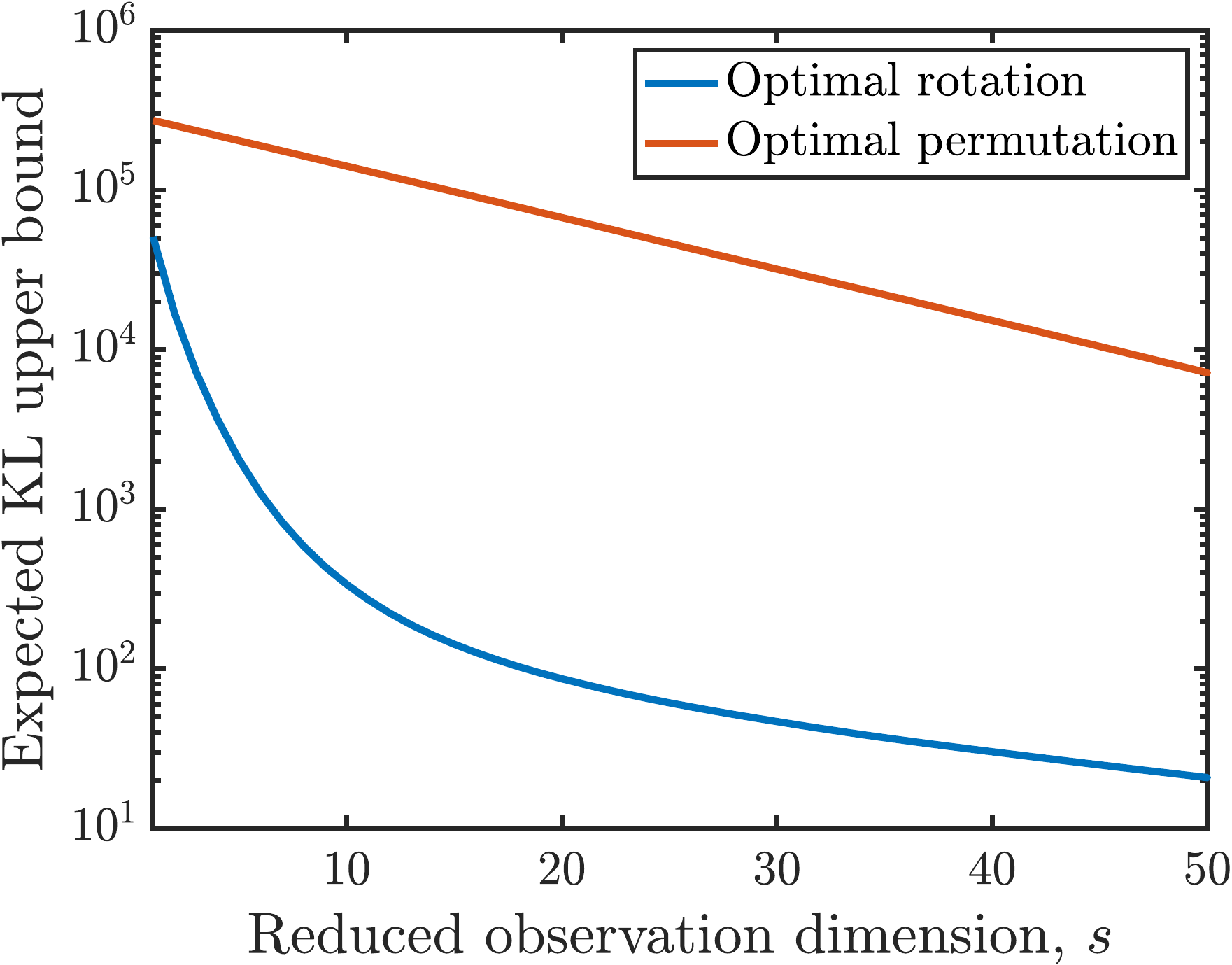}
        \caption{}
        \label{fig:cd_diagnostic_comp_upper_bounds}
    \end{subfigure}
    \caption{{Conditioned diffusion model: (a) Diagonal entries of the diagonal matrices $H_{\X}$ and $H_{\Y}$ without reordering. (b) Comparison of the upper bound in~\eqref{eq:KL_UpperBound}  for optimal rotations and coordinate selection of the data.}}%
\end{figure}

Lastly, we demonstrate the value of dimension reduction in accelerating Markov chain Monte Carlo (MCMC) methods for posterior sampling. Specifically, we assess the mixing of the dimension-independent likelihood-informed (DILI) sampler of~\citet{cui2016dimension} when projecting the data. For reduced data dimensions $s \in [10,100]$, we generate $5 \times 10^5$ approximate samples from the posterior density $\pi_{\X|\Y_s} \propto \pi_{\Y_s|\X}\pi_{\X}$ with the data-marginalized likelihood in~\eqref{eq:tmp385}. Figure~\ref{fig:cd_MCMCresults} plots the integrated autocorrelation time (IACT) of the samples and the relative $L^2$ errors in the posterior mean and standard deviation as a function of the dimension of the reduced observations, $s$. We observe that projecting the observations yields a significant reduction in IACT, while keeping errors in the posterior mean and posterior variance reasonably low. %

\begin{figure}[ht!]
     \centering
     
    \begin{subfigure}[t]{0.49\textwidth}
        \centering
        \includegraphics[width=\textwidth]{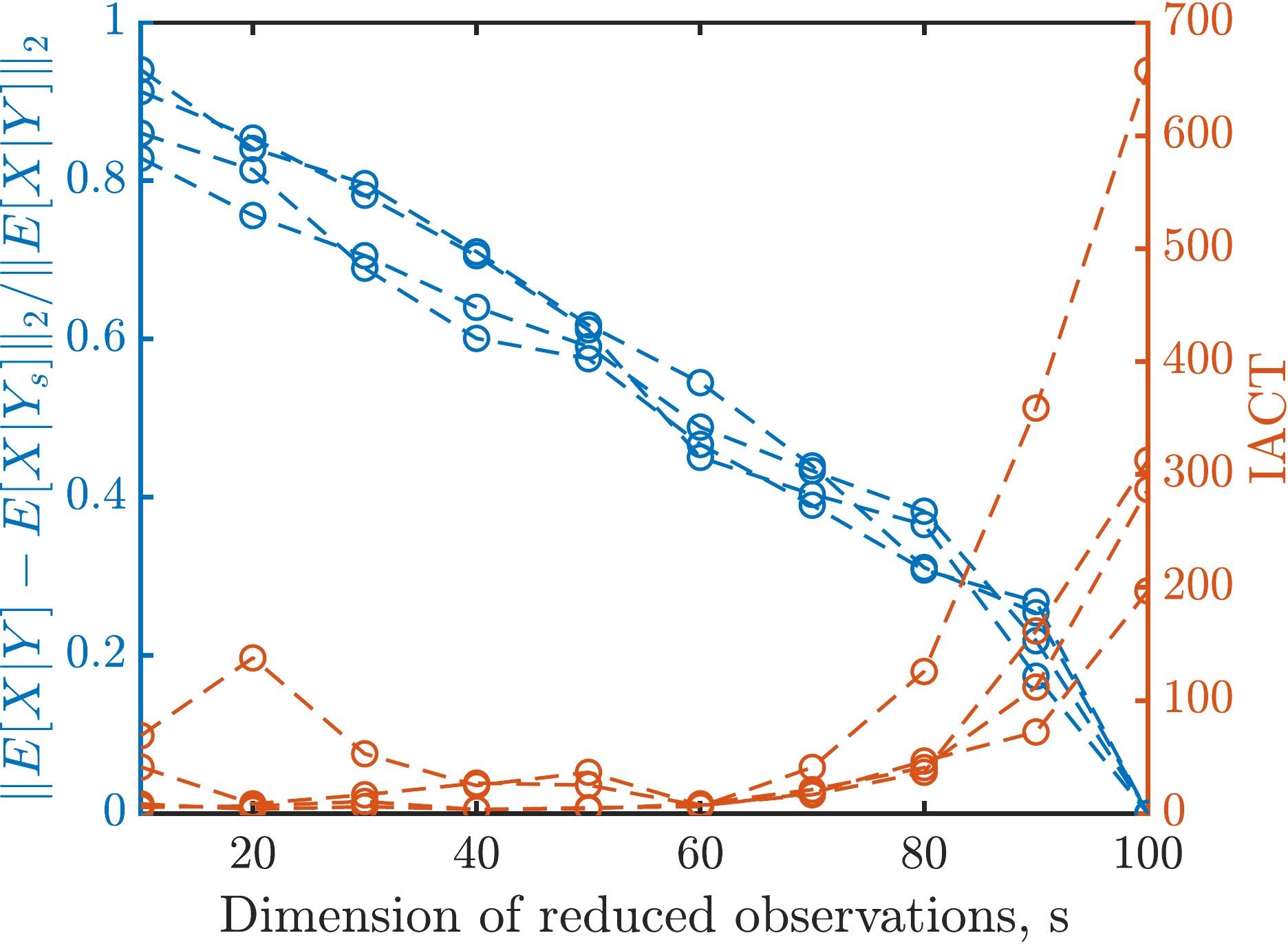}
        \caption{}
        \label{fig:cd_MCMCresults_a}
    \end{subfigure}
    \hspace{0.05cm}
    \begin{subfigure}[t]{0.49\textwidth}
        \centering
        \includegraphics[width=\textwidth]{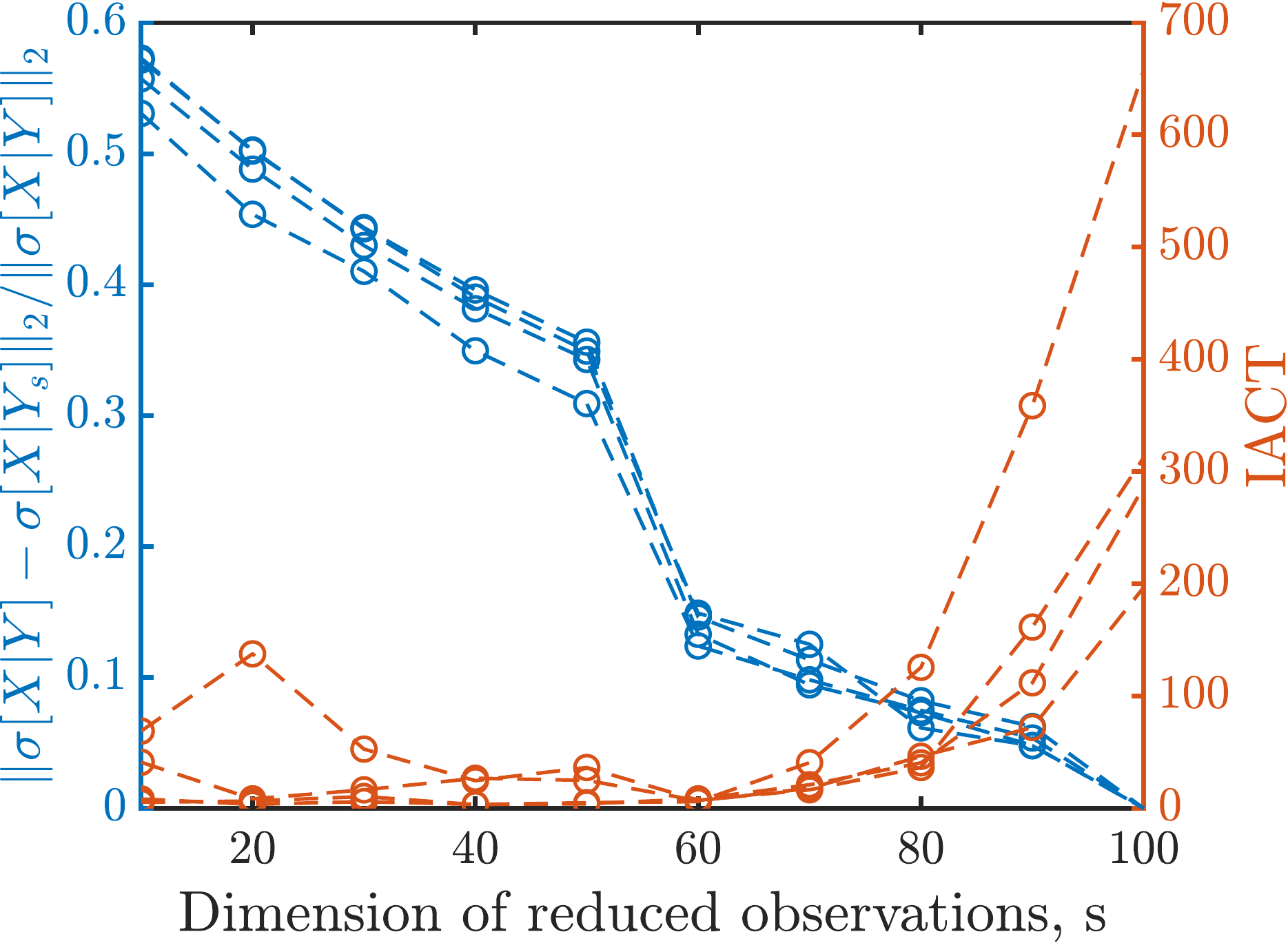}
        \caption{}
        \label{fig:cd_MCMCresults_b}
    \end{subfigure}
    \caption{Conditioned diffusion example: In blue, relative $L^2$ error between (a) the posterior mean $\mathbb{E}[\X|\Y]$ and its approximation $\mathbb{E}[\X|\Y_s]$; and (b) the posterior standard deviation $\sigma[\X|\Y]$ and its approximation $\sigma[\X|\Y_s]$. In red, integrated autocorrelation time of the MCMC sampler used to explore the posterior. Each of the 5 lines corresponds to a different realization of the data $Y$.
    \label{fig:cd_MCMCresults}}
\end{figure}

\section{Discussion}

This work proposes a gradient-based method for concurrently reducing the dimensionality of the data and parameters of Bayesian models, in a general non-Gaussian setting. We reduce dimensionality by identifying an \emph{informed} subspace of the parameter and an \emph{informative} subspace of the data; this identification is performed prior to realizing the observations. 
These subspaces yield posterior approximations that depart from the prior only along a low-dimensional subspace, and via conditioning only on a low-dimensional subspace of the data. We obtain these subspaces by constructing and minimizing a tractable quadratic upper bound for the expected KL divergence from the posterior approximation to the true posterior distribution. This bound is derived from logarithmic Sobolev inequalities, and can be used to evaluate and compare the quality of any data- and parameter-space projectors. We also show how the same ideas can be used for coordinate selection in both the data and parameters, by minimizing the upper bound under more restricted conditions.

In the specific case of linear--Gaussian models, with the reduced data and parameter subspaces fixed to the same dimension, our construction recovers canonical correlation analysis (CCA). In general, however, our approach yields more accurate posterior approximations than those obtained with both CCA and PCA, as well as interpretable projections of the parameters and data, for a range of non-Gaussian inference problems. Our formulation also generalizes and lends theoretical support to heuristics recently considered in specialized settings, for instance in sequential data assimilation \citep{provost2022low}. 
We outline some directions for future work below.

\textbf{Gradient-free methods}. The optimal {projections} %
of the variables are defined via the  eigendirections of two diagnostic matrices $H_\X$ and $H_\Y$, containing mixed partial derivatives of the log-likelihood function. For certain applications, these derivatives might be unavailable or computationally expensive to evaluate. It would be interesting to develop estimators for $H_\X$ and $H_\Y$ based only on differences between forward model evaluations. Furthermore, understanding the sample complexity of these estimators---specifically, the estimators of the leading eigenspaces of $H_{\X}$ and $H_{\Y}$---will be useful to determine the number of samples required to reliably achieve posterior approximations of a given error. (See~\citet{zahm2018certified, lam2020multifidelity} for analysis of subspace approximation errors in settings where gradients are available.) %

\textbf{Nonlinear dimension reduction}. For strongly nonlinear forward models, linear dimension reduction may require many modes to achieve a small posterior approximation error. In these cases, one might instead seek few \emph{nonlinear} functions of the parameters that are informed by nonlinear features of the data.  See~\citet{bigoni2021nonlinear} for an approach to identifying these parameter features in the context of surrogate modeling. It will be interesting to extend our current guarantees on posterior approximation error to nonlinear features, and to compare the resulting features to those identified by nonlinear supervised dimension reduction methods such as~\citet{andrew2013deep, michaeli2016nonparametric}. %

\begin{acks}[Acknowledgments]
RB and YM gratefully acknowledge support from the United States Department of Energy AEOLUS center under award DE-SC0019303.
OZ gratefully acknowledges support from the ANR JCJC project MODENA (ANR-21-CE46-0006-01).
RB, YM, and OZ acknowledge support from the INRIA associate team Unquestionable.
The authors also thank Qiao Chen for finding a mistake in the initial version of this manuscript.

\end{acks}

\appendix

\section{Proof of Propositions} \label{app:proofs} 
\begin{proof}[\textbf{Proof of Proposition~\ref{prop:optimalPosteriorApprox}}]
Let $\widetilde\pi_{\X|\Y}$ be any approximate posterior density of the form $\widetilde\pi_{\X|\Y}(\x|\y)=f_1(\y_s,\x_r)f_2(\x_\perp,\x_r)$. 
Let $f_0(\x_r)=\int f_2(\x_\perp,\x_r) \d \x_\perp$ and
\begin{align*}
    \overline f_1(\y_s,\x_r) &= f_1(\y_s,\x_r)  f_0(\x_r) \\
    \overline f_2(\x_\perp,\x_r) &= f_2(\x_\perp,\x_r) / f_0(\x_r),
\end{align*}
so that $\int \overline f_1(\y_s,\x_r) \d\x_r= \int \widetilde\pi_{\X|\Y}(\x)\d \x = 1$ for all $\y_s$ and $\int \overline f_2(\x_\perp,\x_r) \d\x_\perp=1$ for all $\x_r$.
Then $\x_r\mapsto \overline f_1(\y_s,\x_r)$ and $\x_\perp\mapsto \overline f_2(\x_\perp,\x_r)$ can be interpreted as conditional densities for all $\y_s$ and $\x_r$, respectively. 
From the definition of the KL divergence, we have
\begin{align}
&\quad\; \mathbb{E}_{\Y}\left[\KLDiv(\Post(\cdot|\Y)||\ApproxPost(\cdot|\Y))\right] - \mathbb{E}_{\Y}\left[\KLDiv(\Post(\cdot|\Y)||\OptPost(\cdot|\Y)) \right] \nonumber \\
&= \mathbb{E}_{\X,\Y} [\log \OptPost(\X|\Y) - \log \ApproxPost(\X|\Y)] \nonumber \\
&= \mathbb{E}_{\X,\Y} [\log \pi_{\X_r|\Y_s}(\X_r|\Y_s)\pi_{\X_\perp|\X_r}(\X_\perp|\X_r) - \log \overline f_1(\Y_s,\X_r)\overline f_2(\X_\perp,\X_r)] \nonumber \\
&= \mathbb{E}_{\X,\Y} \left[\log \frac{\pi_{\X_r|\Y_s}(\X_r|\Y_s)}{\overline f_1(\Y_s,\X_r)} + \log \frac{\pi_{\X_\perp|\X_r}(\X_\perp|\X_r)}{\overline f_2(\X_\perp,\X_r)} \right] \nonumber \\
&= \mathbb{E}_{\X_r,\Y_s} \left[\log \frac{\pi_{\X_r|\Y_s}(\X_r|\Y_s)}{\overline f_1(\Y_s,\X_r)} \right] + \mathbb{E}_{\X_\perp,\X_r} \left[\log \frac{\pi_{\X_\perp|\X_r}(\X_\perp|\X_r)}{\overline f_2(\X_\perp,\X_r)} \right] \nonumber \\
&= \mathbb{E}_{\Y_s} \left[\KLDiv(\pi_{\X_r|\Y_s}(\cdot|\Y_s)||\overline f_1(\Y_s,\cdot))\right] + \mathbb{E}_{\X_r}\left[ \KLDiv(\pi_{\X_\perp|\X_r}(\cdot|\X_r)||\overline f_2(\cdot,\X_r))\right] \label{eq:KLdiff_optpost}
\end{align}
By the positivity of the two KL divergence terms in~\eqref{eq:KLdiff_optpost}, we have the result in~\eqref{eq:KLDiv_OptPost}. %
\end{proof}

\begin{proof}[\textbf{Proof of Proposition~\ref{prop:PosteriorErr_MI}}] %
Let $\Joint$ be a joint density of $(\X,\Y)$ with posterior density $$\Post(\x|\y) = \frac{\Like(\y|\x)\Prior(\x)}{\DataMarg(\y)}.$$ %
Let the density for the optimal posterior approximation be 
$$\OptPost(\x|\y) = \frac{\pi_{\Y_s|\X_r}(\y_s|\x_r)\pi_{\X}(\x)}{\ApproxDataMarg(\y_s)},$$ %
where $\pi_{\Y_s|\X_r}(\y_s|\x_r) = \int \pi_{\Y|\X}(\y_s,y_\perp|\x_r,\x_\perp) \pi_{\X_\perp|\X_r}(\x_\perp|\x_r) \d\x_{\perp} \d\y_{\perp}$ is the approximate likelihood function.
For this likelihood, the approximate data marginal satisfies 
\begin{align*}
    \ApproxDataMarg(\y_s) &= \int \pi_{\Y_s|\X_r}(\y_s|\x_r)\pi_{\X_r}(\x_r)\textrm{d}\x_r \\
    &= \int \int \pi_{\Y|\X}(\y|\x) \pi_{\X}(\x) \textrm{d}\y_\perp \textrm{d}\x
    = \int \pi_{\Y_s|\X}(\y_s|\x) \pi_{\X}(\x) \textrm{d}\x = \pi_{\Y_s}(\y_s).
\end{align*}
The KL divergence from the optimal posterior approximation to the true posterior in expectation over the data is then given by
\begin{align}
&\quad\,\, \mathbb{E}_{\Y} \left[\KLDiv(\Post(\cdot|\Y)||\OptPost(\cdot|\Y)) \right] \nonumber \\
&= \int \pi_{\X,\Y}(\x,\y) \log \frac{\Post(\x|\y)}{\OptPost(\x|\y)} \d\x\d\y \nonumber \\
&= \int \pi_{\X,\Y}(\x,\y) \log \frac{\pi_{\Y|\X}(\y|\x)\pi_{\X}(\x)/\pi_{\Y}(\y)}{\pi_{\Y_s|\X_r}(\y_s|\x_r)\pi_{\X}(\x)/\pi_{\Y_s}(\y_s)} \d\x\d\y \nonumber \\
&= \int \pi_{\X,\Y}(\x,\y) \log \frac{\pi_{\Y|\X}(\y|\x)}{\pi_{\Y}(\y)} \d\x\d\y - \int \pi_{\X,\Y}(\x,\y) \log \frac{\pi_{\Y_s|\X_r}(\y_s|\x_r)}{\pi_{\Y_s}(\y_s)} \d\y \nonumber \\
&= \underbrace{\int \pi_{\X,\Y}(\x,\y) \log \frac{\pi_{\Y|\X}(\y|\x)}{\pi_{\Y}(\y)} \d\x\d\y}_{I(\X;\Y)} - \underbrace{\int \pi_{\X_r,\Y_s}(\x_r,\y_s) \log \frac{\pi_{\Y_s|\X_r}(\y_s|\x_r)}{\pi_{\Y_s}(\y_s)} \d\x_r\d\y_s}_{I(\X_r;\Y_s)}. \label{eq:KLequiv_MI_diff} %
\end{align}
Lastly, from the chain rule for mutual information we have
\begin{align*}
    I(\X;\Y) - I(\X_r,\Y_s) = I(\X_\perp;\Y|\X_r) + I(\X,\Y_\perp|\Y_s) - I(\X_\perp;\Y_\perp|\X_r,\Y_s). %
\end{align*}
\end{proof}

\begin{proof}[\textbf{Proof of Proposition~\ref{prop:CMI_bound}}]
The proof is done in two steps. First we assume that $\Z = \emptyset$ so that $I(\X;\Y|\Z)=I(\X;\Y)$ becomes the mutual information.
We can write
\begin{align*}
I(X;Y)
&= \int h(\x,\y)\log\left(h(\x,\y)\right) \, \pi_\X(\x)\pi_\Y(\y) \d \x \d \y.
\end{align*}
where $h(x,y) = \pi_{X,Y}(x,y) / (\pi_X(x)\pi_Y(y) )$.
Because $\pi_{\X,\Y}$ satisfies the logarithmic Sobolev inequality with constant bounded by $\overline{C}(\Joint)$, we have that the product density $\pi_\X\otimes\pi_\Y$ also satisfies the logarithmic Sobolev inequality with constant bounded by $\overline{C}(\Joint)$; see~\citet[Theorem 4.4]{guionnet2003lectures}. Then, since $\int h(x,y) \pi_Y(y)\pi_X(x)   \d y\d x = \int \pi_{X,Y}(x,y) \d y\d x = 1$, the logarithmic Sobolev inequality gives
\begin{align}
I(\X;\Y)
&\leq \frac{\overline{C}(\Joint)}{2} \int  \|\nabla\log h(\x,\y)\|^2 \, h(\x,\y) \, \pi_\X(\x)\pi_\Y(\y)   \d \x \d \y \nonumber\\ 
&= \frac{\overline{C}(\Joint)}{2} \int \left\| \begin{pmatrix} \nabla_\X \log h(\x,\y) \\ \nabla_\Y \log h(\x,\y) \end{pmatrix} \right\|_2^2 \pi_{\X,\Y}(\x,\y)   \d \x \d \y \nonumber\\
&= \frac{\overline{C}(\Joint)}{2} \int \left(\int \left\| F_1(\x,\y) - \nabla_\X \log\pi_\X(\x) \right\|_2^2 \pi_{\Y|\X}(\y|\x) \d \y \right) \pi_X(x) \d \x  \nonumber \\
&\,+ \frac{\overline{C}(\Joint)}{2} \int\left(\int \left\| F_2(\x,\y) - \nabla_\Y \log\pi_\Y(\y) \right\|_2^2 \pi_{\X|\Y}(\x|\y) \d \x \right)\pi_Y(y)\d \y ,\label{eq:tmp23904793}
\end{align}
where $F_1(x,y) = \nabla_X \log \pi_{X,Y}(x,y)$ and $F_2(x,y) = \nabla_Y \log \pi_{X,Y}(x,y)$.
Next we are going to apply the Poincar\'{e} inequality to bound the two terms in the last expression.
Recall that the logarithmic Sobolev inequality~\eqref{eq:LSI} implies the Poincar\'{e} inequality (see~\citet[Proposition 5.1.3]{bakry2014analysis}) so we have 
\begin{equation}\label{eq:subspace_poincare}
\int \| f - \int f \d\rho \|_2^2 \,\d\rho \leq \overline{C}(\Joint) \int \|\nabla f\|_F^2 \, \d \rho ,
\end{equation}
for any smooth vector-valued function $f$, where $\rho$ can be either $\pi_{X|Y}(\cdot|y)$ or $\pi_{Y|X}(\cdot|x)$. Here, $\|\cdot\|_F$ denotes the Frobenius norm.
Notice for any $y\in\R^k$ we can write
\begin{align*}
\int F_1(x,y) \pi_{Y|X}(y|x) \d y
&= \int \nabla_X \log \pi_{X,Y}(x,y) \frac{\pi_{X,Y}(x,y)}{\pi_{X}(x)} \d y\\
&= \frac{\int \nabla_X \pi_{X,Y}(x,y) \d y }{\pi_{X}(x)} = \frac{\nabla_X \pi_{X}(x)}{\pi_{X}(x)} = \nabla_X \log \pi_{X}(x),
\end{align*}
which means that $ \nabla_X \log \pi_{X}(x)$ is the conditional expectation of $F_1(X,Y)$ conditioned on $X=x$.
In the same way, $\nabla_Y \log \pi_{Y}(y)$ is the conditional expectation of $F_2(X,Y)$ conditioned on $Y=y$.
This observations permits us to apply the Poincaré inequality \eqref{eq:subspace_poincare} to the two terms in \eqref{eq:tmp23904793}. This gives
\begin{align*}
I(X;Y)
&\leq \frac{\overline{C}(\Joint)}{2} \int \left(\overline{C}(\Joint)\int \left\| \nabla_Y F_1(\x,\y) \right\|_F^2 \pi_{\Y|\X}(\y|\x) \d \y \right) \pi_X(x) \d \x  \nonumber \\
&\,+ \frac{\overline{C}(\Joint)}{2} \int\left(\overline{C}(\Joint)\int \left\| \nabla_X F_2(\x,\y) \right\|_F^2 \pi_{\X|\Y}(\x|\y) \d \x \right)\pi_Y(y)\d \y \\
&= \frac{\overline{C}(\Joint)^2}{2} \int \left(\| \nabla_Y F_1(\x,\y) \|_F^2 + \| \nabla_X F_2(\x,\y) \|_F^2\right)\pi_{\X,\Y}(\x,\y) \d \y  \d \x . \nonumber 
\end{align*}
Because $(\nabla_Y F_1(\x,\y))^T = \nabla_X \nabla_Y \log \pi_{X,Y}(x,y) = \nabla_X F_2(\x,\y)$, we obtain
\begin{align}\label{eq:tmp2689}
I(X;Y) \leq \overline{C}(\Joint)^2 \int   \| \nabla_X \nabla_Y \log \pi_{X,Y}(x,y) \|_F^2 \pi_{\X,\Y}(\x,\y) \d \y  \d \x  ,
\end{align}
which is~\eqref{eq:CMI_UpperBound} when $Z=\emptyset$.

Now we assume that $Z\neq\emptyset$. Applying the previous result to the conditional density $\pi_{X,Y|Z}(x,y|z)$ yields
\begin{align*}
 \int &\pi_{\X,\Y|Z}(x,y|z)\log\left( \frac{\pi_{\X,\Y|Z}(x,y|z)}{\pi_{\X|Z}(\x|z)\pi_{\Y|Z}(\y|z)}\right) \d \x \d \y \\
 &\overset{\eqref{eq:tmp2689}}{\leq} \overline{C}(\pi_{X,Y|Z=z})^2 \int   \| \nabla_X \nabla_Y \log \pi_{X,Y|Z}(x,y|z) \|_F^2 \pi_{\X,\Y|Z}(\x,\y|z) \d \y  \d \x   \\
 &\;\leq\; \overline{C}(\pi_{X,Y,Z})^2 \int   \| \nabla_X \nabla_Y \log \pi_{X,Y,Z}(x,y,z) \|_F^2 \pi_{\X,\Y|Z}(\x,\y|z) \d \y  \d \x  ,
\end{align*}
where we used the fact that $\overline{C}(\pi_{X,Y|Z=z})\leq \overline{C}(\pi_{X,Y,Z})$ for all $z$; recall Equation~\eqref{eq:condLSI}.
Multiplying with the marginal $\pi_Z(z)$ and integrating over $z$ yields \eqref{eq:CMI_UpperBound} and concludes the proof.
\end{proof}

\begin{proof}[\textbf{Proof of Proposition~\ref{prop:CCA_connection}}]
For the linear model $\Y = G\X + \varepsilon$, the covariance of $\Y$ is $\Cov(\Y) = G\Cpr G^{T} + \Cobs$, and the cross-covariance is $\Cov(\X,\Y) = \Cpr G^T$. Hence, the eigenvalue problems in~\eqref{eq:CCAeigU} and~\eqref{eq:CCAeigV} can be written as
\begin{align}
    \Cpr G^T (G\Cpr G^T + \Cobs)^{-1} G \Cpr u_i^\CCA &= \rho_i \Cpr u_i^\CCA \label{eq:CCAeigU_Linear} \\
    G\Cpr G^T v_i^\CCA &= \rho_i (G\Cpr G^T + \Cobs) v_i^\CCA. \label{eq:CCAeigV_Linear} 
\end{align}
Using the Sherman-Morrison-Woodbury formula, we have
the matrix identity
$$G^T (G\Cpr G^T + \Cobs)^{-1} G = \Cpr^{-1}(\Cpr^{-1} + G^T\Cobs^{-1}G)^{-1} G^T\Cobs^{-1}G.$$
Applying this identity to the left hand side of~\eqref{eq:CCAeigU_Linear}, the eigenvalue problems %
in CCA are also given by
\begin{align*}
G^T\Cobs^{-1}G \Cpr u_i^\CCA &= \frac{\rho_i}{1-\rho_i} u_i^\CCA, \\
G\Cpr G^T v_i^\CCA &= \frac{\rho_i}{1-\rho_i} \Cobs v_i^\CCA.
\end{align*}
The vectors $u_i^\CCA,v_i^\CCA$  are also eigenvectors of~\eqref{eq:Gaussian_eig_param} and~\eqref{eq:Gaussian_eig_data}. Furthermore, given that $\rho_i/(1-\rho_i)$ is a monotonic function of $\rho_i \in [-1,1]$, the eigenvectors are ordered in the same way as the solutions to~\eqref{eq:CCAeigU_Linear} and~\eqref{eq:CCAeigV_Linear}. %
\end{proof}

\section{Additional calculations} \label{app:calculations}
\noindent \textbf{Gaussian subspace logarithmic Sobolev constant~\eqref{eq:MaxEig_Gaussian}}. 
We begin by computing the eigenvalues of the joint covariance matrix $\Cov(\X,\Y)$. The eigenvalues are given by the $d+m$ roots $s$ of the equation
$\det(\Cov(\X,\Y) - s\Id_{d+m}) = 0$. Without loss of generality, we let $d \geq m$.

From the matrix determinant lemma and the SVD of the forward model $G^T = U\Sigma V^T$ where $U \in \R^{d \times d}$, $V \in \R^{m \times m}$ are unitary matrices and $\Sigma \in \R^{d \times m}$ is a diagonal matrix containing zeros below row $m$, we have
\begin{align*}
\det(\Cov(\X,\Y) - s\Id_{d+m}) &= \det((1-s)\Id_d)\det(GG^T + \Id_m - s\Id_m - G(\Id_d - s\Id_d)^{-1}G^T). \\
&= \det((1-s)\Id_d)\det(V(\Sigma^2 +(1-s)\Id_m - \Sigma(1-s)^{-1}\Sigma)V^T) \\
&= \prod_{i=1}^{d} (1-s) \prod_{j=1}^{m} (\sigma_j^2(1 - (1-s)^{-1}) + (1-s)). \\
&= \prod_{i=m+1}^{d} (1-s) \prod_{j=1}^{m} ((1-s)^2 - \sigma_j^2 s).
\end{align*}
Then, there are $d-m$ roots $s = 1$ and $2m$ roots %
\begin{equation} \label{eq:roots_quadratic}
s = \frac{1}{2}\left(2 + \sigma_j^2 \pm \sigma_j \sqrt{\sigma_j^2 + 4} \right),
\end{equation}
from solving the quadratic equations $s^2 + (-2-\sigma_j)s + 1 = 0$ for $j = 1,\dots,m$. %
Given that $\sigma_j > 0$, the largest eigenvalues correspond to the roots with a positive sign. 
From~\eqref{eq:roots_quadratic} with a positive sign being a monotonic functions of $\sigma_j$, it follows that the maximum eigenvalue is given  by~\eqref{eq:MaxEig_Gaussian}.\\

\noindent \textbf{Expected KL divergence for a linear-Gaussian error model}~\eqref{eq:MIGaussian}. 
The difference of mutual information terms for Gaussian variables  %
is given by
\begin{equation} \label{eq:DiffMI_Gaussian}
I(\X,\Y) - I(\X_r,\Y_s) = \frac{1}{2} \log \frac{|\Gamma_{\X}|}{|\Gamma_{\X|\Y}|} - \frac{1}{2} \log \frac{|\Gamma_{\X_r}|}{|\Gamma_{\X_r|\Y_s}|}, %
\end{equation}
where $\Gamma$ represents a covariance. After the whitening transformations, we have the linear model $\Y = G\X + \varepsilon$ with $G \leftarrow \Cobs^{-1/2}G\Cpr^{1/2}$, $\Cov(\X) = \Id_d$ and $\Cov(\varepsilon) = \Id_m$. Then, the (conditional) covariances in~\eqref{eq:DiffMI_Gaussian} are given by
\begin{align*}
\Gamma_{\X} &= \Id_d \\
\Gamma_{\X_r} &= \BasisX_r^T \Id_d \BasisX_r \\ %
\Gamma_{\X|\Y} &= \Id_d - G^T(GG^T + \Id_m)^{-1}G \\
\Gamma_{\X_r|\Y_s} &= \BasisX_r^T \Id_d \BasisX_r - \BasisX_r^T G^T \BasisY_s(\BasisY_s^T G G^T \BasisY_s + \BasisY_s^T\BasisY_s)^{-1}\BasisY_s^T G \BasisX_r.
\end{align*}
Given that the eigenvectors of $H_\X$ and $H_\Y$ are the left and right singular vectors $U$ and $V$ for $G^T$, respectively, we have $\BasisX_r = [u_1,\dots,u_r]$ and $\BasisY_s = [v_1,\dots,v_s]$. %
Then, the conditional covariances can be simplified as
\begin{align*}
    \Gamma_{\X|\Y} &= \Id_d - U\Sigma^2(\Sigma^2 + \Id_q)^{-1}U^T \\ %
    \Gamma_{\X_r|\Y_s} &= \Id_r - U_{t}\Sigma_{t}^2(\Sigma_{t}^2 + \Id_t)^{-1}U_t^T,
\end{align*}
where $q = \min\{d,m\}$, $t = \min\{r,s\}$ and $\Sigma_t \in \R^{t \times t}$ is a diagonal matrix containing the first $t$ singular values of $G^T$. %
By computing the determinants of the (conditional) covariances, we have
\begin{align*}
    I(\X,\Y) - I(\X_r,\Y_s) &= -\frac{1}{2} \log\left|\Id_d - \Sigma^2(\Sigma^2 + \Id_q)^{-1}\right| + \frac{1}{2}\log\left|\Id_r - \Sigma_t^2(\Sigma_t^2 + \Id_t)^{-1}\right| \\
    &= \frac{1}{2}\sum_{i>t}^{q}\log(1 + \sigma_i^2).
\end{align*}

\bibliographystyle{ba}
\bibliography{references}

\end{document}